\documentclass{article}
\usepackage[utf8]{inputenc}
\usepackage{geometry}
\usepackage{graphicx}
\usepackage{longtable}
\usepackage{multicol}
\usepackage{indentfirst,csquotes}
\usepackage{upgreek}
\usepackage{hhline}
\title{Polynomially Superintegrable Hamiltonians Separating in Cartesian Coordinates}
\author{Ian Marquette$^{1}$ and Anthony Parr$^{2}$}

\date{$^{1}$ Department of Mathematical and Physical Sciences, La Trobe University \\ Bendigo 3552, Victoria Australia \\ $^{2}$ School of Mathematics and Physics, The University of Queensland \\ Brisbane 4072, Queensland, Australia}
\usepackage{amsmath,amssymb,amsfonts,amsthm,appendix}
\usepackage{mathtools}
\usepackage{xcolor}

\usepackage{enumitem}
\newcommand{\mrm}[1]{\mathrm{#1}}

\newcommand{\mcl}[1]{\mathcal{#1}}
\newcommand{\mbf}[1]{\mathbf{#1}}

\newcommand{\pd}[2]{\frac{\partial #1}{\partial #2}}

\DeclareMathOperator{\ad}{ad}
\DeclareMathOperator*{\Res}{Res}
\theoremstyle{theorem}

\newtheorem{theorem}{Theorem}

\newtheorem{problem}{Problem}

\newtheorem{lemma}{Lemma}
\newcounter{ex}
\newtheorem*{corollary*}{Corollary}
\theoremstyle{definition}

\newtheorem*{definition*}{Definition}
\newtheorem{example}[ex]{Example}
\newtheorem*{example*}{Example}
\usepackage{graphicx}
\usepackage{paracol,bm}

\usepackage{mathrsfs}
\usepackage[maxbibnames=99,backend=bibtex,sorting=none]{biblatex}
\bibliography{References.bib}
\begin{document}
\maketitle
\noindent
\textsc{pacs} numbers: 03.65.Fd, 03.65.Ge
\\
\textsc{email}: \texttt{i.marquette@latrobe.edu.au}, \texttt{anthony.parr@uq.net.au}
\bigskip
\noindent

\begin{abstract}
The problem of finding superintegrable Hamiltonians and their integrals of motion can be reduced to solving a series of compatibility equations that result from the overdetermination of the commutator or Poisson bracket relations. The computation of the compatibility equations requires a general formula for the coefficients, which in turn must depend on the potential to be solved for. This is in general a nonlinear problem and quite difficult. Thus, research has focused on dividing the classes of potential into standard and exotic ones so that a number of parameters may be set to zero and the coefficients may be obtained in a simpler setting. We have developed a new method in both the classical and quantum settings that allows a formula for the coefficients of the integral to be obtained without recourse to this division for Cartesian-separable Hamiltonians. The expressions we obtain are in general non-polynomial in the momenta whose fractional terms can be arbitrarily set to zero. These conditions are equivalent to the compatibility equations, but the only unknowns in addition to the potential are constant parameters. We also classify all the fourth-order standard superintegrable Hamiltonians.
\end{abstract}
\tableofcontents

\section{Introduction}

Scalar potentials whose corresponding Hamiltonians admit more integrals than degrees of freedom have been partitioned into two classes: those which satisfy linear differential equations, called `standard', and those which do not, called `exotic'. The harmonic oscillator and the Kepler model are well-known members of the first class, and in more recent times, the Smorodinsky-Winternitz, Tremblay-Turbiner-Winternitz (\textsc{ttw}) and Post-Winternitz models have been added to this list \cite{smod65,ttw09,post15}. It is known that systems with third-order integrals lead to standard potentials with wave functions involving exceptional orthogonal polynomials \cite{mar13,mar22}. The exotic class includes potentials that are algebraic in the classical case or satisfy the Painlev\'{e} property in the quantum case.

It is difficult, however, to study even separable Hamiltonians with higher order symmetries because of the increasing non-linearity of the resulting \textsc{pde} system. Thus, research has been pressured towards cases where many of the coefficients can be set to zero \cite{abou}. It is our aim in this work to present a general method valid for Cartesian-separable higher-order superintegrable Hamiltonians by which these additional assumptions may be dismissed and we can obtain a system of \textsc{ode} leading to the complete determination of the potential energy. 

We consider a natural Hamiltonian, either classical or quantum, on a flat configuration space with \(d\) degrees of freedom that is separable in Cartesian coordinates, i.e. of the form
\begin{equation}
H=\sum^d_{i=1}\left[\tfrac{1}{2}p_i^2+V_i(q_i)\right]\label{eq:ham}
\end{equation}
where \(p_1,p_2,\ldots,p_d\), \(q_1,q_2,\ldots,q_d\) are the canonical momenta and coordinates respectively and \(V(\mbf{q})=\sum_{i=1}^dV_i(q_i)\) is the potential energy, \(V_i\) we call a potential summand. This Hamiltonian possesses \(d\) second-order integrals
\begin{equation}
H_i=\tfrac{1}{2}p_i^2+V_i,\qquad i=1,2,\ldots,d.
\end{equation}
For two dimensions, such a Hamiltonian is one of four types that is second-order integrable, the others being separable in polar, parabolic or elliptical coordinates. For \(d>2\), separability is a special case of second-order integrability \cite{mak67,ev90,snobl}. We wish to address the case where the Hamiltonian is (minimally and polynomially) superintegrable, i.e. it possesses at least one integral \(X\) which is of \(n\)th order in the momenta that is not a polynomial in the known integrals \(H_1,H_2,\ldots,H_d\).
\begin{problem}
    Given a Hamiltonian of the form \eqref{eq:ham}, to find which assignments of the potential summands \(V_1,V_2,\ldots,V_d\) such that \(H\) is superintegrable.
\end{problem}
Much progress has been made for systems of two degrees of freedom. When the third integral is of \(n\)th order the system is called \(n\)th-order superintegrable. Second-order superintegrable Hamiltonians in Euclidean space are those that are separable in more than one coordinate system. The class of potentials, which is the same in the classical and quantum cases, was found by Fri\v{s} \textit{et al.} \cite{smod65}. Third-order superintegrable models have also been systematically classified \cite{gra02,gra04,tre10,popper} except those potentials which are only separable in elliptical coordinates. Fourth-order superintegrable systems have been studied in Cartesian and polar coordinates \cite{mar17,esc17,esc18a}. The Cartesian case is incomplete, and standard potentials and exotic classical potentials were left for a subsequent paper, which never came. Extensive study \cite{esc18b} on higher-order systems separable in polar coordinates has lead to the conclusion that the radial component of the potential must be of oscillator \(\alpha r^2\) or Kepler-Coulomb type \(\alpha /r\) (the two-parameter form \((\alpha+\beta/r^2)^{1/2}\) discovered by Onofri and Pauri \cite{onofri} has not been detected yet in polynomial superintegrability). Higher-order superintegrable systems separable in Cartesian coordinates have only been studied in the doubly-exotic case \cite{esc20}. The doubly-exotic fifth-order Hamiltonians have been completely classified \cite{abou,ismet}. Ladder operator methods \cite{kalnins2} and Darboux transformations \cite{mar13b,mar16} are invoked for integrals of arbitrarily high order such as with the \textsc{ttw} model \cite{ttw09}.

The simplicity of Cartesian separability greatly facilitates the search for superintegrable systems. Indeed, for classical systems, the use of action-angle variables has been very fruitful in making general statements and deriving diverse families of superintegrable models \cite{onofri,grig18,miller04}. Such methods however fail for quantum systems because the integrals are usually algebraic in the momenta.

The direct method for systematically finding superintegrable models is to expand the integral into its coefficients and solve the partial differential equations that result from requiring it to be in involution (in the classical case) or commute (in the quantum case) with the Hamiltonian. These \textsc{pde} form an overdetermined system. To ensure compatibility, there is an additional series of equations that the potential and the coefficients of the integral must satisfy. Since the leading-order terms of any integral must be a polynomial combination of the linear and angular momenta, the only unknowns of the first compatibility equation (which is linear in the potential) are constant parameters. However, the other compatibility equations involve functional unknowns which must be determined by integrating the potential and the previous set of coefficients of the integral. A variation of this approach is to expand in terms of powers of the known integrals \cite{kalnins1,bern20a}.

The essence of our approach is to treat the homogeneous components rather than the individual coefficients as the fundamental units of the problem. The system of \textsc{pde} relates one order to all the higher orders and this can be consolidated into a very simple form. We have developed an algorithm that can solve each equation in succession. This gives us a formula for the integral as a rational function of the form
\[X=\text{polynomial in the momenta }+\text{ reciprocal powers in the momenta.}\]
The coefficients of the reciprocal powers need to be equal to zero in order for it to be a polynomial integral. Out of these constraints arise a series of equations for the potential summands which are equivalent to the compatibility equations obtained from the direct method. This technique requires we treat the momenta as commuting variables rather than as differential operators. We therefore recast quantum mechanics as a deformation of classical mechanics by assigning Planck's reduced constant \(\hbar\) as the deformation parameter. This allows for a unified treatment of classical and quantum regimes via McCoy's formulas \cite{mccoy}.

In section 2, we derive the \textsc{pde} for the integral components as the interpretation of Problem 1. In section 3, we establish the framework which will enable us to solve these equations. The things to be determined are a sequence of operators which can be solved algorithmically. In section 4, we demonstrate this algorithm for a few cases (sufficient to calculate any integral and the determining equations for the potential up to tenth order). Section 5 will deal with the general construction of the determining equations. Section 6 recapitulates the process. In section 7, we restrict ourselves to two dimensions and give the linear, quadratic and cubic compatibility equations explicitly, valid for all orders. As an application we give all the fourth-order standard potentials, with a new family discovered.
\section{Description of the Problem}
Identifying \(\hbar\) with the deformation parameter means that at \(\hbar=0\), we should obtain the equations for classical motion.
\subsection{Discussion of the Classical Case}
Let us for the moment restrict ourselves to classical models where the system of equations related to this problem can easily be written down. To be in involution with the Hamiltonian \eqref{eq:ham} is equivalent to the following condition:
\begin{equation}
    0=\sum^d_{i=1}\left(p_i\pd{X}{q_i}-V_i'\pd{X}{p_i}\right).\label{eq:classical}
\end{equation}
When \(d=1\), the only solution to \eqref{eq:classical} is an arbitrary function of \(H\). Let us then take \(d>1\). By considering polynomial integrals, we eliminate the momenta as variables. Let us write
\begin{equation}
X\coloneqq\sum_{i_1+i_2+\cdots+i_d\leq n}\gamma_{\mbf{i}}(\mbf{q})\prod^d_{j=1}p_j^{i_j},\label{eq:ikcomp}
\end{equation}
where \(\mathbf{i}=(i_1,i_2,\ldots,i_d)\) is a tuple of non-negative integers. Then the single equation \eqref{eq:classical} becomes a system of equations
\begin{equation}
    0=\sum^d_{j=1}\left[\pd{\gamma_{\mbf{i}-\mbf{e}_j}}{q_j}-(i_j+1)V'_j\gamma_{\mbf{i}+\mbf{e}_j}\right].\label{eq:classicaleq}
\end{equation}
where \(\mbf{e}_j\) is an elementary unit vector
It is clear we may assume \(X\) is even or odd under time reversal, i.e. \(\gamma_{\mbf{i}}=0\) if \(n-i_1-i_2-\cdots-i_d\) is odd. For odd \(n\), there is an additional condition on the linear terms
\[0=\sum^d_{j=1}V'_j\gamma_{\mbf{e}_j}.\]
Equation \eqref{eq:classicaleq} overdetermines the coefficients on the left. This results in additional compatibility equations. In two dimensions, we may express these explicitly:
\begin{equation}
0=\sum^{n-k+1}_{i=0}(-1)^i\frac{\partial^{n-k+1}}{\partial q_1^{n-k+1-i}\partial q_2^i}\left[(i+1)V'_1\gamma_{i+1,n-k+1}+(n-k+2)V'_2\gamma_{i,n-k+2}\right].\label{eq:compatibility}
\end{equation}
Such compatibility equations are the conditions which limit the form of the potential energy. The assumption of separability of the potential reduces the compatibility equations in the most general case from a \textsc{pde} to a system of \textsc{ode} in the potential summands provided the coefficients are also determined. If we are to determine the potentials which lead to superintegrability we must have a constructive method for the integral. Equation \eqref{eq:classicaleq} allows us to determine one order in terms of another. We therefore write our integral as \(X=X_0+X_1+X_2+\cdots\) where \(X_i\) is a homogeneous polynomial in the momenta of order \(n-i\). Equation \eqref{eq:classical} becomes
\begin{equation}
    0=\sum^d_{j=1}\left(p_j\pd{X_k}{q_j}-V'_j\pd{X_{k-2}}{p_j}\right).\label{eq:newclass}
\end{equation}
A similar system of equations was derived by Holt \cite{holt} for his truncation program. Our restriction to Cartesian separability is enough that equation \eqref{eq:newclass} can be solved.
\subsection{Extension to the Quantum Case}
Before we proceed to our analysis, we shall generalise \eqref{eq:newclass} to the quantum case. We shall do this by defining a total symbol, also denoted \(X\). This allows us to speak of division and differentiation of the momenta. The usual procedure for this purpose is to define a Moyal bracket \cite{hiet83}, but this necessitates the Weyl ordering. To preserve the standard ordering (momenta on the right), we make a different choice of operations. Let \(f,g\) be two functions on phase space, polynomial in the momenta and smooth in the coordinates. We define
\begin{equation}
\begin{aligned}
    (f,g) &\coloneqq\sum_{i_1+i_2+\cdots+i_d\geq 1}\frac{1}{i_1!i_2!\cdots i_d!}\left(\frac{\hbar}{\sqrt{-1}}\right)^{i_1+i_2+\cdots+i_d-1}\\
    &\qquad\times\left(\frac{\partial^{i_1+i_2+\cdots+i_d}f}{\partial q_1^{i_1}\partial q_2^{i_2}\cdots\partial q_d^{i_d}}\frac{\partial^{i_1+i_2+\cdots+i_d}g}{\partial p_1^{i_1}\partial p_2^{i_2}\cdots\partial p_d^{i_d}}\right.\\
    &\qquad\left.-\frac{\partial^{i_1+i_2+\cdots+i_d}f}{\partial p_1^{i_1}\partial p_2^{i_2}\cdots\partial p_d^{i_d}}\frac{\partial^{i_1+i_2+\cdots+i_d}g}{\partial q_1^{i_1}\partial q_2^{i_2}\cdots\partial q_d^{i_d}}\right),\label{eq:poisson}
    \end{aligned}
    \end{equation}
    \begin{equation}
            f^*\coloneqq\sum_{i_1,i_2,\ldots,i_d\geq 0}\frac{1}{i_1!i_2!\cdots i_d!}\left(\frac{\hbar}{\sqrt{-1}}\right)^{i_1+i_2+\cdots+i_d}\frac{\partial^{2i_1+2i_2+\cdots+2i_d}\bar{f}}{\partial q_1^{i_1}\partial q_2^{i_2}\cdots \partial q_d^{i_d}\partial p_1^{i_1}\partial p_2^{i_2}\cdots\partial p_d^{i_d}}\label{eq:fdagger}
    \end{equation}
    where \(\bar{f}\) is the complex conjugate of \(f\). We can simplify \eqref{eq:fdagger} by introducing the operator
    \begin{equation}
    a_0\coloneqq\frac{\hbar}{2\sqrt{-1}}\sum^d_{i=1}\frac{\partial^2}{\partial q_i\partial p_i}.\label{eq:a0def}
\end{equation}
Then by the multinomial theorem, \eqref{eq:fdagger} can be written as
\begin{equation}
    f^*=\sum_{i\geq 0}\frac{2^i}{i!}a_0^i\bar{f}=\exp(2a_0)\bar{f}\label{eq:multinom}
\end{equation}
These expansions are well-defined provided \(f,g\) are polynomial in the momenta. We have
\begin{equation}
(\cdot,\cdot)=(\cdot,\cdot)_{PB}+\mathcal{O}(\hbar)
\end{equation}
where \((\cdot,\cdot)_{PB}\) is the Poisson bracket. Moreover, it follows from McCoy's formulas \cite{mccoy} that 
\begin{equation}
    f\circ g-g\circ f=\hbar\sqrt{-1}(f,g)
\end{equation}
where \(\circ\) represents non-commutative multiplication of observables. By a similar induction, \(f^*\) gives precisely the Hermitean conjugate. Throughout we shall use \((\cdot,\cdot)\) doubly as the Poisson bracket and normalised quantum bracket. The coordinates and their conjugate momenta are always understood to commute under \([\cdot,\cdot]\).
\begin{theorem}
    Problem 1 is equivalent to finding \(n+1\) functions \(X_0,X_1,\ldots,X_n\) on \(2d\)-dimensional phase space such that:
    \begin{enumerate}[label=\((\roman*)\)]
    \item \(X_k\) is a homogeneous polynomial of order \(n-k\) in the momenta for each \(k\).
    \item They constitute a non-trivial solution to the system of \textsc{pde}
\begin{equation}
\sum_{i=1}^d\left(p_i\pd{X_k}{q_i}+\frac{\hbar }{2\sqrt{-1}}\frac{\partial^2X_{k-1}}{\partial q_i^2}\right)=\sum^d_{i=1}\sum_{j=1}^{k-1}\frac{1}{j!}\left(\frac{\hbar}{\sqrt{-1}}\right)^{j-1}V_i^{(j)}\frac{\partial^jX_{k-1-j}}{\partial p_i^j}.\label{eq:quantum}
\end{equation}
    \end{enumerate}
\end{theorem}
\begin{proof}
   We need to find a function such that \((X,H)=0\). Using the definition \eqref{eq:poisson} of the deformed Poisson bracket applied to \eqref{eq:ham}, we get
\begin{equation}
    \sum^d_{i=1}\left[\pd{X}{q_i}p_i+\frac{\hbar}{2\sqrt{-1}}\frac{\partial^2X}{\partial q_i^2}-\sum_{j=1}^n\frac{1}{j!}\left(\frac{\hbar}{\sqrt{-1}}\right)^{j-1}V_i^{(j)}\frac{\partial^jX}{\partial p_i^j}\right]=0.\label{eq:comm}
\end{equation}
Expanding into homogeneous components, each term that appears is homogeneous. Therefore, all the terms of a particular order are independently zero. We see that \(p_i\displaystyle\pd{X_k}{q_i}\) is of order \(n-k+1\) while \(\displaystyle\frac{\partial^jX_k}{\partial p_i^j}\)
is of order \(n-k-j\). Equation \eqref{eq:quantum} gives all the terms that are of order \(n-k+1\).
\end{proof}
Theorem 2 in \cite{post15} states there is no loss of generality in taking our integral to be the particular Hermitean form \(X=\frac{1}{2}(Y+Y^*)\) where \(Y\) is real and even or odd in the quantum momenta according to the parity of \(n\). Writing \(Y=\sum_{i=0}^{\lfloor\frac{1}{2}n\rfloor}Y_i\) where \(Y_i\) is the homogeneous component in the momenta of order \(n-2i\), we get
        \begin{align}
X_{2i}&=Y_i+\sum^{i-1}_{j=0}\frac{2^{2i-2j-1}}{(2i-2j)!}a_0^{2i-2j}Y_j,\label{eq:even}\\
X_{2i+1}&=\sum^i_{j=0}\frac{2^{2i-2j}}{(2i-2j+1)!}a_0^{2i-2j+1}Y_j.\label{eq:odd}
\end{align}
There are now \(\lceil \frac{1}{2}n\rceil+1\) independent equations to solve. In the case \(\hbar=0\), \(a_0\) vanishes and \(X_{2i}=Y_i\) and \(X_{2i+1}=0\). This aligns with the classical case after assuming the integral was even or odd under time reversal.
\section{Iterative Integration of the Equations}
Let us define the partial differential operator
\begin{equation}
    L\coloneqq\sum^d_{i=1}p_i\pd{}{q_i}.
\end{equation}
A solution to equation \eqref{eq:quantum} requires the iterative inversion of this operator.
\subsection{The Leading Order Term}
When \(k=0\), equation \eqref{eq:quantum}, for a Hermitean integral, has
\[LY_0=0\]
The kernel of \(L\) coincides with the integrals of a free particle Hamiltonian, or in geometric language the Killing tensors of the Euclidean metric. An integral of arbitrary order can be expressed as a polynomial in the first-order integrals. These are the linear momenta \(p_1,p_2,\ldots,p_d\) and the angular momenta \(m_{ij}\coloneqq q_ip_j-q_jp_i\). The expression is not in general unique thanks to the dependence relations
\begin{equation}
p_im_{jk}+p_jm_{ki}+p_km_{ij}=0.\label{eq:mdep}
\end{equation}
\subsection{Algorithm for Integration}
We shall endeavour to reduce \(Y\) to a series of `integration constants' \(Z_0,Z_1,\ldots\) that are functions of the linear and angular momenta. Our first assignment is \(Y_0=Z_0\) so \(Z_0\) must be a non-zero polynomial in \(p_i,m_{ij}\). For \(k=2\) in equation \eqref{eq:quantum}, we obtain another equation
\begin{equation}
LY_1=\sum^d_{i=1}\left(V_i'\pd{}{p_i}-\frac{\hbar}{2\sqrt{-1}}a_0\frac{\partial^2}{\partial q_i^2}\right)Z_0\label{eq:y1}
\end{equation}
The homogeneous solution is designated \(Z_1\). We now need to find the particular solution. Differential operators that are a linear combination of
\begin{equation}
F=\frac{f(p_i,q_i)}{k!}\frac{\partial^{k+\ell}}{\partial p_i^k\partial q_i^\ell}.\label{eq:mon}
\end{equation}
we will call \textit{additively separable}. The set of additively separable operators forms a Lie algebra under the commutator.
\begin{theorem}
    For any additively separable operator \(F\), there exists another additively separably operator \(G\) such that \([L,G]=F\).
\end{theorem}
\begin{proof}
Given the linearity of solutions, we may without loss of generality assume \(F\) is in the simple form \eqref{eq:mon}. We observe that
\[\left[L,\frac{\partial^{k+\ell}}{\partial p_i^k\partial q_i^\ell}\right]=-k\frac{\partial^{k+\ell}}{\partial p_i^{k-1}\partial q_i^{\ell+1}}\]
Repeating the commutators, we have the nilpotency property
\begin{equation}
\ad_L^{k+1}\left(\frac{\partial^{k+\ell}}{\partial p_i^k\partial q_i^\ell}\right)=0.\label{eq:nil}
\end{equation}
We require an operator \(G\) such that \([L,G]=F\). Here, we use integration by parts with the nilpotency property \eqref{eq:nil} ensuring that the process will always terminate after a finite number of iterations. The solution is plainly
\begin{equation}
G=\frac{\int\! f\,\mathrm{d}q_i}{k!p_i}\frac{\partial^{k+\ell}}{\partial p_i^k\partial q_i^\ell}+\frac{\iint\! f\,\mathrm{d}^2q_i}{(k-1)!p_i^2}\frac{\partial^{k+\ell}}{\partial p_i^{k-1}\partial q_i^{\ell+1}}+\cdots+\frac{\int\cdots\int f\,\mathrm{d}^{k+1}q_i}{p_i^{k+1}}\frac{\partial^{k+\ell}}{\partial q_i^{\ell+1}}\label{eq:mon2}
\end{equation}
\end{proof}
As an application of \eqref{eq:mon2}, we look for operators \(a_j\) such that
\begin{align}
[L,a_j]&=\frac{1}{j!}\left(\frac{\hbar}{\sqrt{-1}}\right)^{j-1}\sum^d_{i=1}V^{(j)}_i\frac{\partial^j}{\partial p_i^j},\qquad j\geq 1.
\end{align}
We integrate to find
\begin{align}
    a_j&\coloneqq\left(\frac{\hbar}{\sqrt{-1}}\right)^{j-1}\sum^d_{i=1}\left(\frac{\int\!V_i}{p_i^{j+1}}\frac{\partial^j}{\partial q_i^j}+\sum^j_{k=1}\frac{1}{k!}\frac{V_i^{(k-1)}}{p_i^{j-k+1}}\frac{\partial^j}{\partial p_i^k\partial q_i^{j-k}}\right),\qquad j\geq 1.\label{eq:aoperator}
\end{align}
Equation \eqref{eq:quantum} may then be consolidated into the algebraic form
\begin{equation}
    LX_k=\sum^{k-1}_{i=0}[L,a_i]X_{k-1-i}.\label{eq:consol}
\end{equation}
using
\[[L,a_0]=-\frac{\hbar}{2\sqrt{-1}}\sum^d_{i=1}\frac{\partial^2}{\partial q_i^2}.\]
where \(a_0\) is given by \eqref{eq:a0def}.
\subsection{The Next-to-Leading Order Term}
With the new notation, equation \eqref{eq:y1} simplifies to the following
\[LY_1=([L,a_1]+[L,a_0]a_0-La_0^2)Z_0.\]
We are now ready to make our next integration
\begin{equation}
Y_1=Z_1+(a_1-\tfrac{1}{2}a_0^2)Z_0\label{eq:y1sol}
\end{equation}
\begin{example}
Let us consider the the Smorodinsky-Winternitz model
\begin{equation}H=\tfrac{1}{2}(p_1^2+p_2^2)+\beta_1(q_1^2+q_2^2)+\frac{\beta_2}{q_1^2}+\frac{\beta_3}{q_2^2}.\label{eq:smod}
\end{equation}
Now that the potential is given, we can calculate \eqref{eq:aoperator}
\begin{align*}
a_1&=\frac{\beta_1 q_1^4-3\beta_2}{3q_1p_1^2}\pd{}{q_1}+\frac{\beta_1q_1^4+\beta_2}{q_1^2p_1}\pd{}{p_1}+\frac{\beta_1q_2^4-3\beta_3}{3q_2p_2^2}\pd{}{q_2}+\frac{\beta_1q_2^4+\beta_3}{q_2^2p_2}\pd{}{p_2}
\end{align*}
The Hamiltonian \eqref{eq:smod} is second-order superintegrable. In order to determine the integral \textit{a priori} we need to find \(Z_0,Z_1\). This is the task of the compatibility equations, which we derive in Section 6. For the present moment, let us set \(Z_0=m_{12}^2\) which coincides with the leading-order term of the integral. Equation \eqref{eq:y1sol} gives us
\begin{align*}
Y_1&=Z_1-\tfrac{1}{2}\hbar^2+\frac{2\beta_1m_{12}^4}{3p_1^2p_2^2}+2(q_1^2p_2^2-q_2^2p_1^2)\left(\frac{\beta_3}{q_2^2p_2^2}-\frac{\beta_2}{q_1^2p_1^2}\right)
\end{align*}
We must choose \(Z_1\) so that the fractional terms in \(Y_1\) are annihilated. This is done if we take
\[Z_1=\frac{2\beta_2p_2^2}{p_1^2}+\frac{2\beta_1p_1^2}{p_2^2}-\frac{2\beta_1m_{12}^4}{3p_1^2p_2^2}.\]
Then our integral is
\[X=q_2^2p_1^2+q_1^2p_2^2-2q_1q_2p_1p_2+\hbar\sqrt{-1}(q_1p_1+q_2p_2)+\tfrac{1}{2}\hbar^2+\frac{2\beta_2q_2^2}{q_1^2}+\frac{2\beta_3q_1^2}{q_2^2}\]
which we calculated by using \(X=\tfrac{1}{2}(Y+Y^*)\) and \eqref{eq:multinom}.
\end{example}
We see then that \eqref{eq:y1sol} splits the polynomial components of the integral into rational parts. The integration constants \(Z_k\) are found by eliminating all the residues. For any other choice of \(Z_0\) in our example, besides \(p_1^2,p_2^2\), we would not have been able to find a polynomial form for \(Y_1\).
\subsection{General Solution as a Recurrence Relation}
To systematically derive each term successively, we consider the following recurrence relation
\begin{equation}
    [L,b_{i,k}]=\sum_{j=0}^ib_{i,k-1}[L,a_{i-j}],\qquad b_{i,0}=\delta_{i0}.\label{eq:bee2}
    \end{equation}
In Section 4, we shall solve \eqref{eq:bee2}. Supposing for now though that \eqref{eq:bee2} has a solution, the convolutive sum
\begin{equation}
    Z_k\coloneqq \sum_{i+j\leq 2k}(-1)^jb_{i,j}X_{2k-i-j}\label{eq:mop}
\end{equation}
satisfies \(LZ_k=0\). We then look for a reecursive inverse. If we substitute \eqref{eq:even} and \eqref{eq:odd}, we obtain
\begin{equation}
    Y_k=\sum^k_{i=0}W_iZ_{k-i}\label{eq:solution}
\end{equation}
where \(W_0=1\) and
\begin{equation}
W_k=\sum^{k-1}_{j=0}\left[\sum^{2k-2j}_{i=1}(-1)^{i+1}b_{2k-i-2j,i}+\sum^{2k-2j-1}_{i=0}\sum^{2k-i}_{\ell=2j+1}\frac{(-1)^{i+1}2^{\ell-2j-1}}{(\ell-2j)!}b_{2k-i-\ell,i}a_0^{\ell-2j}\right]W_j\label{eq:dequ}
\end{equation}
for \(k\geq 1\). It follows by this inversion that \(Z_0,Z_1,Z_2,\ldots\) form a sequence of independent functions of the linear and angular momenta which generate the integral \(X\). As it stands, these functions develop into an infinite sequence, and we must impose extra constraints so that \(Y_0,Y_1,Y_2,\ldots\) truncates.
\section{Determination of the Constituent Operators}
In this section, we present an algorithmic solution to \eqref{eq:bee2}. This is taken up to the point that is necessary to determine \(W_k\) for \(k\leq 5\). This is sufficient to determine a formula for an integral in terms of the potential function of order less than eleven. It also allows us to calculate the first five compatibility equations which define the potentials for integrals up to tenth order.

 Let \(s\) be an indeterminate and write \(A\coloneqq \sum_{i\geq 0} a_is^i\). Since \(A\) is additively separable, we take for granted the existence of a formal differential operator \(A_i\) that satisfies \([L,A_i]=\frac{i}{(i+1)!}(\ad A)^i[L,A]\) with \(A_0=0\). Let \(t\) be another indeterminate and write \(A'\coloneqq \sum_{j\geq 1}A_jt^j\). Repeating this process, we set \(A'_i\) to be a separable operator such that \([L,A'_i]=\frac{i}{(i+1)!}(\ad A')^i[L,A']\) with \(A'_0=0\). The individual coefficients
\begin{align*}
    a_{i,j}&\coloneqq\frac{1}{i!}\frac{\mathrm{d}^iA_j}{\mathrm{d}s^i}\bigg|_{s=0}\\
    a_{i,j,k}&\coloneqq\frac{1}{i!j!}\frac{\partial^{i+j}A'_k}{\partial s^i\partial t^j}\bigg|_{s,t=0}
\end{align*}
may, in principle, be calculated explicitly however a general formula has evaded us.

Let \(B\coloneqq \sum_{i,j\geq 0}b_{i,j}s^it^j\). Then equation \eqref{eq:bee2} becomes
\begin{equation}
[L,B]=tB[L,A],\qquad B(t=0)=1.\label{eq:beq}
\end{equation}
We shall show how \eqref{eq:beq} can be reduced to operators of the form \(a_i,a_{i,j},a_{i,j,k},\ldots\) which can be calculated \textit{ex post facto}.
\begin{lemma}
    For \(k\geq 1\), the following identity holds\label{lem:id}
    \begin{equation}
kA^{k-1}[L,A]=\left[L,A^k\right]+\sum_{i=0}^{k-1}\frac{k!}{(k-1-i)!}[L,A_i]A^{k-1-i}.\label{eq:aye}
\end{equation}
\label{lemma2}
\end{lemma}
\begin{proof}
For \(j\geq -1\), define
\[c_{j,k,\ell}\coloneqq\sum_{m=0}^{\ell} (-1)^m(\ell-m-1)\begin{pmatrix}k\\\ell-m\end{pmatrix}\begin{pmatrix}j+m\\m\end{pmatrix}\]
This sequence satisfies the recursion
\begin{align*}
    \sum^{k-\ell-1}_{j=i+1}c_{j,k,\ell}&=c_{i,k,\ell+1}
\end{align*}
In particular,
\[c_{-1,k,\ell}=(\ell-1)\begin{pmatrix}k\\\ell\end{pmatrix}.\]
We have:
\begin{align*}
    kA^{k-1}[L,A]&-\sum^{k-1}_{i=0}A^{k-1-i}[L,A]A^i=\sum^{k-1}_{i=0}A^{k-1-i}\left[A^i,[L,A]\right]\\
    &=\sum^{k-1}_{i=0}\sum^{i-1}_{j=0}A^{k-2-j}[A,[L,A]]A^j=\sum^{k-2}_{j=0}(k-1-j)A^{k-2-j}[A,[L,A]]A^j\\
    &=\sum^{k-2}_{j=0}(j+1)A^j[A,[L,A]]A^{k-2-j}=\sum^{k-2}_{j=0}c_{j,k,1}A^{j}[A,[L,A]]A^{k-2-j}\\
    &=\sum^{k-2}_{j=0}c_{j,k,1}[A,[L,A]]A^{k-2}+\sum^{k-2}_{j=1}\sum^{j-1}_{i=0}c_{j,k,1}A^i[A,[A,[L,A]]]A^{k-3-i}\\
    &=c_{-1,k,2}[A,[L,A]]A^{k-2}+\sum^{k-3}_{i=0}c_{i,k,2}A^i[A,[A,[L,A]]]A^{k-3-i}\\
        &=\sum_{i=1}^{k-1}c_{-1,k,i+1}\underbrace{[A,[\cdots[A,}_{i\text{ times}}[L,A]]]]A^{k-1-i}
\end{align*}
Substituting \(c_{-1,k,i}\) yields the lemma.
\end{proof}
For \(C,D,E\) be formal power series such that
\[C=D+s^kt^\ell E.\]
We denote this relation as \(C=D+\mathcal{O}(s^kt^\ell)\).
\begin{theorem}
    Let \(B'_i,i\geq 0\) be a sequence of operators satisfying \([L,B'_i]=B'_{i-1}[L,A']\) with \(B'_0=1\). Then a solution to \eqref{eq:beq} satisfies
    \[B=(B'_0+tB'_1+t^2B'_2+\cdots+t^{k-1}B'_{k-1})\mathrm{e}^{tA}+\mathcal{O}(s^kt^{2k}).\]\label{bth}
\end{theorem}
\begin{proof}
By definition,
\[\left[L,\sum^{k-1}_{i=0}t^iB'_i\right]=t\sum^{k-2}_{i=0}t^iB'_i[L,A']\]
From the lemma, we have
\[[L,\mrm{e}^{tA}]=t\mrm{e}^{tA}[L,A]-t[L,A']\mrm{e}^{tA}\]
Combining these two identities, we get
    \begin{align*}
        \left[L,\sum^{k-1}_{i=0}t^iB'_i\mrm{e}^{tA}\right]&=t\sum^{k-1}_{i=0}t^iB'_i\mrm{e}^{tA}[L,A]-t^kB'_{k-1}[L,A']\mrm{e}^{tA}.
    \end{align*}
As \([a_0,[L,a_0]]=0\), we have \(A_i|_{s=0}=0\). So \(A'=\mathcal{O}(st)\), and it follows \(B'_k=\mathcal{O}(s^kt^k)\). Then the last term is \(\mathcal{O}(s^kt^{2k})\) so this expansion is a valid solution of \eqref{eq:beq} up to this degree.
\end{proof}
Theorem \ref{bth} gives us a way to compute \(b_{i,j}\) rapidly. If we take
\[B_i=\frac{1}{i!}\frac{\partial^iB}{\partial t^i}\bigg|_{t=0}\]
then \([L,B_i]=B_{i-1}[L,A]\) with \(B_0=1\). This is the same recurrence relation satisfied by \(B'_i\), with \(A'\) swapped for \(A\). Using \(B_0'=1\), we have
\begin{equation}
B=\mathrm{e}^{tA}+\mathcal{O}(st^2)
\end{equation}
which means that
\begin{align*}
b_{0,0}=1, \quad b_{0,1}=a_0, \quad b_{0,2}=\tfrac{1}{2}a_0^2,\quad b_{1,0}=0,\quad b_{1,1}=a_1.
\end{align*}
Then \(W_1\) is readily found by \eqref{eq:dequ} to be
\[W_1=a_1-\tfrac{1}{2}a^2_0\]
which matches \eqref{eq:y1sol}. To go further, we note that \(B_1=A\) so \(B_1'=A'\) by analogy. Using Theorem \ref{bth} again,
\begin{equation}
B=(1+tA')\mrm{e}^{tA}+\mathcal{O}(s^2t^4)
\end{equation}
from which we calculate
\begin{align}
W_2&=\tfrac{1}{2}a_1^2+\tfrac{1}{2}[a_2,a_0]-\tfrac{1}{3}a_0^2a_1-\tfrac{1}{3}a_0a_1a_0+\tfrac{1}{6}a_1a_0^2+\tfrac{5}{24}a_0^4+a_3-a_{2,1}+a_{1,2}\label{eq:w2}
\end{align}
The last two operators are found to be
\begin{align*}
    a_{2,1}&=\sum_{i=1}^d\left[\frac{3V_i\int\! V_i-3\int\! V_i^2}{2p_i^4}\pd{}{q_i}+\frac{V_i'\int\! V_i}{2p_i^3}\pd{}{p_i}-\hbar^2\left(\frac{9\int\! V_i}{4p_i^4}\frac{\partial^3}{\partial q_i^3}+\frac{3V_i}{4p_i^4}\frac{\partial^2}{\partial q_i^2}\right.\right.\\
    &\qquad\left.\left.+\frac{7V_i}{4p_i^3}\frac{\partial^3}{\partial p_i\partial q_i^2}+\frac{V_i'}{2p_i^3}\frac{\partial^2}{\partial p_i\partial q_i}+\frac{5V_i'}{8p_i^2}\frac{\partial^3}{\partial p_i^2\partial q_i}+\frac{V_i''}{8p_i^2}\frac{\partial^2}{\partial p_i^2}+\frac{V_i''}{8p_i}\frac{\partial^3}{\partial p_i^3}\right)\right]\\
    a_{1,2}&=\hbar^2\sum_{i=1}^d\left(-\frac{\int\!V_i}{2p_i^4}\frac{\partial^3}{\partial q_i^3}+\frac{V_i}{2p_i^4}\frac{\partial^2}{\partial q_i^2}-\frac{5V_i}{6p_i^3}\frac{\partial^3}{\partial p_i\partial q_i^2}+\frac{V_i'}{4p_i^4}\pd{}{q_i}-\frac{5V_i'}{12p_i^2}\frac{\partial^3}{\partial q_i\partial p_i^2}\right.\\
    &\qquad\left.+\frac{V_i''}{12p_i^3}\pd{}{p_i}-\frac{V_i''}{12p_i^2}\frac{\partial^2}{\partial p_i^2}-\frac{V_i''}{12p_i}\frac{\partial^3}{\partial p_i^3}\right)
\end{align*}
Then \(Y_2\) is determined from \eqref{eq:solution}. Thus far we have determined
\[B=1+tA+t^2(\tfrac{1}{2}A^2+A_1)+t^3(\tfrac{1}{6}A^3+A_1A+A_2)+\mcl{O}(t^4)\]
up to third order in \(t\). Reading off the quadratic and cubic terms, and using the analogy,
\[B'_2=\tfrac{1}{2}(A')^2+A'_1,\qquad B'_3=\tfrac{1}{6}(A')^3+A'_1A'+A'_2\]
we then generate more terms in the expansion
\[B=\left\{1+tA'+t^2\left[\tfrac{1}{2}(A')^2+A'_1\right]+t^3\left[\tfrac{1}{6}(A')^3+A'_1A'+A'_2\right]\right\}\mrm{e}^{tA}+\mathcal{O}(s^4t^8)\]
This is sufficient to determine \(W_3,W_4,W_5\) in terms of \(a\)-operators. They require the explicit calculation of \(a_{i,j},a_{i,j,k}\) for \(i+j+k<10\), \(i+j+k\neq 8\). It is no challenge in this algebraic exercise to continue working out the components of \(B\) to an arbitrary degree, besides the limitations set upon us by our computers.
\section{Assignment of the Correction Functions}
We have shown how to realise \(W_0,W_1,W_2,\ldots\) as differential operators whose only momenta dependence in the coefficients is by reciprocal powers. Since \(Y_k\) is a homogeneous polynomial of order \(n-2k\), \(W_kZ_0\) must be homogeneous to the same degree, and in particular
\[\left(\prod^d_{i=1}p_i^{2k}\right)W_k\]
has polynomial coefficients in the momenta. It is sufficient to take
\begin{equation}
Z_k=\sum_{i_1+i_2+\cdots+i_d+j_{12}+j_{13}+\cdots+j_{d-1,d}=n+2(d-1)k}\alpha_{\mbf{i},\mbf{j}}\prod^d_{a=1}p_a^{i_a-2k}\prod_{b< c}m_{bc}^{j_{bc}}.\label{eq:mpol}
\end{equation}
where \(\alpha_{\mbf{i},\mbf{j}}\) are (not necessarily unique) constants and \(\mathbf{j}=(j_{k\ell})_{1\leq k<\ell\leq d}\). We set 
\begin{equation}
  R_{kij}\coloneqq \Res_{p_j=0}\frac{Y_k}{p_j^{i+1}}\prod_{\ell=1}^dp_\ell^{2k}.\label{eq:comp}
\end{equation}
It is necessary and sufficient for \(X\) to be a polynomial integral that
\begin{equation}
R_{kij}= 0,\qquad 0\leq i\leq 2k-1,1\leq j\leq d, 1\leq k\leq \lceil\tfrac{1}{2}n\rceil \label{eq:res}
\end{equation}
These are the compatibility equations. They govern the assignment of the parameters \(\alpha_{\mbf{i},\mbf{j}}\) and the potential function \(V\). There will be a total of
\[N(n,d)\coloneqq d\sum^{\lceil\frac{1}{2}n\rceil}_{k=1}\sum^{2k}_{i=1}\begin{pmatrix}n+2(d-1)k-i+d-1\\d-2\end{pmatrix}\]
integro-differential equations involving the \(d\) summands of the potential. In particular,
\begin{align*}
    N(n,2)&=\begin{cases}\frac{1}{2}n(n+2),&\quad n\text{ even},\\\frac{1}{2}(n+1)(n+3),&\quad n\text{ odd};\end{cases}\\
    N(n,3)&=\begin{cases}
        \frac{3}{8}n(n+2)(4n+5),&\quad n\text{ even},\\\frac{3}{8}(n+1)(n+3)(4n+7),&\quad n\text{ odd};
    \end{cases}\\
   N(n,4)&=\begin{cases}
       \frac{1}{12}n(n+2)(45n^2+122n+62),&\quad n\text{ even},\\
        \frac{1}{4}(n+1)(n+3)(15n^2+60n+53),&\quad n\text{ odd}.
    \end{cases}
\end{align*}
Let \(\Pi(n,d)\) be the number of \(\alpha_{\mbf{i},\mbf{j}}\) to determine. Eliminating redundant parameters via \eqref{eq:mdep}, we calculate
\begin{align*}
    \Pi(n,2)&=\sum^{\lceil\frac{1}{2}n\rceil}_{k=0}\begin{pmatrix}n+2k+2\\2\end{pmatrix}\\
    \Pi(n,3)&=\sum_{k=0}^{\lceil\frac{1}{2}n\rceil}\left[2\begin{pmatrix}n+4k+4\\4\end{pmatrix}-\begin{pmatrix}n+4k+3\\3\end{pmatrix}\right]\\
    \Pi(n,4)&=\sum^{\lceil\frac{1}{2}n\rceil}_{k=0}\left[2\begin{pmatrix}n+6k+7\\7\end{pmatrix}-\begin{pmatrix}n+6k+5\\5\end{pmatrix}\right]
\end{align*}
\section{Calculation of the Potential and the Integral}
We give here an outline about how to proceed in finding the potential and the integral. First it is necessary to find the appropriate algebraic expressions for \(W_1,W_2,\ldots\) in terms of \(a_i,a_{i,j},a_{i,j,k},\ldots\) as given in Section 4. Once these are found, the particular \(a\)-operators which appear in these expressions should then be computed. The inversion formula \eqref{eq:mon2} is sufficient. No variables are mixed, so this amounts to computing expressions of the form
\[\iint\cdots\int V_i^{(k)}V_i^{(\ell)}\cdots V_i^{(m)} \]
and such like. Using integration by parts, we can reduce to a few canonical forms which in the following sections are denoted \(Q_{i,j}\). Once this is achieved, we compute the compatibility equations \eqref{eq:comp}. Then we solve for \(V\) for all the possible assignments of the \(\alpha\)-parameters such that \(Z_0\) is not identically zero. Once \(V\) and \(\alpha_{\mathbf{i},\mathbf{j}}\) are determined, we can compute \(Y\) using equation (28) and take the self-adjoint part to find \(X\). Our calculations were performed using Mathematica.
\section{Equations for the Potential in Two Dimensions}
For two dimensions, it is possible for us to express each of the compatibility equations as the integral develops in a concise form that is valid for arbitrary order. Indeed, let us take
\begin{subequations}
\begin{align}
\xi_{ki1}(q_1)&\coloneqq\sum_{j=0}^{n+2k-i}\alpha_{i,j,n+2k-i-j}q_1^{n+2k-i-j}\\
\xi_{kj2}(q_2)&\coloneqq\sum_{i=0}^{n+2k-j}\alpha_{i,j,n+2k-i-j}(-1)^iq_2^{n+2k-i-j}
\end{align}
\end{subequations}
Then \eqref{eq:mpol} can be alternately written as
\begin{subequations}
\begin{align}
Z_k&=\sum^{n+2k}_{i=0}\sum^{n+2k-i}_{j=0}\frac{(-q_2)^j}{j!}\xi^{(j)}_{ki1}(q_1)p_1^{i+j-2k}p_2^{n-i-j}\label{eq:zk}\\
&=(-1)^n\sum^{n+2k}_{j=0}(-1)^j\sum^{n+2k-j}_{i=0}\frac{(-q_1)^i}{i!}\xi^{(i)}_{kj2}(q_2)p_1^{n-i-j}p_2^{i+j-2k}
\end{align}
\end{subequations}
The component in \(R_{ki1}\) which comes from \(Z_k\) is
\[\xi_{ki1}(q_1)+\sum^i_{\ell=1}\frac{(-q_2)^\ell}{\ell!}\xi_{k,i-\ell,1}^{(\ell)}(q_1).\]
A similar expression can be obtained for \(R_{ki2}\). Each equation \(R_{kij}=0\) can then be repurposed as a definition for \(\xi_{kij}(q_j)\) in order to eliminate this term from subsequent equations. We designate \(S_{kij}\) as the equation \(R_{kij}=0\) once this elimination has been performed.
\subsection{Linear Equations}
We calculate
\begin{equation}
S_{10i}:\xi_{00i}'(q_i)Q_{i,1}+\xi_{10i}(q_i)= 0\label{eq:standard1}
\end{equation}
\begin{equation}
S_{11i}:\left[\xi_{01i}(q_i)Q_{i,1}\right]'+\xi_{11i}(q_i)= 0\label{eq:standard2}
\end{equation}
where \(Q_{i,1}\coloneqq\int V_i\,\mrm{d}q_i\). These are required to hold for all integrals regardless of order. When the coefficients are not all equal to zero, the solution is a rational function with no simple poles.
\begin{theorem}
    The potential summand \(V_i\) is a solution to the linear compatibility equations for non-trivial values of the parameters if \(\dfrac{\partial Z_0}{\partial q_i}\bigg|_{p_i=0}\neq 0\) or \(\dfrac{\partial Z_0}{\partial p_i}\bigg|_{p_i=0}\neq 0\).
\end{theorem}
\begin{proof}
    We require that \(\xi_{00i}'(q_i)\neq 0\) or \(\xi_{01i}(q_i)\neq 0\). From \eqref{eq:zk}, we have
    \begin{align*}
        \pd{Z_0}{q_1}\bigg|_{p_1=0}&=\xi'_{001}(q_1)p_2^n&\pd{Z_0}{q_2}&=(-1)^n\xi'_{002}(q_2)p_1^n\\
        \pd{Z_0}{p_1}\bigg|_{p_1=0}&=[\xi_{011}(q_1)-q_2\xi'_{001}(q_1)]p_2^{n-1}&\pd{Z_0}{p_2}&=(-1)^{n+1}[\xi_{012}(q_2)+q_1\xi'_{002}(q_1)]p_1^{n-1}
    \end{align*}
The theorem follows naturally.
\end{proof}
The above result allows us to identify whether the linear compatibility equations hold just from inspection of the leading-order term.

Working out the compatibility equation from \eqref{eq:compatibility} for \(k=2\), after some simplification, gives
\begin{equation}\begin{aligned}0&=\frac{\mathrm{d}^{n+1}}{\mathrm{d}q_1^{n+1}}\left[\xi_{011}(q_1)Q_{1,1}\right]-q_2\frac{\mathrm{d}^{n+1}}{\mathrm{d}q_1^{n+1}}\left[\xi'_{001}(q_1)Q_{1,1}\right]\\
    &\qquad+\frac{\mathrm{d}^{n+1}}{\mathrm{d}q_2^{n+1}}\left[\xi_{012}(q_2)Q_{2,1}\right]+q_1\frac{\mathrm{d}^{n+1}}{\mathrm{d}q_2^{n+1}}\left[\xi'_{002}(q_2)Q_{2,1}\right].
    \end{aligned}\label{eq:compat2}
    \end{equation}
Equation \eqref{eq:compat2} is manifestly the same as the four simultaneous equations \(S_{101},S_{111},S_{102},S_{112}=0\). 
\subsection{Quadratic Equations}
We compute
\begin{equation}    S_{20i}:\xi_{00i}'(q_i)\left(\tfrac{3}{2}Q_{i,2}-V_iQ_{i,1}\right)-\tfrac{1}{2}\xi_{00i}''(q_i)Q_{i,1}^2-\tfrac{1}{4}\hbar^2\xi_{00i}'''(q_i)Q_{i,1}+\xi_{20i}(q_i)=0\label{eq:s20}\end{equation}
\begin{equation}    S_{21i}:\left[\xi_{01i}(q_i)\left(\tfrac{3}{2}Q_{i,2}-V_iQ_{i,1}\right)-\tfrac{1}{2}\xi_{01i}'(q_i)Q_{i,1}^2-\tfrac{1}{4}\hbar^2\xi_{01i}''(q_i)Q_{i,1}\right]'+\xi_{21i}(q_i)= 0\label{eq:s21}
\end{equation}
\begin{equation}\begin{aligned}S_{22i}&:\xi_{02i}'(q_i)\left(\tfrac{3}{2}Q_{i,2}-\tfrac{1}{4}\hbar^2V_i'\right)+\tfrac{1}{2}\xi_{02i}''(q_i)\left(Q_{i,1}^2-\hbar^2V_i\right)\\
    &\qquad+[\xi_{12i}'(q_i)-\tfrac{1}{4}\hbar^2\xi_{02i}'''(q_i)]Q_{i,1}+\xi_{22i}(q_i)= 0\end{aligned}\label{eq:s22}\end{equation}
\begin{equation}\begin{aligned}S_{23i}&:\left\{\xi_{03i}(q_i)\left(\tfrac{3}{2}Q_{i,2}-\tfrac{1}{4}\hbar^2V_i'\right)+\tfrac{1}{2}\xi_{03i}'(q_i)\left(Q_{i,1}^2-\hbar^2V_i\right)\right.\\
    &\qquad\left.+[\xi_{13i}(q_i)-\tfrac{1}{4}\hbar^2\xi_{03i}''(q_i)]Q_{i,1}\right\}'+\xi_{23i}(q_i)= 0\end{aligned}
    \end{equation}
where \(Q_{i,2}\coloneqq\int V_i^2\,\mrm{d}q_i\). These equations hold for all \(n\) but are superfluous for \(n<3\).

As a special case, let us take \(Z_0=\frac{1}{n-4}p_1^2p_2^2m_{12}^{n-4}\). Then \(\xi_{02i}(q_i)=\frac{1}{n-4}q_i^{n-4}\) and differentiating \eqref{eq:s22}, we obtain
\begin{equation}
\begin{aligned}
0&=\hbar^2q_i^{n-4}Q_{i,1}'''+2\hbar^2(n-5)q_i^{n-5}Q_{i,1}''-6q^{n-4}_i\left(Q'_{i,1}\right)^2-4(n-5)q_i^{n-5}Q_{i,1}Q_{i,1}'\\
&\qquad+2(n-5)q_i^{n-6}Q_{i,1}^2+Q_{i,1}'[(n-5)(n-8)\hbar^2q_i^{n-6}-4q_i\xi'_{12i}(q_i)]\\
&\qquad-2Q_{i,1}[(n-5)(n-6)\hbar^2q_i^{n-7}-2(n-5)\xi'_{12i}(q_i)+2q_i\xi''_{12i}(q_i)]\\
&\qquad+4(n-5)\xi_{22i}(q_i)-4q_i\xi_{22i}'(q_i).
\end{aligned}
\end{equation}
Since \(\xi_{00i}=\xi_{01i}=0\), \eqref{eq:s20} and \eqref{eq:s21} give \(\xi_{20i}=\xi_{21i}=0\) and in particular \(\alpha_{02,n+2}=\alpha_{12,n+1}=\alpha_{20,n+2}=\alpha_{21,n+1}=0\). So \(\xi_{22i}\) is of degree \(n\) or less. This proves Conjecture 2 of Escobar-Ruiz, Linares and Winternitz \cite{esc20}. For \(n=3,4\), the solutions give the exotic potentials of Gravel \cite{gra04} and Marquette, Sajedi and Winternitz \cite{mar17}.
\subsection{Cubic Equations}
The next set of equations is computed likewise. Taking the auxiliary function
\[Q_{i,3}=\int\left[V_i^3+\tfrac{1}{4}\hbar^2(V_i')^2\right]\,\mathrm{d}q_i\]
we obtain
\[\begin{aligned}S_{30i}&: \xi'_{00i}(q_i)\left[\tfrac{5}{2}Q_{i,3}-\tfrac{3}{2}V_iQ_{i,2}-\tfrac{1}{2}V_i^2Q_{i,1}+\tfrac{1}{2}V_i'Q_{i,1}^2+\hbar^2\left(\tfrac{1}{4}V_i''Q_{i,1}-\tfrac{1}{2}V_iV_i'\right)\right]\nonumber\\
  &\qquad+\xi''_{00i}(q_i)\left[V_iQ_{i,1}^2-\tfrac{3}{2}Q_{i,1}Q_{i,2}+\hbar^2\left(\tfrac{3}{4}V_i'Q_{i,1}-\tfrac{1}{2}V_i^2\right)\right]\nonumber\\
  &\qquad+\xi'''_{00i}(q_i)\left[\tfrac{1}{6}Q_{i,1}^3+\hbar^2\left(V_iQ_{i,1}-\tfrac{5}{4}Q_{i,2}\right)\right]+\tfrac{1}{4}\hbar^2\xi''''_{00i}(q_i)Q_{i,1}^2\nonumber\\
  &\qquad+\tfrac{1}{16}\hbar^4\xi'''''_{00i}(q_i)Q_{i,1}+\xi_{30i}(q_i)= 0\end{aligned}\]
\[\begin{aligned}S_{31i}&: \left\{\xi_{01i}(q_i)\left[\tfrac{5}{2}Q_{i,3}-\tfrac{3}{2}V_iQ_{i,2}-\tfrac{1}{2}V_i^2Q_{i,1}+\tfrac{1}{2}V_i'Q_{i,1}^2+\hbar^2\left(\tfrac{1}{4}V_i''Q_{i,1}-\tfrac{1}{2}V_iV_i'\right)\right]\right.\nonumber\\
  &\qquad+\xi'_{01i}(q_i)\left[V_iQ_{i,1}^2-\tfrac{3}{2}Q_{i,1}Q_{i,2}+\hbar^2\left(\tfrac{3}{4}V_i'Q_{i,1}-\tfrac{1}{2}V_i^2\right)\right]\nonumber\\
  &\qquad+\xi''_{01i}(q_i)\left[\tfrac{1}{6}Q_{i,1}^3+\hbar^2\left(V_iQ_{i,1}-\tfrac{5}{4}Q_{i,2}\right)\right]+\tfrac{1}{4}\hbar^2\xi'''_{01i}(q_i)Q_{i,1}^2\nonumber\\
  &\qquad\left.+\tfrac{1}{16}\hbar^4\xi''''_{01i}(q_i)Q_{i,1}\right\}'+\xi_{31i}(q_i)= 0\end{aligned}\]
  \[\begin{aligned}S_{32i}&:\xi'_{02i}(q_i)\left[\tfrac{5}{2}Q_{i,3}-\tfrac{3}{2}V_i^2Q_{i,1}+\hbar^2\left(\tfrac{1}{4}V_i''Q_{i,1}-\tfrac{1}{2}V_iV_i'\right)\right]-\xi''_{02i}(q_i)\left[V_iQ_{i,1}^2\right.\nonumber\\
  &\qquad\left.+\hbar^2\left(\tfrac{1}{2}V_i^2-\tfrac{3}{4}V_i'Q_{i,1}\right)\right]-\xi'''_{02i}(q_i)\left[\tfrac{1}{3}Q_{i,1}^3+\hbar^2\left(\tfrac{5}{4}Q_{i,2}-\tfrac{3}{4}V_iQ_{i,1}\right)-\tfrac{1}{16}\hbar^4V_i'\right]\nonumber\\
  &\qquad+\tfrac{1}{8}\hbar^4\xi''''_{02i}(q_i)V_i+\tfrac{1}{16}\hbar^4\xi'''''_{02i}(q_i)Q_{i,1}+\xi_{12i}'(q_i)\left(\tfrac{3}{2}Q_{i,2}-V_iQ_{i,1}\right)\nonumber\\
  &\qquad-\tfrac{1}{2}\xi_{12i}''(q_i)Q_{i,1}^2-\tfrac{1}{4}\hbar^2\xi_{12i}'''(q_i)Q_{i,1}+\xi_{32i}(q_i)= 0\end{aligned}\]
  \[\begin{aligned}S_{33i}&:\left\{\xi_{03i}(q_i)\left[\tfrac{5}{2}Q_{i,3}-\tfrac{3}{2}V_i^2Q_{i,1}+\hbar^2\left(\tfrac{1}{4}V_i''Q_{i,1}-\tfrac{1}{2}V_iV_i'\right)\right]-\xi'_{03i}(q_i)\left[V_iQ_{i,1}^2\right.\right.\nonumber\\
  &\qquad\left.+\hbar^2\left(\tfrac{1}{2}V_i^2-\tfrac{3}{4}V_i'Q_{i,1}\right)\right]-\xi''_{03i}(q_i)\left[\tfrac{1}{3}Q_{i,1}^3+\hbar^2\left(\tfrac{5}{4}Q_{i,2}-\tfrac{3}{4}V_iQ_{i,1}\right)-\tfrac{1}{16}\hbar^4V_i'\right]\nonumber\\
  &\qquad+\tfrac{1}{8}\hbar^4\xi'''_{03i}(q_i)V_i+\tfrac{1}{16}\hbar^4\xi''''_{03i}(q_i)Q_{i,1}+\xi_{13i}(q_i)\left(\tfrac{3}{2}Q_{i,2}-V_iQ_{i,1}\right)\nonumber\\
  &\qquad\left.-\tfrac{1}{2}\xi_{13i}'(q_i)Q_{i,1}^2-\tfrac{1}{4}\hbar^2\xi_{13i}''(q_i)Q_{i,1}\right\}'+\xi_{33i}(q_i)= 0\end{aligned}\]
  \[\begin{aligned}
      S_{34i}&:\xi_{04i}'(q_i)\left(\tfrac{5}{2}Q_{i,3}-\tfrac{5}{4}\hbar^2V_iV_i'+\tfrac{1}{16}\hbar^4V_i'''\right)+\xi_{04i}''(q_i)\left[\tfrac{3}{2}Q_{i,1}Q_{i,2}\right.\\
      &\qquad\left.-\hbar^2\left(\tfrac{3}{2}V_i^2+\tfrac{1}{4}V_i'Q_{i,1}\right)+\tfrac{1}{4}\hbar^4V_i''\right]+\xi_{04i}'''(q_i)\left[\tfrac{1}{6}Q_{i,1}^3-\hbar^2\left(\tfrac{5}{4}Q_{i,2}+\tfrac{1}{2}V_iQ_{i,1}\right)\right.\\
      &\qquad\left.+\tfrac{3}{8}\hbar^4V_i'\right]-\tfrac{1}{4}\hbar^2\xi_{04i}''''(q_i)\left(Q_{i,1}^2-\hbar^2V_i\right)+\xi_{14i}'(q_i)\left(\tfrac{3}{2}Q_{i,2}-\tfrac{1}{4}\hbar^2V_i'\right)\\
      &\qquad+\xi_{14i}''(q_i)\left(\tfrac{1}{2}Q_{i,1}^2-\tfrac{1}{2}\hbar^2V_i\right)+\left[\tfrac{1}{16}\hbar^4\xi_{04i}'''''(q_i)-\tfrac{1}{4}\hbar^2\xi_{14i}'''(q_i)+\xi_{24i}'(q_i)\right]Q_{i,1}\\
      &\qquad+\xi_{34i}(q_i)= 0
  \end{aligned}\]
  \[\begin{aligned}
      S_{35i}&:\left\{\xi_{05i}(q_i)\left(\tfrac{5}{2}Q_{i,3}-\tfrac{5}{4}\hbar^2V_iV_i'+\tfrac{1}{16}\hbar^4V_i'''\right)+\xi_{05i}'(q_i)\left[\tfrac{3}{2}Q_{i,1}Q_{i,2}\right.\right.\\
      &\qquad\left.-\hbar^2\left(\tfrac{3}{2}V_i^2+\tfrac{1}{4}V_i'Q_{i,1}\right)+\tfrac{1}{4}\hbar^4V_i''\right]+\xi_{05i}''(q_i)\left[\tfrac{1}{6}Q_{i,1}^3-\hbar^2\left(\tfrac{5}{4}Q_{i,2}+\tfrac{1}{2}V_iQ_{i,1}\right)\right.\\
      &\qquad\left.+\tfrac{3}{8}\hbar^4V_i'\right]-\tfrac{1}{4}\hbar^2\xi_{05i}'''(q_i)\left(Q_{i,1}^2-\hbar^2V_i\right)+\xi_{15i}(q_i)\left(\tfrac{3}{2}Q_{i,2}-\tfrac{1}{4}\hbar^2V_i'\right)\\
      &\qquad\left.+\xi_{15i}'(q_i)\left(\tfrac{1}{2}Q_{i,1}^2-\tfrac{1}{2}\hbar^2V_i\right)+\left[\tfrac{1}{16}\hbar^4\xi_{05i}''''(q_i)-\tfrac{1}{4}\hbar^2\xi_{15i}''(q_i)+\xi_{25i}(q_i)\right]Q_{i,1}\right\}'\\
      &\qquad+\xi_{35i}(q_i)= 0
  \end{aligned}\]
 The higher order equations are of a similar form but are prohibitively long. We give here their auxiliary functions 
\begin{align*}
  Q_{i,4}&=\int\left[V_i^4+\hbar^2V_i(V_i')^2+\tfrac{1}{20}\hbar^4(V_i'')^2\right]\,\mathrm{d}q_i\\
Q_{i,5}&=\int\left[V_i^5+\tfrac{5}{2}\hbar^2V_i^2(V_i')^2+\tfrac{1}{4}\hbar^4V_i(V_i'')^2+\tfrac{1}{112}\hbar^6(V_i''')^2\right]\,\mathrm{d}q_i.
\end{align*}
Given that if \(H\) is superintegrable so must \(H+\alpha\), we would expect that the auxiliary functions possess a similar symmetry under translation. Indeed, we note that for \(V_i\mapsto V_i+\alpha\), \(Q_{i,j}\mapsto Q_{i,j}+j\alpha Q_{i,j-1}+\cdots\). We expect the general expansion
\begin{equation}
\begin{aligned}
Q_{i,j}&=\int\left[V_i^j+\frac{1}{4}\begin{pmatrix}j\\3\end{pmatrix}\hbar^2V_i^{j-3}(V_i')^2+\frac{1}{20}\begin{pmatrix}j\\4\end{pmatrix}\hbar^4V_i^{j-4}(V_i'')^2\right.\\
&\qquad\left.+\frac{1}{112}\begin{pmatrix}j\\5\end{pmatrix}\hbar^6V_i^{j-5}(V_i''')^2+\cdots\right]\mathrm{d}q_i
\end{aligned}
\end{equation}
though we have as yet no way of determining the law which produces the next coefficient.
\section{Rational Potentials}
The classification of superintegrable systems with higher-order integrals has focused mainly on exotic potentials, i.e., potential for which the linear equations \(S_{10i},S_{11i}\) vanish identically. We shall provide some details on obtaining potentials in the case where \(S_{10i},S_{11i}\) do not vanish. We perform a complete classification of fourth-order standard models as an application of these compatibility equations. The generalisation to higher order follows naturally, the only difficulty being algebraic complexity.
\begin{lemma}
    If \(V_i\) is a solution to the linear compatibility equations for \(n=3,4\) then \(Q_{i,1}\) and \(Q_{i,2}\) must be rational. If \(n=5,6\) then \(Q_{i,3}\) must be rational also. \label{lemma4}
\end{lemma}
\begin{proof}
If \(S_{10i}\) or \(S_{11i}\) possess non-trivial parameters then so do \(S_{j0i}\) or \(S_{j1i}\) for \(j=2,3\). We then use \(S_{j0i}\) or \(S_{j1i}\) to solve for \(Q_{i,j}\) in terms of \(V_i\), its derivatives, \(Q_{i,k}\) for \(k<j\), ratios of \(\xi_{jki}\) and its derivatives. The lemma follows by induction.
\end{proof}
\begin{lemma}
    If \(V_i\) is a solution to the linear compatibility equations for \(n\leq 4\), then it must be of the form
    \[V_i=\beta_1q_i^2+\beta_2q_i+\sum_{j=1}^{n-1}\frac{\gamma_j}{(q_i-\delta_j)^2}\]\label{lemma3}
\end{lemma}
\begin{proof}
As \(S_{j1i}\) is of the same form as \(S_{j0i}\) but differentiated, we may assume without loss of generality that \(S_{j0i}\) holds non-trivially, and prove that \(V_i\) cannot be otherwise. Suppose that up to leading-order, \(V_i=\beta q_i^k+\cdots\) where \(k>2\) and \(\xi_{00i}=\gamma q_i^{\ell+1}/(\ell+1)+\cdots\). Then from \(S_{20i}\) the leading-order term from the potential summand is
\[-\frac{\beta^2\gamma (2k+2k^2+\ell+2k\ell)q_i^{2k+\ell+1}}{2(k+1)^2(2k+1)}+\cdots\]
which must be annihilated by \(\xi_{20i}\sim \mathcal{O}(q_i^{n+2})\) where \(n\leq 4\). However, \(2k+\ell+1\geq 7\) so this is not possible.

From Lemma \ref{lemma4}, we know that \(Q_{i,1}\) is rational so \(V_i\) cannot have simple poles. Suppose \(V_i\sim \delta q_i^{-k}\) for \(k>2\) as \(q_i\to 0\). Then the lowest-order term in \(S_{20i}\) is
\[-\frac{\delta^2\gamma q_i^{1-2k+\ell}(2k-2k^2-\ell+2k\ell)}{2(k-1)^2(2k-1)}\]
The coefficient is not zero for any integer value of \(k>2\) when \(\ell=0,1,2,3\). Since \(k\geq 3\), the exponent \(1-2k+\ell\leq -2\) so cannot be annihilated by \(\xi_{20i}\). Then \(V_i\) has only poles of second-order. These are exactly the simple poles of \(Q_{i,1}\). From \(S_{10i}\), the number of these is equal to \(\ell\leq n-1\).
\end{proof}
\begin{theorem}
    If \(V_i\) for \(i=1,2\) is a solution to the linear compatibility equations and \(H=\tfrac{1}{2}(p_1^2+p_2^2)+V_1+V_2\) is superintegrable of at least fourth-order then the potential summands must be of the form
    \begin{align*}
(\mathrm{I})&\quad \beta q_i,\\
(\mathrm{II})&\quad\beta_1q_i^2+\frac{\beta_2}{q_i^2},\\
(\mathrm{III})&\quad\beta_1\left(q_i^2+\frac{\beta_2^2}{q_i^2}\right)+\hbar^2\left[\frac{2(q_i^2+\beta_2)}{(q_i^2-\beta_2)^2}-\frac{1}{8q_i^2}\right],\\
(\mathrm{IV})&\quad\frac{3\hbar^2q_i(q_i^3+2\beta)}{(q_i^3-\beta)^2}.
\end{align*}
\end{theorem}
\begin{proof}
For \(n=1,2\), we can solve the linear compatibility equations directly. The answers match with the known solutions, being Types I and II.

For \(n=3,4\), subject to the constraints of Lemma \ref{lemma3}, we look for systems of the following form
\begin{equation}
    V_i=\beta_1q_i^2+\beta_2 q_i+\frac{\beta_3}{q_i^2}+\frac{\beta_4}{(q_i-\beta_6)^2}+\frac{\beta_5}{(q_i-\beta_7)^2}\label{eq:ansatz}\end{equation}
We then eliminate all the simple poles of \(Q_{i,j}'\) as per Lemma \ref{lemma4}. This amounts to a series of algebraic equations involving \(\beta_1,\beta_2,\beta_3,\beta_4,\beta_5,\beta_6,\beta_7\). The four families listed above present all the admissible solutions.
\end{proof}
The Type II models have been well studied \cite{evans}. Type III is a Darboux trasformation of Type II (after a reparametrisation) \cite{marquette25}. Type III only appears in a special case in Gravel's classification, when it is a Darboux transformation of the simple harmonic oscillator \cite{dubov}. As for Type IV, we believe this model to be completely new.

To generalise to higher order, one needs to consider rational functions with more poles as per Lemma 3. However, the algebraic conditios become more complicated in this case and a general solution for any order higher than four is yet to be had.
\subsection{List of Potentials}
Once the forms I--IV are assigned, the compatibility equations reduce to a system in terms of rational functions which has at most four parameters occurring in a non-linear manner. This can readily be solved. The non-zero \(\alpha_{ijk}\) allow us to determine \(Z_0,Z_1,Z_2\). One can calculate the integral via
\[Y=W_0Z_0\]
for \(n=1\),
\[Y=W_0(Z_0+Z_1)+W_1Z_0\]
for \(n=2,3\) and
\[Y=W_0(Z_0+Z_1+Z_2)+W_1(Z_0+Z_1)+W_2Z_0\]
for \(n=4\). Then one takes \(X=\tfrac{1}{2}(Y+Y^*)\). We list below the potentials found for non-zero \(\alpha_{ijk}\). Models are to be distinguished either by their form or by the number of integrals they possess (that is, if there are additional integrals for special values of their parameters). At first-order, we obtain Hooke's law
\begin{enumerate}
\begin{multicols}{2}
    \item[(1)] \(\beta(q_1^2+q_2^2)\)
\end{multicols}
\end{enumerate}
\noindent
as well as an oblique linear potential (cf. 3a). At second-order, we have the Smorodinsky-Winternitz systems and their special cases:
\begin{enumerate}
\begin{multicols}{2}
    \item[(2)] \(\beta_1(q_1^2+q_2^2)+\displaystyle\frac{\beta_2}{q_1^2}+\frac{\beta_3}{q_2^2}\)
    \item[(2a)] \(\dfrac{\beta}{q_1^2}\)
    \item[(3)] \(\displaystyle\beta_1(q_1^2+4q_2^2)+\frac{\beta_2}{q_1^2}+\beta_3q_2\)
    \item[(3a)] \(\beta q_2\)
    \end{multicols}
\end{enumerate}
\noindent
At third-order are the Gravel models:
\begin{enumerate}
\begin{multicols}{2}
    \item[(2ai)] \(\displaystyle\frac{\hbar^2}{q_1^2}\)
    \item[(2b)] \(\displaystyle\frac{\beta}{q_1^2}+\frac{\hbar^2}{q_2^2}\)
    \item[(2bi)] \(\displaystyle \frac{\hbar^2}{q_1^2}+\frac{\hbar^2}{q_2^2}\)
    \item[(2c)] \(\displaystyle \beta(q_1^2+q_2^2)+\frac{\hbar^2}{q_1^2}\)
    \item[(2d)] \(\displaystyle \beta(q_1^2+q_2^2)+\frac{\hbar^2}{q_1^2}+\frac{\hbar^2}{q_2^2}\)
    \item[(4a)] \(\displaystyle \hbar^2\left[\frac{q_1^2+q_2^2}{8\beta^2}+\frac{2(q_1^2+\beta)}{(q_1^2-\beta)^2}\right]\)
    \item[(4b)] \(\displaystyle \hbar^2\left[\frac{q_1^2+q_2^2}{8\beta^2}+\frac{2(q_1^2+\beta)}{(q_1^2-\beta)^2}+\frac{1}{q_2^2}\right]\)
    \item[(5a)] \(\displaystyle \hbar^2\left[\frac{q_1^2+q_2^2}{8\beta^2}+\frac{2(q_1^2+\beta)}{(q_1^2-\beta)^2}+\frac{2(q_2^2+\beta)}{(q_2^2-\beta)^2}\right]\)
    \item[(5b)] \(\displaystyle \hbar^2\left[\frac{q_1^2+q_2^2}{8\beta^2}+\frac{2(q_1^2+\beta)}{(q_1^2-\beta)^2}+\frac{2(q_2^2-\beta)}{(q_2^2+\beta)^2}\right]\)
    \end{multicols}
\end{enumerate}
\noindent
At fourth-order, we obtain the potentials:
\begin{enumerate}
\item[(4)] \(\displaystyle \beta_1\left(q_1^2+q_2^2+\frac{\beta_2^2}{q_1^2}\right)+\frac{\beta_3}{q_2^2}+\hbar^2\left[\frac{2(q_1^2+\beta_2)}{(q_1^2-\beta_2)^2}-\frac{1}{8q_1^2}\right]\)
\begin{multicols}{2}
\item[(4c)] \(\displaystyle \hbar^2\left[\frac{2(q_1^2+\beta)}{(q_1^2-\beta)^2}-\frac{1}{8q_1^2}\right]\)
\item[(4d)] \(\displaystyle \hbar^2\left[\frac{2(q_1^2+\beta)}{(q_1^2-\beta)^2}-\frac{1}{8q_1^2}+\frac{1}{q_2^2}\right]\)
\end{multicols}
\item[(5)] \(\displaystyle \beta_1\left(q_1^2+q_2^2+\frac{\beta_2^2}{q_1^2}+\frac{\beta_3^2}{q_2^2}\right)+\hbar^2\left[\frac{2(q_1^2+\beta_2)}{(q_1^2-\beta_2)^2}-\frac{1}{8q_1^2}+\frac{2(q_2^2+\beta_3)}{(q_2^2-\beta_3)^2}-\frac{1}{8q_2^2}\right]\)
\begin{multicols}{2}
\item[(6)] \(\displaystyle \frac{3\hbar^2q_1(q_1^3+2\beta_1)}{(q_1^3-\beta_1)^2}+\frac{\beta_2}{q_2^2}\)
\item[(6a)] \(\displaystyle \frac{3\hbar^2q_1(q_1^3+2\beta)}{(q_1^3-\beta)^2}\)
\item[(6b)] \(\displaystyle \frac{3\hbar^2q_1(q_1^3+2\beta)}{(q_1^3-\beta)^2}+\frac{\hbar^2}{q_2^2}\)
\item[(7)] \(\displaystyle \hbar^2\left[\frac{3q_1(q_1^3+2\beta_1)}{(q_1^3-\beta_1)^2}+\frac{2(q_2^2+\beta_2)}{(q_2^2-\beta_2)^2}-\frac{1}{8q_2^2}\right]\)
\item[(8)] \(\displaystyle\hbar^2\left[\frac{3q_1(q_1^3+2\beta_1)}{(q_1^3-\beta_1)^2}+\frac{3q_2(q_2^3+2\beta_2)}{(q_2^3-\beta_2)^2}\right]\)
\end{multicols}
\end{enumerate}
It is of interest to note the following rational solutions, which have summands that are of Types I--IV, but their linear compatibility equations happen to be zero identically. Consequently, they are properly speaking exotic models. At third-order these are
\begin{enumerate}
\begin{multicols}{2}
    \item[(3b)] \(\dfrac{\hbar^2}{q_1^2}+\beta q_2\)
    \item[(4e)] \(\displaystyle \hbar^2\left[\frac{9(q_1^2+q_2^2)}{8\beta^2}+\frac{2(q_1^2+\beta)}{(q_1^2-\beta)^2}+\frac{1}{q_1^2}\right]\)
    \item[(9)] \(\beta(q_1^2+9q_2^2)\)
    \item[(10)] \(\beta(q_1^2+9q_2^2)+\dfrac{\hbar^2}{q_1^2}\)
    \item[(11)] \(\displaystyle \hbar^2\left[\frac{q_1^2+9q_2^2}{8\beta^2}+\frac{2(q_1^2+\beta)}{(q_1^2-\beta)^2}\right]\)
\end{multicols}
\end{enumerate}
and at fourth-order
\begin{enumerate}
\begin{multicols}{2}
\item[(12)] \(\displaystyle \hbar^2\left[\frac{2(q_2^2+\beta_1)}{(q_2^2-\beta_1)^2}-\frac{1}{8q_2^2}\right]+\beta_2q_2\)
\item[(13)] \(\displaystyle \frac{3\hbar^2q_2(q_2^3+2\beta_1)}{(q_2^3-\beta_1)^2}+\beta_2q_2\)
\item[(14)] \(\beta_1(q_1^2+4q_2^2)+\displaystyle\frac{\beta_2}{q_1^2}+\frac{\beta_3}{q_2^2}\)
\item[(14a)] \(\beta_1(q_1^2+4q_2^2)+\displaystyle\frac{\beta_2}{q_1^2}+\frac{\hbar^2}{q_2^2}\)
\item[(15)] \(\displaystyle\beta_1(q_1^2+16q_2^2)+\frac{\beta_2}{q_1^2}\)
\item[(16)] \(\displaystyle \beta_1(9q_1^2+4q_2^2)+\frac{\beta_2}{q_1^2}\)
\item[(17)] \(\displaystyle \beta_1(9q_1^2+4q_2^2)+\frac{\beta_2}{q_1^2}+\frac{\hbar^2}{q_2^2}\)
\end{multicols}
\item[(18)] \(\displaystyle \beta_1\left(q_1^2+4q_2^2+\frac{\beta_2^2}{q_1^2}\right)+\frac{\beta_3}{q_2^2}+\hbar^2\left[\frac{2(q_1^2+\beta_2)}{(q_1^2-\beta_2)^2}-\frac{1}{8q_1^2}\right]\)
\item[(18a)] \(\displaystyle \beta_1\left(q_1^2+4q_2^2+\frac{\beta_2^2}{q_1^2}\right)+\hbar^2\left[\frac{2(q_1^2+\beta_2)}{(q_1^2-\beta_2)^2}-\frac{1}{8q_1^2}\right]\)
\item[(19)] \(\displaystyle \beta_1\left(q_1^2+16q_2^2+\frac{\beta_2^2}{q_1^2}\right)+\hbar^2\left[\frac{2(q_1^2+\beta_2)}{(q_1^2-\beta_2)^2}-\frac{1}{8q_1^2}\right]\)
\item[(20)] \(\displaystyle \hbar^2\left[\frac{4q_1^2+q_2^2}{32\beta_1^2}+\frac{2(q_1^2+\beta_1)}{(q_1^2-\beta_1)^2}\right]+\frac{\beta_2}{q_2^2}\)
\item[(21)] \(\displaystyle\hbar^2\left[\frac{9(4q_1^2+q_2^2)}{32\beta_1^2}+\frac{2(q_1^2+\beta_1)}{(q_1^2-\beta_1)^2}+\frac{1}{q_1^2}\right]+\frac{\beta_2}{q_2^2}\)
\item[(22)] \(\displaystyle \hbar^2\left[\frac{4q_1^2+9q_2^2}{32\beta_1^2}+\frac{2(q_1^2+\beta_1)}{(q_1^2-\beta_1)^2}\right]+\frac{\beta_2}{q_2^2}\)
\end{enumerate}
as well as two degenerate cases
\begin{enumerate}
\begin{multicols}{2}
    \item[(\(\upalpha\))] \(V_2(q_2)\)
    \item[(\(\upbeta\))] \(\dfrac{\hbar^2}{q_1^2}+V_2(q_2)\)
    \end{multicols}
\end{enumerate}
which have a first- and third-order integral \cite{hiet89,gra02} respectively for \(V_2\) arbitrary. The former is not in general superintegrable. 

These models were found through a direct calculation from the compatibility equations. In the appendices, we have listed all of the integrals we have obtained up to fourth order, independent of the trivial ones (\(H_1,H_2\) and their products).
\section{Conclusion}
We have expounded a new method for finding Cartesian-separable superintegrable models and constructing their integrals. Instead of dealing with individual coefficients, we have reduced the number of equations and outlined a systematic algorithm for deriving each homogeneous component. The residue becomes the compatibility equations, recast into integral rather than differential equations. These compatibility equations have not been able to be derived using the ordinary methods in their full generality. While we can in theory calculate each compatibility equation successively, the computations become very involved. More work is needed in order to obtain general results about the compatibility equations of arbitrarily large order and in particular the determination of their auxiliary functions, which will yield more information about their rational potentials.

We have constrained ourselves to the case of Cartesian separability in order that the integration by parts can be achieved in finite steps. Whether it can be extended to other cases is a question for further investigation.

Our classification has added several new models to the catalogue. More detail into their symmetry algebra and relations will be given in an upcoming work.
\sloppy
\printbibliography[title={References}]
\newpage
\appendix
\section{List of Integrals}
We give here all of the integrals of order less than or equal to four up to linear independence for the Hamiltonians listed in Section 8.1. Throughout, we use \(\{a,b\}=\tfrac{1}{2}(a\circ b+b\circ a)\) for the symmetrised product. To avoid repetition, the trivial integrals \(H_1,H_2,\{H_1,H_1\},\{H_1,H_2\},\{H_2,H_2\}\) are omitted.
\begin{enumerate}
    \item[(1)] \(H=\frac{1}{2}(p_1^2+p_2^2)+\beta(q_1^2+q_2^2)\).
    \end{enumerate}
There is a single first-order integral \(m_{12}\). The second-order integrals are \(\{m_{12},m_{12}\}\), \((H_1-H_2,m_{12})\). The third-order integrals are 
\[\{m_{12},\{m_{12},m_{12}\}\},\{(H_1-H_2,m_{12}),m_{12}\},\{H_1,m_{12}\},\{H_2,m_{12}\}.\]
The fourth-order integrals are \begin{align*}
&\{\{m_{12},m_{12}\},\{m_{12},m_{12}\}\},\{H_1,\{m_{12},m_{12}\}\},\{H_2,\{m_{12},m_{12}\}\},\\
&\{H_1,(H_1-H_2,m_{12})\},\{H_2,(H_1-H_2,m_{12})\},\{(H_1-H_2,m_{12}),\{m_{12},m_{12}\}\}.
\end{align*}
\begin{enumerate}
    \item[(2)] \(\displaystyle H=\tfrac{1}{2}(p_1^2+p_2^2)+\beta_1(q_1^2+q_2^2)+\frac{\beta_2}{q_1^2}+\frac{\beta_3}{q_2^2}\)
\end{enumerate}
There is one second-order integral
\[X=\{m_{12},m_{12}\}+\tfrac{1}{2}\hbar^2+\frac{2\beta_2q_2^2}{q_1^2}+\frac{2\beta_3q_1^2}{q_2^2}\]
There is one third-order integral \((H_1-H_2,X)\) and three fourth-order integrals 
\[\{X,X\},\{H_1,X\},\{H_2,X\}.\]
\begin{enumerate}
\item[(2a)] \(\displaystyle H=\tfrac{1}{2}(p_1^2+p_2^2)+\frac{\beta}{q_1^2}\)
\end{enumerate}
This is a special case of (2) with \(\beta_1=0,\beta_2=\beta,\beta_3=0\). It has a first-order integral \(p_2\). There is an additional second-order integral \((p_2,X^{(1)})\). There are three more third-order integrals
\[p_2^3,\{X^{(1)},p_2\},\{H_1,p_2\}\]
There are three more fourth-order integrals
\[\{X,(p_2,X)\},\{(p_2,X),H_1\},\{(p_2,X),H_2\}\]
\begin{enumerate}
\item[(2ai)] \(\displaystyle H=\tfrac{1}{2}(p_1^2+p_2^2)+\frac{\hbar^2}{q_1^2}\)
\end{enumerate}
This is a special case of (2a) with \(\beta=\hbar^2\). It has four additional third-order integrals
\begin{align*}
    X^{(1)}&=\{m_{12},\{m_{12},m_{12}\}\}-\hbar^2q_2\left\{p_1,2+\frac{3q_2^2}{q_1^2}\right\}+\frac{\hbar^2}{q_1}\left\{p_2,2q_1^2+3q_2^2\right\}
\end{align*}
and \((p_2,X^{(1)}),(p_2,(p_2,X^{(1)})),(p_2,(p_2,(p_2,X^{(1)})))\). There are four more fourth-order integrals 
\[\{p_2,X^{(1)}\},\{p_2,(p_2,X^{(1)})\},\{p_2,(p_2,(p_2,X^{(1)}))\},\{p_2,(p_2,(p_2,(p_2,X^{(1)})))\}\]
\begin{enumerate}
\item[(2b)] \(\displaystyle H=\tfrac{1}{2}(p_1^2+p_2^2)+\frac{\beta}{q_1^2}+\frac{\hbar^2}{q_2^2}\)
\end{enumerate}
This is a special case of (2) with \(\beta_1=0,\beta_2=\beta,\beta_3=\hbar^2\). It has two additional third-order integrals
\begin{align*}
    X^{(1)}&=\{p_2,\{m_{12},m_{12}\}\}-\frac{2\hbar^2}{q_2}\{p_1,q_1\}+\left\{p_2,\frac{\hbar^2}{2}+\frac{3\hbar^2q_1^2}{q_2^2}+\frac{2\beta q_2^2}{q_1^2}\right\}\\
    X^{(2)}&=p_2^3+3\hbar^2\left\{p_2,\frac{1}{q_2^2}\right\}
\end{align*}
There are two more fourth-order integrals \((X,X^{(1)}),(X,X^{(2)})\).
\begin{enumerate}
\item[(2bi)] \(\displaystyle H=\tfrac{1}{2}(p_1^2+p_2^2)+\frac{\hbar^2}{q_1^2}+\frac{\hbar^2}{q_2^2}\)
\end{enumerate}
This is a special case of (2b) with \(\beta=\hbar^2\). There are three more third-order integrals
\begin{align*}
    X^{(3)}&=\{m_{12},\{m_{12},m_{12}\}\}-\frac{\hbar^2}{q_2}\left\{p_1,3q_1^2+2q_2^2+\frac{3q_2^4}{q_1^2}\right\}+\frac{\hbar^2}{q_1}\left\{p_2,3q_2^2+2q_1^2+\frac{3q_1^4}{q_2^2}\right\}\\
X^{(4)}&=\{p_1,\{m_{12},m_{12}\}\}+\hbar^2\left\{p_1,\frac{1}{2}+\frac{3\hbar^2q_2^2}{q_1^2}+\frac{2q_1^2}{q_2^2}\right\}-\frac{2\hbar^2}{q_1}\{p_2,q_2\}\\
    X^{(5)}&=p_1^3+3\hbar^2\left\{p_1,\frac{1}{q_1^2}\right\}
\end{align*}
There are consequently three more fourth-order integrals \((X,X^{(4)}),(X,X^{(5)}),(H_1-H_2,X^{(3)})\).
\begin{enumerate}
\item[(2c)] \(\displaystyle H=\tfrac{1}{2}(p_1^2+p_2^2)+\beta(q_1^2+q_2^2)+\frac{\hbar^2}{q_1^2}\)
\end{enumerate}
This is a special case of (2) with \(\beta_1=\beta,\beta_2=\hbar^2,\beta_3=0\). There are two more third-order integrals
\begin{align*}
    X^{(1)}&=\{m_{12},\{m_{12},m_{12}\}\}-\hbar^2q_2\left\{p_1,2+\frac{3q_2^2}{q_1^2}\right\}+\frac{\hbar^2}{q_1}\left\{p_2,2q_1^2+3q_2^2\right\}\\
    X^{(2)}&=\{p_1^2,m_{12}\}-q_2\left\{p_1,2\beta q_1^2+\frac{3\hbar^2}{q_1^2}\right\}+\left(2\beta q_1^3+\frac{\hbar^2}{q_1}\right)p_2
\end{align*}
and two more fourth-order integrals \((H_1-H_2,X^{(1)}),(H_1-H_2,X^{(2)})\).
\begin{enumerate}
    \item[(2d)] \(\displaystyle H=\tfrac{1}{2}(p_1^2+p_2^2)+\beta(q_1^2+q_2^2)+\frac{\hbar^2}{q_1^2}+\frac{\hbar^2}{q_2^2}\)
    \end{enumerate}
This is a special case of (2) with \(\beta_1=\beta,\beta_2=\hbar^2,\beta_3=\hbar^2\). There is one more third-order integral
\[X^{(1)}=\{m_{12},\{m_{12},m_{12}\}\}-\frac{\hbar^2}{q_2}\left\{p_1,3q_1^2+2q_2^2+\frac{3q_2^4}{q_1^2}\right\}+\frac{\hbar^2}{q_1}\left\{p_2,2q_1^2+3q_2^2+\frac{3q_1^4}{q_2^2}\right\}\]
and one more fourth-order integral \((H_1-H_2,X^{(1)})\).
\begin{enumerate}
    \item[(3)] \(\displaystyle H=\tfrac{1}{2}(p_1^2+p_2^2)+\beta_1(q_1^2+4q_2^2)+\frac{\beta_2}{q_1^2}+\beta_3q_2\)
\end{enumerate}
There is one second-order integral
\[X=\{p_1,m_{12}\}-\frac{2\beta_2q_2}{q_1^2}+\tfrac{1}{2}q_1^2(4q_2\beta_1+\beta_3)\]
There is one third-order integral \((H_1-H_2,X)\) and three fourth-order integrals 
\[\{X,X\},\{H_1,X\},\{H_2,X\}.\]
\begin{enumerate}
    \item[(3a)] \(\displaystyle H=\tfrac{1}{2}(p_1^2+p_2^2)+\beta q_2\)
\end{enumerate}
This is a special case of (3) with \(\beta_1=\beta_2=0,\beta_3=\beta\). It has a first-order integral \(p_1\). It has another second-order integral \((p_1,X)\). There are three more third-order integrals are fourfold
\[p_1^3,\{X,p_1\},\{(p_1,X),p_1\}.\]
There are three more fourth-order integrals
\[\{(p_1,X),X\},\{(p_1,X),H_1\},\{(p_1,X),H_2\}\]
\begin{enumerate}
    \item[(3b)] \(H=\tfrac{1}{2}(p_1^2+p_2^2)+\dfrac{\hbar^2}{q_1^2}+\beta q_2\)
    \end{enumerate}
This is a special case of (3) with \(\beta_1=0,\beta_2=\hbar^2,\beta_3=\beta\). There are two more third-order integrals
\begin{align*}
X^{(1)}&=\{p_1^2,m_{12}\}+\left\{p_1,\tfrac{1}{2}\beta q_1^2-\frac{3\hbar^2q_2}{q_1^2}\right\}+\frac{\hbar^2p_2}{q_1}\\
X^{(2)}&=p_1^3+3\hbar^2\left\{p_1,\frac{1}{q_1^2}\right\}
\end{align*}
and two more fourth-order integrals \((H_1-H_2,X^{(1)}),(X,X^{(1)})\).
   \begin{enumerate}
\item[(4)] \(\displaystyle H=\tfrac{1}{2}(p_1^2+p_2^2)+\beta_1\left(q_1^2+q_2^2+\frac{\beta_2^2}{q_1^2}\right)+\frac{\beta_3}{q_2^2}+\hbar^2\left[\frac{2(q_1^2+\beta_2)}{(q_1^2-\beta_2)^2}-\frac{1}{8q_1^2}\right]\)
\end{enumerate}
There are two fourth-order integrals
\[\begin{aligned}
    X^{(1)}&=\{\{m_{12},m_{12}\},\{m_{12},m_{12}\}-2\beta_2p_2^2\}+\left\{p_1^2,-\beta _2 \hbar ^2+\frac{16 \beta _2 q_2^4 \hbar ^2}{\left(q_1^2-\beta _2\right){}^2}+\frac{8 q_2^4 \hbar ^2}{q_1^2-\beta _2}\right.\\
    &\qquad\left.+\frac{8 \beta _1 \beta _2^2 q_2^4-q_2^4 \hbar ^2}{2
   q_1^2}-4 \beta _1 \beta _2 q_2^4+4 \beta _3 q_1^2+5 q_2^2 \hbar ^2\right\}+\left\{p_1p_2,-\frac{48 \beta _2 q_2^3 q_1 \hbar ^2}{\left(q_1^2-\beta _2\right){}^2}\right.\\
   &\qquad\left.+\frac{2 q_1 \left(4 \beta _2 \beta _3+4 \beta _1 \beta _2 q_2^4-5 q_2^2 \hbar ^2\right)}{q_2}-\frac{16 q_2^3
   q_1 \hbar ^2}{q_1^2-\beta _2}+\frac{q_2^3 \hbar ^2-8 \beta _1 \beta _2^2 q_2^3}{q_1}-\frac{8 \beta _3 q_1^3}{q_2}\right\}\\
   &\qquad+\left\{p_2^2,-\frac{q_1^2 \left(8 \beta _2 \beta _3+4 \beta _1 \beta _2 q_2^4-5 q_2^2 \hbar ^2\right)}{q_2^2}+\frac{32 \beta _2^2 q_2^2 \hbar ^2}{\left(q_1^2-\beta _2\right){}^2}+\frac{40 \beta _2
   q_2^2 \hbar ^2}{q_1^2-\beta _2}\right.\\
   &\qquad\left.+\frac{1}{2} \left(-2 \beta _2 \hbar ^2+8 \beta _1 \beta _2^2 q_2^2+15 q_2^2 \hbar ^2\right)+\frac{\beta _2 q_2^2 \hbar ^2-8 \beta _1 \beta _2^3
   q_2^2}{2 q_1^2}+\frac{4 \beta _3 q_1^4}{q_2^2}\right\}\\
   &\qquad+\frac{64 \beta _2^2 q_2^4 \hbar ^4}{\left(q_1^2-\beta _2\right){}^4}+\frac{16 q_2^2 \left(9 \beta _2^2 \hbar ^4+4 \beta _2 q_2^2 \hbar ^4\right)}{\left(q_1^2-\beta
   _2\right){}^3}-\frac{2 q_1^2 \left(4 \beta _2 \beta _3^2+\beta _1 \beta _2 q_2^4 \hbar ^2-3 \beta _3 q_2^2 \hbar ^2\right)}{q_2^4}\\
   &\qquad+\frac{\hbar ^2 \left(-8 \beta _2 \beta _3+24
   \beta _1 \beta _2 q_2^4+8 \beta _1 \beta _2^2 q_2^2+8 \beta _3 q_2^2+q_2^2 \hbar ^2\right)}{4 q_2^2}\\
   &\qquad+\frac{4 \left(4 \beta _2^2 \hbar ^4+8 \beta _2^2 \beta _3 \hbar ^2+33 \beta _2
   q_2^2 \hbar ^4+3 q_2^4 \hbar ^4\right)}{\left(q_1^2-\beta _2\right){}^2}\\
   &\qquad+\frac{2 \left(10 \beta _2^2 \hbar ^4+16 \beta _2^2 \beta _3 \hbar ^2+6 \beta _2 q_2^2 \hbar ^4-8 \beta _1
   \beta _2^2 q_2^4 \hbar ^2+q_2^4 \hbar ^4\right)}{\beta _2 \left(q_1^2-\beta _2\right)}\\
   &\qquad+\frac{\beta _2^2 \hbar ^4-8 \beta _1 \beta _2^4 \hbar ^2-3 \beta _2 q_2^2 \hbar ^4+24 \beta
   _1 \beta _2^3 q_2^2 \hbar ^2+68 \beta _1 \beta _2^2 q_2^4 \hbar ^2-32 \beta _1^2 \beta _2^4 q_2^4-8 q_2^4 \hbar ^4}{4 \beta _2 q_1^2}\\
   &\qquad+\frac{-16 \beta _1 \beta _2^2 q_2^4 \hbar
   ^2+64 \beta _1^2 \beta _2^4 q_2^4+q_2^4 \hbar ^4}{16 q_1^4}+\frac{4 \beta _3^2 q_1^4}{q_2^4}
\end{aligned}\]
\[\begin{aligned}X^{(2)}&=\{p_1^2,\{m_{12},m_{12}\}\}+\left\{p_1^2,\frac{16 \beta _2 q_2^2 \hbar ^2}{\left(q_1^2-\beta _2\right){}^2}+\frac{1}{2} \left(\hbar ^2-8 \beta _1 \beta _2 q_2^2\right)+\frac{8 q_2^2 \hbar ^2}{q_1^2-\beta _2}\right.\\
&\qquad\left.+\frac{8 \beta _1
   \beta _2^2 q_2^2-q_2^2 \hbar ^2}{2 q_1^2}+\frac{2 q_1^2 \left(\beta _3+\beta _1 q_2^4\right)}{q_2^2}\right\}+\left\{p_1p_2,-\frac{24 \beta _2 q_2 q_1 \hbar ^2}{\left(q_1^2-\beta _2\right){}^2}\right.\\
   &\qquad\left.-\frac{8 q_2 q_1 \hbar ^2}{q_1^2-\beta _2}+\frac{q_2 \hbar ^2-8 \beta _1 \beta _2^2 q_2}{2 q_1}-4 \beta _1 q_2
   q_1^3+8 \beta _1 \beta _2 q_2 q_1\right\}\\
   &\qquad+\left\{p_2^2,\frac{8 \beta _2^2 \hbar ^2}{\left(q_1^2-\beta _2\right){}^2}+\frac{8 \beta _2 \hbar ^2}{q_1^2-\beta _2}+2 \beta _1 q_1^4-4 \beta _1 \beta _2 q_1^2+\frac{\hbar ^2}{2}\right\}\\
   &\qquad+\frac{64 \beta _2^2 q_2^2 \hbar ^4}{\left(q_1^2-\beta _2\right){}^4}+\frac{8 \left(3 \beta _2^2 \hbar ^4+8 \beta _2 q_2^2 \hbar ^4\right)}{\left(q_1^2-\beta _2\right){}^3}+\frac{\hbar
   ^2 \left(\beta _3+\beta _1 q_2^4+2 \beta _1 \beta _2 q_2^2\right)}{q_2^2}\\
   &\qquad-\frac{\beta _1 q_1^2 \left(8 \beta _2 \beta _3+3 q_2^2 \hbar ^2\right)}{q_2^2}+\frac{-\beta _2 \hbar ^4+8
   \beta _1 \beta _2^3 \hbar ^2+136 \beta _1 \beta _2^2 q_2^2 \hbar ^2-64 \beta _1^2 \beta _2^4 q_2^2-16 q_2^2 \hbar ^4}{8 \beta _2 q_1^2}\\
   &\qquad+\frac{-16 \beta _1 \beta _2^2 q_2^2 \hbar
   ^2+64 \beta _1^2 \beta _2^4 q_2^2+q_2^2 \hbar ^4}{16 q_1^4}+\frac{2 \left(8 \beta _2^2 \beta _3 \hbar ^2+\beta _2 q_2^2 \hbar ^4+q_2^4 \hbar ^4\right)}{\beta _2 q_2^2
   \left(q_1^2-\beta _2\right)}\\
   &\qquad+\frac{2 \left(8 \beta _2^2 \beta _3 \hbar ^2+11 \beta _2 q_2^2 \hbar ^4+8 \beta _1 \beta _2^2 q_2^4 \hbar ^2+6 q_2^4 \hbar ^4\right)}{q_2^2
   \left(q_1^2-\beta _2\right){}^2}+\frac{4 \beta _1 \beta _3 q_1^4}{q_2^2}\end{aligned}\]
\begin{enumerate}
    \item[(4a)] \(\displaystyle H=\tfrac{1}{2}(p_1^2+p_2^2)+\hbar^2\left[\frac{q_1^2+q_2^2}{8\beta^2}+\frac{2(q_1^2+\beta)}{(q_1^2-\beta)^2}\right]\)
    \end{enumerate}
This is a special case of (4) where \(\beta_1=\dfrac{\hbar^2}{8\beta^2},\beta_2=\beta,\beta_3=0\). There are two more third-order integrals
\[\begin{aligned}
    X^{(3)}&=\{m_{12},\{m_{12},m_{12}\}-3\beta p_2^2\}+\hbar^2\left\{p_1,-\frac{6 q_2^3}{q_1^2-\beta }-\frac{12 \beta  q_2^3}{\left(q_1^2-\beta \right){}^2}-\frac{q_2 \left(8 \beta -3 q_2^2\right)}{4 \beta }\right\}\\
    &\qquad+\hbar^2\left\{p_2,\frac{6 q_1 q_2^2}{q_1^2-\beta }+\frac{24 \beta  q_1 q_2^2}{\left(q_1^2-\beta \right){}^2}+\frac{q_1 \left(8 \beta -3 q_2^2\right)}{4 \beta }\right\}
\end{aligned}\]
\[\begin{aligned}
    X^{(4)}&=\{p_1^2,m_{12}\}+\hbar^2\left\{p_1,-\frac{q_2 q_1^2}{4 \beta ^2}+\frac{q_2}{\beta }-\frac{6 q_2}{q_1^2-\beta }-\frac{12 \beta  q_2}{\left(q_1^2-\beta \right){}^2}\right\}\\
    &\qquad+\hbar^2\left\{p_2,\frac{q_1^3}{4 \beta ^2}-\frac{q_1}{\beta }+\frac{2 q_1}{q_1^2-\beta }+\frac{8 \beta  q_1}{\left(q_1^2-\beta \right){}^2}\right\}
\end{aligned}\]
and two more fourth-order integrals \((H_1-H_2,X^{(3)}),(H_1-H_2,X^{(4)})\).
\begin{enumerate}
    \item[(4b)] \(\displaystyle H=\tfrac{1}{2}(p_1^2+p_2^2)+\hbar^2\left[\frac{q_1^2+q_2^2}{8\beta^2}+\frac{2(q_1^2+\beta)}{(q_1^2-\beta)^2}+\frac{1}{q_2^2}\right]\)
    \end{enumerate}
This is a special case of (4) where \(\beta_1=\dfrac{\hbar^2}{8\beta^2},\beta_2=\beta,\beta_3=\hbar^2\). There is a third-order integral 
\[\begin{aligned}
    X^{(3)}&=\{m_{12},\{m_{12},m_{12}\}-3\beta p_2^2\}+\hbar^2\left\{p_1,\frac{12 \beta ^2-8 \beta  q_2^2+3 q_2^4}{4 \beta  q_2}-\frac{6 q_2^3}{q_1^2-\beta }-\frac{12 \beta  q_2^3}{\left(q_1^2-\beta \right){}^2}-\frac{3 q_1^2}{q_2}\right\}\\
    &\qquad+\hbar^2\left\{p_2,-\frac{q_1 \left(36 \beta ^2-8 \beta  q_2^2+3 q_2^4\right)}{4 \beta  q_2^2}+\frac{6 q_2^2 q_1}{q_1^2-\beta }+\frac{24 \beta  q_2^2 q_1}{\left(q_1^2-\beta \right){}^2}+\frac{3
   q_1^3}{q_2^2}\right\}
\end{aligned}\]
and one more fourth-order integral \((H_1-H_2,X^{(3)})\)
\begin{enumerate}
    \item[(4c)] \(\displaystyle H=\tfrac{1}{2}(p_1^2+p_2^2)+\hbar^2\left[\frac{2(q_1^2+\beta)}{(q_1^2-\beta)^2}-\frac{1}{8q_1^2}\right]\)
\end{enumerate}
This is a special case of (4) where \(\beta_1=0,\beta_2=\beta,\beta_3=0\). There is a first-order integral \(p_2\), a third-order integral \(p_1^3\) and two additional fourth-order integrals \((p_2,X^{(1)}),(p_2,(p_2,(p_2,X^{(1)})))\).
\begin{enumerate}
\item[(4d)] \(\displaystyle H=\tfrac{1}{2}(p_1^2+p_2^2)+\hbar^2\left[\frac{2(q_1^2+\beta)}{(q_1^2-\beta)^2}-\frac{1}{8q_1^2}+\frac{1}{q_2^2}\right]\)
\end{enumerate}
This is a special case of (4) where \(\beta_1=0,\beta_2=\beta,\beta_3=\hbar^2\). There is a third-order integral
\[X^{(3)}=p_2^3+3\hbar^2\left\{p_2,\frac{1}{q_2^2}\right\}\]
and one more fourth-order integral
\[\begin{aligned}
    X^{(4)}&=\{p_1,\{m_{12},\{m_{12},m_{12}\}-\beta p_2^2\}\}+\hbar^2\left\{p_1^2,-\frac{8 q_2^3}{q_1^2-\beta }-\frac{16 \beta  q_2^3}{\left(q_1^2-\beta \right){}^2}\right.\\
    &\qquad\left.-\frac{2 q_2^2-\beta }{q_2}+\frac{q_2^3}{2 q_1^2}-\frac{3 q_1^2}{q_2}\right\}+\hbar^2\left\{p_1p_2,\frac{12 q_2^2 q_1}{q_1^2-\beta }+\frac{36 \beta  q_2^2 q_1}{\left(q_1^2-\beta \right){}^2}\right.\\
    &\qquad\left.-\frac{q_1 \left(3 \beta -2 q_2^2\right)}{q_2^2}+\frac{3 q_1^3}{q_2^2}-\frac{3 q_2^2}{4 q_1}\right\}+\hbar^2\left\{p_2^2,-\frac{16 \beta ^2 q_2}{\left(q_1^2-\beta \right){}^2}\right.\\
    &\qquad\left.-\frac{20 \beta  q_2}{q_1^2-\beta }-\frac{\beta  q_2}{4 q_1^2}-\frac{15 q_2}{4}\right\}+\hbar^4\left[-\frac{64 \beta ^2 q_2^3}{\left(q_1^2-\beta \right){}^4}-\frac{8 q_2 \left(9 \beta ^2+8 \beta  q_2^2\right)}{\left(q_1^2-\beta \right){}^3}\right.\\
    &\qquad-\frac{2 \left(10 \beta ^2+3 \beta 
   q_2^2+q_2^4\right)}{\beta  q_2 \left(q_1^2-\beta \right)}-\frac{2 \left(8 \beta ^2+33 \beta  q_2^2+6 q_2^4\right)}{q_2 \left(q_1^2-\beta \right){}^2}+\frac{-2 \beta ^2+3 \beta 
   q_2^2+16 q_2^4}{8 \beta  q_1^2 q_2}\\
   &\qquad\left.+\frac{3 \left(2 \beta -7 q_2^2\right)}{4 q_2^3}-\frac{q_2^3}{16 q_1^4}-\frac{9 q_1^2}{2 q_2^3}\right]
\end{aligned}\]
\begin{enumerate}
    \item[(4e)] \(\displaystyle H=\tfrac{1}{2}(p_1^2+p_2^2)+\hbar^2\left[\frac{9(q_1^2+q_2^2)}{8\beta^2}+\frac{2(q_1^2+\beta)}{(q_1^2-\beta)^2}+\frac{1}{q_1^2}\right]\)
\end{enumerate}
This is a special case of (4) where \(\beta_1=\dfrac{9\hbar^2}{8\beta^2},\beta_2=\beta,\beta_3=0\). There is one third-order integral
\begin{align*}
X^{(3)}&=\{p_1^2,m_{12}\}+\hbar^2q_2\left\{p_1,\frac{6}{\beta}-\frac{3}{q_1^2}-\frac{9q_1^2}{4\beta^2}-\frac{6(q_1^2+\beta)}{(q_1^2-\beta)^2}\right\}\\
&\qquad+\hbar^2p_2\left[\frac{1}{q_1}-\frac{6q_1}{\beta}+\frac{9q_1^3}{4\beta^2}+\frac{2q_1(q_1^2+3\beta)}{(q_1^2-\beta)^2}\right]
\end{align*}
and one more fourth-order integral \((H_1-H_2,X^{(3)})\).
\begin{enumerate}
\item[(5)] \(\displaystyle H=\tfrac{1}{2}(p_1^2+p_2^2)+\beta_1\left(q_1^2+q_2^2+\frac{\beta_2^2}{q_1^2}+\frac{\beta_3^2}{q_2^2}\right)+\hbar^2\left[\frac{2(q_1^2+\beta_2)}{(q_1^2-\beta_2)^2}-\frac{1}{8q_1^2}+\frac{2(q_2^2+\beta_3)}{(q_2^2-\beta_3)^2}-\frac{1}{8q_2^2}\right]\)
\end{enumerate}
There is one fourth-order integral\footnotesize
\[\begin{aligned}
    X&=\{\{m_{12},m_{12}\},\{m_{12},m_{12}\}-2\beta_2p_2^2-2\beta_3p_1^2\}+\left\{p_1^2,\frac{8 q_2^2 \left(q_2^2 \hbar ^2-2 \beta _3 \hbar ^2\right)}{q_1^2-\beta _2}+\frac{16 q_2^2 \left(\beta _2 q_2^2 \hbar ^2-2 \beta _2 \beta _3 \hbar ^2\right)}{\left(q_1^2-\beta
   _2\right){}^2}\right.\\
   &\qquad+q_2^2 \left(8 \beta _1 \beta _2 \beta _3-4 \beta _1 \beta _3 q_1^2+5 \hbar ^2\right)+\frac{8 \left(5 \beta _3 q_1^2 \hbar ^2-2 \beta _2 \beta _3 \hbar
   ^2\right)}{q_2^2-\beta _3}+\frac{2 \beta _3 q_2^2 \hbar ^2+8 \beta _1 \beta _2^2 q_2^4-16 \beta _1 \beta _2^2 \beta _3 q_2^2+q_2^4 \left(-\hbar ^2\right)}{2 q_1^2}\\
   &\qquad\left.+\frac{1}{2}
   \left(-2 \beta _2 \hbar ^2-2 \beta _3 \hbar ^2+8 \beta _1 \beta _3^2 q_1^2+15 q_1^2 \hbar ^2\right)+\frac{16 \left(2 \beta _3^2 q_1^2 \hbar ^2-\beta _2 \beta _3^2 \hbar
   ^2\right)}{\left(q_2^2-\beta _3\right){}^2}+\frac{\beta _3 q_1^2 \hbar ^2-8 \beta _1 \beta _3^3 q_1^2}{2 q_2^2}-4 \beta _1 \beta _2 q_2^4\right\}\\
   &\qquad+\left\{p_1p_2,-2 q_1 q_2 \left(8 \beta _1 \beta _2 \beta _3-4 \beta _1 \beta _3 q_1^2+5 \hbar ^2\right)-\frac{16 \left(q_1 q_2^3 \hbar ^2-\beta _3 q_1 q_2 \hbar ^2\right)}{q_1^2-\beta _2}-\frac{48
   \left(\beta _2 q_1 q_2^3 \hbar ^2-\beta _2 \beta _3 q_1 q_2 \hbar ^2\right)}{\left(q_1^2-\beta _2\right){}^2}\right.\\
   &\qquad-\frac{48 \left(\beta _3 q_1^3 q_2 \hbar ^2-\beta _2 \beta _3 q_1 q_2
   \hbar ^2\right)}{\left(q_2^2-\beta _3\right){}^2}+\frac{-\beta _3 q_2 \hbar ^2-8 \beta _1 \beta _2^2 q_2^3+8 \beta _1 \beta _2^2 \beta _3 q_2+q_2^3 \hbar ^2}{q_1}-\frac{16
   \left(q_1^3 q_2 \hbar ^2-\beta _2 q_1 q_2 \hbar ^2\right)}{q_2^2-\beta _3}\\
   &\qquad\left.+\frac{-\beta _2 q_1 \hbar ^2-8 \beta _1 \beta _3^2 q_1^3+8 \beta _1 \beta _2 \beta _3^2 q_1+q_1^3 \hbar
   ^2}{q_2}+8 \beta _1 \beta _2 q_1 q_2^3\right\}+\left\{p_2^2,-\beta _2 \hbar ^2-\beta _3 \hbar ^2\right.\\
   &\qquad+\frac{1}{2} q_2^2 \left(8 \beta _1 \beta _2^2-8 \beta _1 \beta _2 q_1^2+15 \hbar ^2\right)+\frac{8 \left(5 \beta _2 q_2^2 \hbar ^2-2 \beta _2 \beta
   _3 \hbar ^2\right)}{q_1^2-\beta _2}+\frac{16 \left(\beta _3 q_1^4 \hbar ^2-2 \beta _2 \beta _3 q_1^2 \hbar ^2\right)}{\left(q_2^2-\beta _3\right){}^2}\\
   &\qquad+\frac{16 \left(2 \beta _2^2
   q_2^2 \hbar ^2-\beta _2^2 \beta _3 \hbar ^2\right)}{\left(q_1^2-\beta _2\right){}^2}+\frac{2 \beta _2 q_1^2 \hbar ^2+8 \beta _1 \beta _3^2 q_1^4-16 \beta _1 \beta _2 \beta _3^2
   q_1^2+q_1^4 \left(-\hbar ^2\right)}{2 q_2^2}+\frac{8 \left(q_1^4 \hbar ^2-2 \beta _2 q_1^2 \hbar ^2\right)}{q_2^2-\beta _3}\\
   &\qquad\left.+\frac{\beta _2 q_2^2 \hbar ^2-8 \beta _1 \beta _2^3
   q_2^2}{2 q_1^2}-4 \beta _1 \beta _3 q_1^4+8 \beta _1 \beta _2 \beta _3 q_1^2+5 q_1^2 \hbar ^2\right\}+\frac{2475 \beta _2 \beta _3^2 \hbar ^4}{4 \left(q_1^2-\beta _2\right){}^3}+q_2^4 \left(\frac{4 \beta _1^2 \beta _2^6}{\left(q_1^2-\beta _2\right){}^4}\right.\\
   &\qquad-\frac{8 \beta _1^2 \beta
   _2^5}{\left(q_1^2-\beta _2\right){}^3}+\frac{4 \beta _1^2 \beta _2^4}{q_1^4}+\frac{12 \beta _1^2 \beta _2^4}{\left(q_1^2-\beta _2\right){}^2}-\frac{\hbar ^2 \beta _1 \beta
   _2^4}{\left(q_1^2-\beta _2\right){}^4}-\frac{16 \beta _1^2 \beta _2^3}{q_1^2-\beta _2}+\frac{2 \hbar ^2 \beta _1 \beta _2^3}{\left(q_1^2-\beta _2\right){}^3}-\frac{\hbar ^2 \beta _1
   \beta _2^2}{q_1^4}\\
   &\qquad-\frac{3 \hbar ^2 \beta _1 \beta _2^2}{\left(q_1^2-\beta _2\right){}^2}+\frac{64 \hbar ^4 \beta _2^2}{\left(q_1^2-\beta _2\right){}^4}+\frac{4 \hbar ^2 \beta _1
   \beta _2}{q_1^2-\beta _2}+\frac{64 \hbar ^4 \beta _2}{\left(q_1^2-\beta _2\right){}^3}+\frac{17 \hbar ^2 \beta _1 \beta _2-8 \beta _1^2 \beta _2^3}{q_1^2}+\frac{3 \left(8 \beta _1^2
   \beta _2^3-7 \hbar ^2 \beta _1 \beta _2\right)}{q_1^2-\beta _2}\\
   &\qquad-\frac{8 \beta _1^2 \beta _2^3-\hbar ^2 \beta _1 \beta _2}{q_1^2-\beta _2}+\frac{6 \left(2 \beta _1^2 \beta _2^4-3
   \hbar ^2 \beta _1 \beta _2^2\right)}{\left(q_1^2-\beta _2\right){}^2}-\frac{3 \left(8 \beta _1^2 \beta _2^4-7 \hbar ^2 \beta _1 \beta _2^2\right)}{\left(q_1^2-\beta
   _2\right){}^2}+\frac{3 \left(8 \beta _1^2 \beta _2^5-7 \hbar ^2 \beta _1 \beta _2^3\right)}{\left(q_1^2-\beta _2\right){}^3}\\
   &\qquad-\frac{16 \beta _1^2 \beta _2^5-19 \hbar ^2 \beta _1
   \beta _2^3}{\left(q_1^2-\beta _2\right){}^3}-\frac{3 \left(8 \beta _1^2 \beta _2^6-7 \hbar ^2 \beta _1 \beta _2^4\right)}{\left(q_1^2-\beta _2\right){}^4}-\frac{2 \left(18 \beta
   _1^2 \beta _2^6-19 \hbar ^2 \beta _1 \beta _2^4\right)}{\left(q_1^2-\beta _2\right){}^4}+\frac{2 \left(28 \beta _1^2 \beta _2^6-29 \hbar ^2 \beta _1 \beta
   _2^4\right)}{\left(q_1^2-\beta _2\right){}^4}\\
   &\qquad\left.+\frac{12 \hbar ^4}{\left(q_1^2-\beta _2\right){}^2}+\frac{\hbar ^4}{16 q_1^4}+\frac{2 \hbar ^4}{\left(q_1^2-\beta _2\right) \beta
   _2}-\frac{2 \hbar ^4}{q_1^2 \beta _2}\right)+\frac{3 \left(107 \hbar ^4 \beta _2-5 \hbar ^4 \beta _3\right)}{4 \left(q_1^2-\beta _2\right)}+\frac{-\beta _2 \hbar ^4-\beta _3 \hbar
   ^4}{4 \left(q_1^2-\beta _2\right)}\\
   &\qquad+\frac{\beta _2 \hbar ^4+\beta _3 \hbar ^4}{4 q_1^2}+\frac{3 \left(101 \hbar ^4 \beta _2^2-59 \hbar ^4 \beta _2 \beta _3\right)}{4
   \left(q_1^2-\beta _2\right){}^2}+\frac{2 \left(\hbar ^2 \beta _1 \beta _2^3+\hbar ^2 \beta _1 \beta _3 \beta _2^2\right)}{q_1^2-\beta _2}-\frac{2 \left(\hbar ^2 \beta _1 \beta
   _2^3+\hbar ^2 \beta _1 \beta _3 \beta _2^2\right)}{q_1^2}\\
   &\qquad+\frac{3 \hbar ^4 \beta _2^4-29 \hbar ^4 \beta _2^3 \beta _3}{4 \left(q_1^2-\beta _2\right){}^4}-\frac{2 \left(\hbar ^2
   \beta _1 \beta _2^5+\hbar ^2 \beta _1 \beta _3 \beta _2^4\right)}{\left(q_1^2-\beta _2\right){}^3}+\frac{\beta _2^2 \hbar ^4-72 \beta _3^2 \hbar ^4+\beta _2 \beta _3 \hbar ^4}{4
   \left(q_1^2-\beta _2\right){}^2}\\
   &\qquad-\frac{2 \left(-9 \beta _3^2 \hbar ^4+\beta _1 \beta _2^4 \hbar ^2+\beta _1 \beta _2^3 \beta _3 \hbar ^2\right)}{\left(q_1^2-\beta _2\right){}^2}+2
   \left(2 \hbar ^4+\beta _1 \beta _2^2 \hbar ^2+\beta _1 \beta _3^2 \hbar ^2-q_1^2 \beta _1 \beta _2 \hbar ^2+3 q_1^2 \beta _1 \beta _3 \hbar ^2-2 \beta _1 \beta _2 \beta _3 \hbar
   ^2\right)\\
   &\qquad+\frac{\beta _2^3 \hbar ^4-2475 \beta _2 \beta _3^2 \hbar ^4-191 \beta _2^2 \beta _3 \hbar ^4}{4 \left(q_1^2-\beta _2\right){}^3}+\frac{2 \left(\hbar ^2 \beta _1 \beta
   _2^5+\hbar ^2 \beta _1 \beta _3 \beta _2^4-9 \hbar ^4 \beta _3^2 \beta _2\right)}{\left(q_1^2-\beta _2\right){}^3}+\frac{-\beta _2^3 \hbar ^4+72 \beta _2 \beta _3^2 \hbar ^4-\beta _2^2
   \beta _3 \hbar ^4}{4 \left(q_1^2-\beta _2\right){}^3}\\
   &\qquad+\frac{16 \left(9 q_1^2 \beta _3^2 \hbar ^4-3 \beta _2 \beta _3^2
   \hbar ^4+4 q_1^4 \beta _3 \hbar ^4-8 q_1^2 \beta _2 \beta _3 \hbar ^4\right)}{\left(q_2^2-\beta _3\right){}^3}-\frac{2 \left(\hbar ^2 \beta _1 \beta _2^3+\hbar ^2 \beta _1 \beta _3 \beta _2^2-16 \hbar ^2 \beta _1 \beta _3^2 \beta
   _2\right)}{q_1^2-\beta _2}\\
   &\qquad+\frac{\beta _2^4 \hbar ^4-72 \beta _2^2 \beta _3^2 \hbar ^4+\beta _2^3 \beta _3 \hbar ^4}{4 \left(q_1^2-\beta _2\right){}^4}-\frac{2 \left(\hbar ^2 \beta
   _1 \beta _2^6+\hbar ^2 \beta _1 \beta _3 \beta _2^5-9 \hbar ^4 \beta _3^2 \beta _2^2\right)}{\left(q_1^2-\beta _2\right){}^4}+\frac{2 \left(\hbar ^2 \beta _1 \beta _2^4+\hbar ^2
   \beta _1 \beta _3 \beta _2^3+16 \hbar ^2 \beta _1 \beta _3^2 \beta _2^2\right)}{\left(q_1^2-\beta _2\right){}^2}
   \end{aligned}\]
   \[\begin{aligned}&-\frac{2 \left(4 \hbar ^2 \beta _1 \beta _2^6+\hbar ^2 \beta _1 \beta
   _3 \beta _2^5+\hbar ^2 \beta _1 \beta _3^2 \beta _2^4\right)}{\left(q_1^2-\beta _2\right){}^4}+\frac{10 \hbar ^2 \beta _1 \beta _2^6+4 \hbar ^2 \beta _1 \beta _3 \beta _2^5-\hbar ^4
   \beta _2^4+2 \hbar ^2 \beta _1 \beta _3^2 \beta _2^4+7 \hbar ^4 \beta _3 \beta _2^3}{\left(q_1^2-\beta _2\right){}^4}\\
   &+\frac{q_1^4 \hbar ^4-2 q_1^2 \beta _2 \hbar ^4-16 q_1^4 \beta
   _1 \beta _3^2 \hbar ^2+32 q_1^2 \beta _1 \beta _2 \beta _3^2 \hbar ^2+64 q_1^4 \beta _1^2 \beta _3^4-128 q_1^2 \beta _1^2 \beta _2 \beta _3^4}{16 q_2^4}\\
   &+\frac{-8 \hbar ^4 q_1^8-32
   \beta _1^2 \beta _3^4 q_1^8+68 \hbar ^2 \beta _1 \beta _3^2 q_1^8+128 \beta _1^2 \beta _2 \beta _3^4 q_1^6+24 \hbar ^2 \beta _1 \beta _3^3 q_1^6-272 \hbar ^2 \beta _1 \beta _2 \beta
   _3^2 q_1^6+32 \hbar ^4 \beta _2 q_1^6-3 \hbar ^4 \beta _3 q_1^6}{4 q_2^2 \left(q_1^2-\beta _2\right){}^2 \beta _3}\\
   &+\frac{-160 \beta _1^2 \beta _2^2 \beta _3^4 q_1^4-8 \hbar ^2 \beta _1 \beta _3^4 q_1^4-56 \hbar ^2 \beta _1 \beta _2 \beta
   _3^3 q_1^4-40 \hbar ^4 \beta _2^2 q_1^4+\hbar ^4 \beta _3^2 q_1^4+340 \hbar ^2 \beta _1 \beta _2^2 \beta _3^2 q_1^4+7 \hbar ^4 \beta _2 \beta _3 q_1^4+64 \beta _1^2 \beta _2^3 \beta
   _3^4 q_1^2}{4 q_2^2 \left(q_1^2-\beta _2\right){}^2 \beta _3}\\
   &+\frac{-112 \hbar ^2 \beta _1 \beta _2 \beta _3^4 q_1^2+16 \hbar ^4 \beta _2^3 q_1^2+40 \hbar ^2 \beta _1 \beta _2^2 \beta _3^3 q_1^2-136 \hbar ^2 \beta _1 \beta _2^3 \beta _3^2
   q_1^2+14 \hbar ^4 \beta _2 \beta _3^2 q_1^2-5 \hbar ^4 \beta _2^2 \beta _3 q_1^2-8 \hbar ^2 \beta _1 \beta _2^2 \beta _3^4}{4 q_2^2 \left(q_1^2-\beta _2\right){}^2 \beta _3}\\
   &+\frac{-8 \hbar ^2 \beta _1 \beta _2^3 \beta _3^3+\hbar ^4 \beta
   _2^2 \beta _3^2+\hbar ^4 \beta _2^3 \beta _3}{4 q_2^2 \left(q_1^2-\beta _2\right){}^2 \beta _3}+q_2^2 \left(-\frac{8 \beta _1^2 \beta _3 \beta _2^6}{\left(q_1^2-\beta
   _2\right){}^4}+\frac{16 \beta _1^2 \beta _3 \beta _2^5}{\left(q_1^2-\beta _2\right){}^3}-\frac{8 \beta _1^2 \beta _3 \beta _2^4}{q_1^4}-\frac{24 \beta _1^2 \beta _3 \beta
   _2^4}{\left(q_1^2-\beta _2\right){}^2}+\frac{2 \hbar ^2 \beta _1 \beta _3 \beta _2^4}{\left(q_1^2-\beta _2\right){}^4}\right.\\
   &+\frac{32 \beta _1^2 \beta _3 \beta _2^3}{q_1^2-\beta
   _2}-\frac{4 \hbar ^2 \beta _1 \beta _3 \beta _2^3}{\left(q_1^2-\beta _2\right){}^3}+\frac{87 \hbar ^4 \beta _2^3}{4 \left(q_1^2-\beta _2\right){}^4}+\frac{2 \hbar ^2 \beta _1 \beta
   _3 \beta _2^2}{q_1^4}+\frac{6 \hbar ^2 \beta _1 \beta _3 \beta _2^2}{\left(q_1^2-\beta _2\right){}^2}-\frac{265 \hbar ^4 \beta _3 \beta _2^2}{2 \left(q_1^2-\beta
   _2\right){}^4}-\frac{8 \hbar ^2 \beta _1 \beta _3 \beta _2}{q_1^2-\beta _2}+\frac{991 \hbar ^4 \beta _3 \beta _2}{4 \left(q_1^2-\beta _2\right){}^3}\\
   &+\frac{531 \hbar ^4 \beta _2}{4
   \left(q_1^2-\beta _2\right){}^2}-\frac{57 \hbar ^4 \beta _3}{2 \left(q_1^2-\beta _2\right){}^2}-\frac{\hbar ^4 \beta _3}{8 q_1^4}-\frac{3 \hbar ^4 \left(\beta _2+6 \beta
   _3\right)}{4 \left(q_1^2-\beta _2\right){}^2}+2 \left(3 \hbar ^2 \beta _1 \beta _2-\hbar ^2 \beta _1 \beta _3\right)+\frac{3 \hbar ^4 \left(\beta _2^2+6 \beta _3 \beta _2\right)}{4
   \left(q_1^2-\beta _2\right){}^3}\\
   &+\frac{3 \left(191 \hbar ^4 \beta _2^2-495 \hbar ^4 \beta _2 \beta _3\right)}{4 \left(q_1^2-\beta _2\right){}^3}-\frac{3 \hbar ^4 \left(\beta _2^3+6
   \beta _3 \beta _2^2\right)}{4 \left(q_1^2-\beta _2\right){}^4}+\frac{2 \left(8 \beta _1^2 \beta _3 \beta _2^3+3 \hbar ^2 \beta _1 \beta _2^2-17 \hbar ^2 \beta _1 \beta _3 \beta
   _2\right)}{q_1^2-\beta _2}\\
   &+\frac{2 \left(8 \beta _1^2 \beta _3 \beta _2^3+3 \hbar ^2 \beta _1 \beta _2^2-17 \hbar ^2 \beta _1 \beta _3 \beta _2\right)}{q_1^2}-\frac{6 \left(8 \beta
   _1^2 \beta _3 \beta _2^3+\hbar ^2 \beta _1 \beta _2^2-7 \hbar ^2 \beta _1 \beta _3 \beta _2\right)}{q_1^2-\beta _2}-\frac{2 \left(12 \beta _1^2 \beta _3 \beta _2^4+3 \hbar ^2 \beta
   _1 \beta _2^3-2 \hbar ^2 \beta _1 \beta _3 \beta _2^2\right)}{\left(q_1^2-\beta _2\right){}^2}\\
   &+\frac{3 \left(3 \beta _3 \hbar ^4+2 \beta _1 \beta _2^3 \hbar ^2-14 \beta _1 \beta
   _2^2 \beta _3 \hbar ^2+16 \beta _1^2 \beta _2^4 \beta _3\right)}{\left(q_1^2-\beta _2\right){}^2}+\frac{2 \left(16 \beta _1^2 \beta _3 \beta _2^5+3 \hbar ^2 \beta _1 \beta _2^4-19
   \hbar ^2 \beta _1 \beta _3 \beta _2^3\right)}{\left(q_1^2-\beta _2\right){}^3}\\
   &-\frac{3 \left(16 \beta _1^2 \beta _3 \beta _2^5+2 \hbar ^2 \beta _1 \beta _2^4-14 \hbar ^2 \beta _1
   \beta _3 \beta _2^3+3 \hbar ^4 \beta _3 \beta _2\right)}{\left(q_1^2-\beta _2\right){}^3}+\frac{3 \left(16 \beta _1^2 \beta _3 \beta _2^6+2 \hbar ^2 \beta _1 \beta _2^5-14 \hbar ^2
   \beta _1 \beta _3 \beta _2^4+3 \hbar ^4 \beta _3 \beta _2^2\right)}{\left(q_1^2-\beta _2\right){}^4}\\
   &+\frac{2 \left(36 \beta _1^2 \beta _3 \beta _2^6+6 \hbar ^2 \beta _1 \beta
   _2^5-37 \hbar ^2 \beta _1 \beta _3 \beta _2^4\right)}{\left(q_1^2-\beta _2\right){}^4}-\frac{112 \beta _1^2 \beta _3 \beta _2^6+18 \hbar ^2 \beta _1 \beta _2^5-114 \hbar ^2 \beta _1
   \beta _3 \beta _2^4+21 \hbar ^4 \beta _2^3}{\left(q_1^2-\beta _2\right){}^4}+\frac{12 \hbar ^4}{q_1^2-\beta _2}-\frac{3 \hbar ^4}{4 q_1^2}\\
   &\left.-\frac{4 \hbar ^4 \beta
   _3}{\left(q_1^2-\beta _2\right) \beta _2}+\frac{4 \hbar ^4 \beta _3}{q_1^2 \beta _2}\right)+\frac{1}{q_2^2-\beta _3}\left(\frac{765 \beta _2^2 \beta _3^3 \hbar ^4}{\left(q_1^2-\beta
   _2\right){}^4}+\frac{2475 \beta _2 \beta _3^3 \hbar ^4}{2 \left(q_1^2-\beta _2\right){}^3}+\frac{225 \beta _3^3 \hbar ^4}{\left(q_1^2-\beta _2\right){}^2}+\frac{319 \beta _2^3 \beta
   _3^2 \hbar ^4}{2 \left(q_1^2-\beta _2\right){}^4}+\frac{2101 \beta _2^2 \beta _3^2 \hbar ^4}{2 \left(q_1^2-\beta _2\right){}^3}\right.\\
   &+\frac{1947 \beta _2 \beta _3^2 \hbar ^4}{2
   \left(q_1^2-\beta _2\right){}^2}+\frac{165 \beta _3^2 \hbar ^4}{2 \left(q_1^2-\beta _2\right)}-\frac{4 \beta _2 \beta _3 \hbar ^4}{q_1^2-\beta _2}+\frac{4 \beta _2 \beta _3 \hbar
   ^4}{q_1^2}+16 \beta _3 \hbar ^4+\frac{3 \left(214 \beta _2 \beta _3-55 \beta _3^2\right) \hbar ^4}{2 \left(q_1^2-\beta _2\right)}+\frac{\beta _2^2 \beta _3 \left(7 \beta _2^2+9
   \beta _3^2\right) \hbar ^4}{2 \left(q_1^2-\beta _2\right){}^4}\\
   &-\frac{\beta _2 \beta _3 \left(7 \beta _2^2+9 \beta _3^2\right) \hbar ^4}{2 \left(q_1^2-\beta
   _2\right){}^3}+\frac{\beta _3 \left(7 \beta _2^2+9 \beta _3^2\right) \hbar ^4}{2 \left(q_1^2-\beta _2\right){}^2}+\frac{3 \left(202 \beta _2^2 \beta _3-649 \beta _2 \beta
   _3^2\right) \hbar ^4}{2 \left(q_1^2-\beta _2\right){}^2}-\frac{q_1^4 \hbar ^4}{4 \beta _3}-\frac{4 q_1^2 \beta _2 \hbar ^4}{\beta _3}-\frac{32 \beta _1 \beta _2^3 \beta _3 \hbar
   ^2}{q_1^2}\\
   &+\frac{32 \beta _1 \beta _2^3 \beta _3 \hbar ^2}{q_1^2-\beta _2}+32 \beta _1 \beta _2^2 \beta _3 \hbar ^2-20 q_1^4 \beta _1 \beta _3 \hbar ^2+8 q_1^2 \beta _1 \beta _2
   \beta _3 \hbar ^2+16 q_1^4 \beta _1^2 \beta _3^3-32 q_1^2 \beta _1^2 \beta _2 \beta _3^3+\frac{-\hbar ^4 \beta _2^5-\hbar ^4 \beta _3 \beta _2^4}{4 \left(q_1^2-\beta _2\right){}^4}\\
   &+\frac{1}{4} \left(45 q_1^2 \hbar ^4-15 \beta _2 \hbar ^4+257 \beta _3 \hbar
   ^4\right)+\frac{\hbar ^4 \beta _2^5+\hbar ^4 \beta _3 \beta _2^4}{4 \left(q_1^2-\beta
   _2\right){}^4}+\frac{13 \hbar ^4 \beta _2^4 \beta _3-15 \hbar ^4 \beta _2^5}{4 \left(q_1^2-\beta _2\right){}^4}+\frac{9 q_1^4 \hbar ^4-\beta _3^2 \hbar ^4+3 q_1^2 \beta _3 \hbar
   ^4-\beta _2 \beta _3 \hbar ^4}{4 \beta _3}\\
   &+\frac{15 \hbar ^4 \beta _2^5-17 \hbar ^4 \beta _3 \beta _2^4-638 \hbar ^4 \beta _3^2 \beta _2^3}{4 \left(q_1^2-\beta
   _2\right){}^4}+\frac{\hbar ^4 \beta _2^2 \beta _3-63 \hbar ^4 \beta _3^3}{2 \left(q_1^2-\beta _2\right){}^2}+\frac{27 \hbar ^4 \beta _3^3-32 \hbar ^2 \beta _1 \beta _2^4 \beta
   _3}{\left(q_1^2-\beta _2\right){}^2}+\frac{-32 \hbar ^2 \beta _1 \beta _3 \beta _2^5+3 \hbar ^4 \beta _3 \beta _2^3-990 \hbar ^4 \beta _3^3 \beta _2}{\left(q_1^2-\beta
   _2\right){}^3}\\
   &+\frac{-495 \beta _2 \beta _3^3 \hbar ^4-2101 \beta _2^2 \beta _3^2 \hbar ^4+2 \beta _2^3 \beta _3 \hbar ^4}{2 \left(q_1^2-\beta _2\right){}^3}+\frac{32 \hbar ^2 \beta
   _1 \beta _2^5 \beta _3-27 \hbar ^4 \beta _2 \beta _3^3}{\left(q_1^2-\beta _2\right){}^3}+\frac{63 \hbar ^4 \beta _2 \beta _3^3-\hbar ^4 \beta _2^3 \beta _3}{2 \left(q_1^2-\beta
   _2\right){}^3}\\
   &+\frac{-57 \beta _2 \beta _3 \hbar ^4-32 \beta _1 \beta _2 \beta _3^3 \hbar ^2-32 \beta _1 \beta _2^3 \beta _3 \hbar ^2}{q_1^2-\beta _2}-\frac{4 \left(\hbar ^4 \beta
   _2 \beta _3-8 \hbar ^2 \beta _1 \beta _2 \beta _3^3\right)}{q_1^2-\beta _2}+\frac{\hbar ^4 \beta _2^4 \beta _3-63 \hbar ^4 \beta _2^2 \beta _3^3}{2 \left(q_1^2-\beta _2\right){}^4}\\
   &+4 \left(-4 \beta _1^2 \beta _3^3 q_1^4+\hbar ^2 \beta _1 \beta _3 q_1^4+8 \beta _1^2 \beta _2 \beta _3^3
   q_1^2-2 \hbar ^2 \beta _1 \beta _2 \beta _3 q_1^2\right)+\frac{27 \hbar ^4
   \beta _2^2 \beta _3^3-32 \hbar ^2 \beta _1 \beta _2^6 \beta _3}{\left(q_1^2-\beta _2\right){}^4}
   \end{aligned}\]
   \[\begin{aligned}&+\frac{-225 \beta _3^3 \hbar ^4-111 \beta _2^2 \beta _3 \hbar ^4-32 \beta _1 \beta
   _2^2 \beta _3^3 \hbar ^2+32 \beta _1 \beta _2^4 \beta _3 \hbar ^2}{\left(q_1^2-\beta _2\right){}^2}-\frac{4 \left(\hbar ^4 \beta _2^2 \beta _3-8 \hbar ^2 \beta _1 \beta _2^2 \beta
   _3^3\right)}{\left(q_1^2-\beta _2\right){}^2}\\
   &\left.+\frac{160 \hbar ^2 \beta _1 \beta _3 \beta _2^6-2 \hbar ^2 \beta _1 \beta _3^2 \beta _2^5-2 \hbar ^2 \beta _1 \beta _3^3 \beta _2^4-3
   \hbar ^4 \beta _3 \beta _2^4}{\left(q_1^2-\beta _2\right){}^4}+\frac{-128 \hbar ^2 \beta _1 \beta _3 \beta _2^6+2 \hbar ^2 \beta _1 \beta _3^2 \beta _2^5+2 \hbar ^2 \beta _1 \beta
   _3^3 \beta _2^4-765 \hbar ^4 \beta _3^3 \beta _2^2}{\left(q_1^2-\beta _2\right){}^4}\right)\\
   &+\frac{1}{\left(q_2^2-\beta _3\right){}^2}\left(\frac{57}{4} q_1^4 \hbar ^4+\frac{9033 \beta _2^2 \beta _3^4 \hbar ^4}{4
   \left(q_1^2-\beta _2\right){}^4}+\frac{1485 \beta _2 \beta _3^4 \hbar ^4}{\left(q_1^2-\beta _2\right){}^3}+\frac{675 \beta _3^4 \hbar ^4}{\left(q_1^2-\beta _2\right){}^2}-\frac{4
   \beta _2 \beta _3^2 \hbar ^4}{q_1^2-\beta _2}+\frac{4 \beta _2 \beta _3^2 \hbar ^4}{q_1^2}-117 \beta _3^2 \hbar ^4-24 q_1^2 \beta _2 \hbar ^4\right.\\
   &+99 q_1^2 \beta _3 \hbar ^4-33 \beta _2
   \beta _3 \hbar ^4+2 \beta _1 \beta _3^4 \hbar ^2-6 q_1^2 \beta _1 \beta _3^3 \hbar ^2+2 \beta _1 \beta _2 \beta _3^3 \hbar ^2-\frac{32 \beta _1 \beta _2^3 \beta _3^2 \hbar
   ^2}{q_1^2}+\frac{32 \beta _1 \beta _2^3 \beta _3^2 \hbar ^2}{q_1^2-\beta _2}+32 \beta _1 \beta _2^2 \beta _3^2 \hbar ^2\\
   &-18 q_1^4 \beta _1 \beta _3^2 \hbar ^2+4 q_1^2 \beta _1 \beta
   _2 \beta _3^2 \hbar ^2+12 q_1^4 \beta _1^2 \beta _3^4-24 q_1^2 \beta _1^2 \beta _2 \beta _3^4+\frac{1}{4} \left(-9 q_1^4 \hbar ^4+\beta _3^2 \hbar ^4-3 q_1^2 \beta _3 \hbar ^4+\beta
   _2 \beta _3 \hbar ^4\right)\\
   &+\frac{3}{4} \left(257 \beta _3^2 \hbar ^4+45 q_1^2 \beta _3 \hbar ^4-15 \beta _2 \beta _3 \hbar ^4\right)-\frac{4 \beta _3^2 \left(8 \hbar ^2 \beta _1
   \beta _2 \beta _3^2-\hbar ^4 \beta _2\right)}{q_1^2-\beta _2}-\frac{4 \beta _3^2 \left(8 \hbar ^2 \beta _1 \beta _2^2 \beta _3^2-\hbar ^4 \beta _2^2\right)}{\left(q_1^2-\beta
   _2\right){}^2}\\
   &-\frac{3 \left(79 \hbar ^4 \beta _2 \beta _3^2-120 \hbar ^4 \beta _3^3\right)}{2 \left(q_1^2-\beta _2\right)}+\frac{9 \left(107 \hbar ^4 \beta _2 \beta _3^2-40 \hbar
   ^4 \beta _3^3\right)}{2 \left(q_1^2-\beta _2\right)}-\frac{3 \left(\hbar ^4 \beta _2^2 \beta _3^2-1416 \hbar ^4 \beta _2 \beta _3^3\right)}{2 \left(q_1^2-\beta
   _2\right){}^2}+\frac{9 \left(101 \hbar ^4 \beta _2^2 \beta _3^2-472 \hbar ^4 \beta _2 \beta _3^3\right)}{2 \left(q_1^2-\beta _2\right){}^2}\\
   &-\beta _3 \left(24 \beta _1^2 \beta _3^3
   q_1^4-21 \hbar ^2 \beta _1 \beta _3 q_1^4-48 \beta _1^2 \beta _2 \beta _3^3 q_1^2-6 \hbar ^2 \beta _1 \beta _3^2 q_1^2+42 \hbar ^2 \beta _1 \beta _2 \beta _3 q_1^2+2 \hbar ^2 \beta
   _1 \beta _3^3+2 \hbar ^2 \beta _1 \beta _2 \beta _3^2\right)\\
   &\left.+\frac{4584 \hbar ^4 \beta _2^2 \beta _3^3-13 \hbar ^4 \beta _2^3 \beta _3^2}{2 \left(q_1^2-\beta _2\right){}^3}+\frac{3
   \left(15 \hbar ^4 \beta _3 \beta _2^5-23 \hbar ^4 \beta _3^2 \beta _2^4-464 \hbar ^4 \beta _3^3 \beta _2^3\right)}{4 \left(q_1^2-\beta _2\right){}^4}+\frac{696 \beta _2^3 \beta _3^3
   \hbar ^4+35 \beta _2^4 \beta _3^2 \hbar ^4}{2 \left(q_1^2-\beta _2\right){}^4}+\frac{\hbar ^4 \beta _2^2 \beta _3^2-27 \hbar ^4 \beta _3^4}{2 \left(q_1^2-\beta
   _2\right){}^2}\right.\\
   &+\frac{9 \beta _3^4 \hbar ^4+7 \beta _2^2 \beta _3^2 \hbar ^4}{2 \left(q_1^2-\beta _2\right){}^2}+\frac{9 \hbar ^4 \beta _3^4-32 \hbar ^2 \beta _1 \beta _2^4 \beta
   _3^2}{\left(q_1^2-\beta _2\right){}^2}+\frac{-32 \hbar ^2 \beta _1 \beta _3^2 \beta _2^5+9 \hbar ^4 \beta _3^2 \beta _2^3-2970 \hbar ^4 \beta _3^4 \beta _2}{\left(q_1^2-\beta
   _2\right){}^3}+\frac{-9 \beta _2 \beta _3^4 \hbar ^4-7 \beta _2^3 \beta _3^2 \hbar ^4}{2 \left(q_1^2-\beta _2\right){}^3}\\
   &+\frac{32 \hbar ^2 \beta _1 \beta _2^5 \beta _3^2-9 \hbar ^4
   \beta _2 \beta _3^4}{\left(q_1^2-\beta _2\right){}^3}+\frac{27 \hbar ^4 \beta _2 \beta _3^4-\hbar ^4 \beta _2^3 \beta _3^2}{2 \left(q_1^2-\beta _2\right){}^3}+\frac{3 \left(990
   \beta _2 \beta _3^4 \hbar ^4-1528 \beta _2^2 \beta _3^3 \hbar ^4+\beta _2^3 \beta _3^2 \hbar ^4\right)}{2 \left(q_1^2-\beta _2\right){}^3}\\
   &+\frac{-171 \beta _2 \beta _3^2 \hbar ^4+32
   \beta _1 \beta _2 \beta _3^4 \hbar ^2-32 \beta _1 \beta _2^3 \beta _3^2 \hbar ^2}{q_1^2-\beta _2}-3 \left(-4 \beta _1^2 \beta _3^4 q_1^4+\hbar ^2 \beta _1 \beta _3^2 q_1^4+8 \beta
   _1^2 \beta _2 \beta _3^4 q_1^2-2 \hbar ^2 \beta _1 \beta _2 \beta _3^2 q_1^2\right)\\
   &+\frac{\hbar ^4 \beta _2^4 \beta _3^2-27 \hbar ^4 \beta _2^2 \beta _3^4}{2 \left(q_1^2-\beta
   _2\right){}^4}+\frac{7 \hbar ^4 \beta _3^2 \beta _2^4+9 \hbar ^4 \beta _3^4 \beta _2^2}{2 \left(q_1^2-\beta _2\right){}^4}+\frac{9 \hbar ^4 \beta _2^2 \beta _3^4-32 \hbar ^2 \beta
   _1 \beta _2^6 \beta _3^2}{\left(q_1^2-\beta _2\right){}^4}+\frac{-177 \hbar ^4 \beta _3 \beta _2^5-5 \hbar ^4 \beta _3^2 \beta _2^4+147 \hbar ^4 \beta _3^4 \beta _2^2}{4
   \left(q_1^2-\beta _2\right){}^4}\\
   &+\frac{-675 \beta _3^4 \hbar ^4-333 \beta _2^2 \beta _3^2 \hbar ^4+32 \beta _1 \beta _2^2 \beta _3^4 \hbar ^2+32 \beta _1 \beta _2^4 \beta _3^2 \hbar
   ^2}{\left(q_1^2-\beta _2\right){}^2}+\frac{160 \hbar ^2 \beta _1 \beta _3^2 \beta _2^6+2 \hbar ^2 \beta _1 \beta _3^3 \beta
   _2^5+2 \hbar ^2 \beta _1 \beta _3^4 \beta _2^4-3 \hbar ^4 \beta _3^2 \beta _2^4}{\left(q_1^2-\beta _2\right){}^4}\\
   &\left.+\frac{-128 \hbar ^2 \beta _1 \beta _3^2 \beta _2^6-2 \hbar ^2 \beta _1 \beta _3^3 \beta _2^5+33 \hbar ^4 \beta _3 \beta _2^5-2 \hbar ^2 \beta _1
   \beta _3^4 \beta _2^4-2295 \hbar ^4 \beta _3^4 \beta _2^2}{\left(q_1^2-\beta _2\right){}^4}\right)+\frac{1}{\left(q_2^2-\beta _3\right){}^4}\left(\frac{765 \hbar ^4 \beta
   _2^2 \beta _3^6}{\left(q_1^2-\beta _2\right){}^4}+\frac{2475 \hbar ^4 \beta _2 \beta _3^6}{4 \left(q_1^2-\beta _2\right){}^3}\right.\\
   &+\frac{225 \hbar ^4 \beta _3^6}{\left(q_1^2-\beta
   _2\right){}^2}+\frac{2 \left(\hbar ^2 \beta _1 \beta _2^5+\hbar ^2 \beta _1 \beta _3 \beta _2^4\right) \beta _3^5}{\left(q_1^2-\beta _2\right){}^4}-\frac{4 \left(8 \hbar ^2 \beta _1
   \beta _2 \beta _3^2-\hbar ^4 \beta _2\right) \beta _3^4}{q_1^2-\beta _2}-\frac{4 \left(8 \hbar ^2 \beta _1 \beta _2^2 \beta _3^2-\hbar ^4 \beta _2^2\right) \beta
   _3^4}{\left(q_1^2-\beta _2\right){}^2}\\
   &-\frac{\left(\hbar ^4 \beta _2^5+\hbar ^4 \beta _3 \beta _2^4\right) \beta _3^3}{4 \left(q_1^2-\beta _2\right){}^4}-\left(24 \beta _1^2 \beta
   _3^3 q_1^4-21 \hbar ^2 \beta _1 \beta _3 q_1^4-48 \beta _1^2 \beta _2 \beta _3^3 q_1^2-6 \hbar ^2 \beta _1 \beta _3^2 q_1^2+42 \hbar ^2 \beta _1 \beta _2 \beta _3 q_1^2+2 \hbar ^2
   \beta _1 \beta _3^3\right.\\
   &\left.+2 \hbar ^2 \beta _1 \beta _2 \beta _3^2\right) \beta _3^3+\frac{113}{8} \hbar ^4 q_1^4 \beta _3^2+\frac{35}{8} \hbar ^4 q_1^2 \beta _2 \beta _3^2+\frac{1}{4}
   \left(-9 q_1^4 \hbar ^4+\beta _3^2 \hbar ^4-3 q_1^2 \beta _3 \hbar ^4+\beta _2 \beta _3 \hbar ^4\right) \beta _3^2+\left(4 \beta _1^2 \beta _3^4 q_1^4-\hbar ^2 \beta _1 \beta _3^2
   q_1^4\right.\\
   &\left.-8 \beta _1^2 \beta _2 \beta _3^4 q_1^2+2 \hbar ^2 \beta _1 \beta _2 \beta _3^2 q_1^2\right) \beta _3^2+\frac{3}{8} \left(49 \hbar ^4 q_1^4 \beta _3^2-75 \hbar ^4 q_1^2 \beta
   _2 \beta _3^2\right)+\frac{1}{2} \left(147 \beta _3^4 \hbar ^4-177 q_1^2 \beta _3^3 \hbar ^4+59 \beta _2 \beta _3^3 \hbar ^4-135 q_1^2 \beta _2 \beta _3^2 \hbar
   ^4\right)\\
   &+\frac{1}{4} \left(257 \beta _3^4 \hbar ^4+45 q_1^2 \beta _3^3 \hbar ^4-15 \beta _2 \beta _3^3 \hbar ^4+135 q_1^4 \beta _3^2 \hbar ^4-147 q_1^2 \beta _2 \beta _3^2 \hbar
   ^4\right)+\frac{\hbar ^4 \left(161 \beta _2 \beta _3^4-30 \beta _3^5\right)}{2 \left(q_1^2-\beta _2\right)}-\frac{3 \hbar ^4 \left(79 \beta _2 \beta _3^4-10 \beta _3^5\right)}{2
   \left(q_1^2-\beta _2\right)}\\
   &+\frac{\hbar ^4 \left(143 \beta _2^2 \beta _3^4-354 \beta _2 \beta _3^5\right)}{2 \left(q_1^2-\beta _2\right){}^2}-\frac{3 \hbar ^4 \left(\beta _2^2
   \beta _3^4-118 \beta _2 \beta _3^5\right)}{2 \left(q_1^2-\beta _2\right){}^2}-\frac{\hbar ^4 \left(13 \beta _2^3 \beta _3^4-382 \beta _2^2 \beta _3^5\right)}{2 \left(q_1^2-\beta
   _2\right){}^3}+\frac{15 \hbar ^4 \beta _3^3 \beta _2^5-3 \hbar ^4 \beta _3^4 \beta _2^4-116 \hbar ^4 \beta _3^5 \beta _2^3}{4 \left(q_1^2-\beta _2\right){}^4}\\
   &+\frac{-59 \hbar ^4
   \beta _3^3 \beta _2^5+2 \hbar ^4 \beta _3^4 \beta _2^4+58 \hbar ^4 \beta _3^5 \beta _2^3}{2 \left(q_1^2-\beta _2\right){}^4}+\frac{\hbar ^4 \left(1485 \beta _2 \beta _3^6-764 \beta
   _2^2 \beta _3^5+2 \beta _2^3 \beta _3^4\right)}{4 \left(q_1^2-\beta _2\right){}^3}+\frac{6 \left(\hbar ^4 \beta _2^3 \beta _3^4-165 \hbar ^4 \beta _2 \beta
   _3^6\right)}{\left(q_1^2-\beta _2\right){}^3}\\
   &+\frac{2 \left(16 \hbar ^2 \beta _1 \beta _2 \beta _3^6+17 \hbar ^4 \beta _2 \beta _3^4\right)}{q_1^2-\beta _2}+2 \left(10 q_1^4 \beta
   _1^2 \beta _3^6+\hbar ^2 \beta _1 \beta _3^6-20 q_1^2 \beta _1^2 \beta _2 \beta _3^6-3 \hbar ^2 q_1^2 \beta _1 \beta _3^5+\hbar ^2 \beta _1 \beta _2 \beta _3^5-69 \hbar ^4 \beta
   _3^4-10 \hbar ^2 q_1^4 \beta _1 \beta _3^4\right.
   \end{aligned}\]
   \[\begin{aligned}
   &\left.+20 \hbar ^2 q_1^2 \beta _1 \beta _2 \beta _3^4+39 \hbar ^4 q_1^2 \beta _3^3-13 \hbar ^4 \beta _2 \beta _3^3\right)+\frac{-225 \hbar ^4
   \beta _3^6+32 \hbar ^2 \beta _1 \beta _2^2 \beta _3^6-74 \hbar ^4 \beta _2^2 \beta _3^4}{\left(q_1^2-\beta _2\right){}^2}\\
   &\left.+\frac{-2 \hbar ^2 \beta _1 \beta _2^4 \beta _3^6-765 \hbar
   ^4 \beta _2^2 \beta _3^6-2 \hbar ^2 \beta _1 \beta _2^5 \beta _3^5+26 \hbar ^4 \beta _2^5 \beta _3^3}{\left(q_1^2-\beta _2\right){}^4}\right)
\end{aligned}\]\normalsize
\begin{enumerate}
    \item[(5a)] \(\displaystyle H=\tfrac{1}{2}(p_1^2+p_2^2)+\hbar^2\left[\frac{q_1^2+q_2^2}{8\beta^2}+\frac{2(q_1^2+\beta)}{(q_1^2-\beta)^2}+\frac{2(q_2^2+\beta)}{(q_2^2-\beta)^2}\right]\)
    \end{enumerate}
This is a special case of (5) where \(\beta_1=\dfrac{\hbar^2}{8\beta^2},\beta_2=\beta_3=\beta\). There is a third-order integral
\[\begin{aligned}
    X^{(1)}&=\{m_{12},\{m_{12},m_{12}\}-3\beta p_1^2-3\beta p_2^2\}+\hbar^2\left\{p_1,-\frac{3 q_2 q_1^2 \left(23 \beta ^2+10 \beta  q_2^2-q_2^4\right)}{4 \beta  \left(\beta -q_2^2\right){}^2}\right.\\
    &\qquad\left.+\frac{12 \left(3 \beta ^2 q_2-\beta  q_2^3\right)}{\left(q_1^2-\beta
   \right){}^2}+\frac{q_2 \left(52 \beta ^3+67 \beta ^2 q_2^2-26 \beta  q_2^4+3 q_2^6\right)}{4 \beta  \left(\beta -q_2^2\right){}^2}+\frac{6 \left(3 \beta 
   q_2-q_2^3\right)}{q_1^2-\beta }\right\}\\
   &\qquad+\hbar^2\left\{p_2,\frac{3 q_2^2 q_1 \left(23 \beta ^2+10 \beta  q_1^2-q_1^4\right)}{4 \beta  \left(\beta -q_1^2\right){}^2}+\frac{138 \beta ^2 q_1-127 \beta  q_1^3}{4 \left(\beta
   -q_1^2\right){}^2}-\frac{5 \left(19 \beta ^2 q_1-6 \beta  q_1^3\right)}{2 \left(\beta -q_1^2\right){}^2}\right.\\
   &\qquad\left.-\frac{12 \left(3 \beta ^2 q_1-\beta  q_1^3\right)}{\left(q_2^2-\beta
   \right){}^2}+\frac{8 q_1^5}{\left(q_1^2-\beta \right){}^2}-\frac{3 \left(2 \beta  q_1^5+q_1^7\right)}{4 \beta  \left(\beta -q_1^2\right){}^2}-\frac{6 \left(3 \beta 
   q_1-q_1^3\right)}{q_2^2-\beta }\right\}
\end{aligned}\]
and one more fourth-order integral \((H_1-H_2,X^{(1)})\).
\begin{enumerate}
    \item[(5b)] \(\displaystyle H=\tfrac{1}{2}(p_1^2+p_2^2)+\hbar^2\left[\frac{q_1^2+q_2^2}{8\beta^2}+\frac{2(q_1^2+\beta)}{(q_1^2-\beta)^2}+\frac{2(q_2^2-\beta)}{(q_2^2+\beta)^2}\right]\)
    \end{enumerate}
This is a special case of (5) where \(\beta_1=\dfrac{\hbar^2}{8\beta^2},\beta_2=\beta,\beta_3=-\beta\). There is a third-order integral
\[\begin{aligned}
    X^{(1)}&=\{m_{12},\{m_{12},m_{12}\}+3\beta p_1^2-3\beta p_2^2\}+\hbar^2\left\{p_1,\frac{3 q_2 q_1^2 \left(23 \beta ^2-10 \beta  q_2^2-q_2^4\right)}{4 \beta  \left(\beta +q_2^2\right){}^2}-\frac{12 \left(3 \beta ^2 q_2+\beta  q_2^3\right)}{\left(q_1^2-\beta
   \right){}^2}\right.\\
   &\qquad\left.-\frac{q_2 \left(68 \beta ^3-35 \beta ^2 q_2^2-10 \beta  q_2^4-3 q_2^6\right)}{4 \beta  \left(\beta +q_2^2\right){}^2}-\frac{6 \left(3 \beta 
   q_2+q_2^3\right)}{q_1^2-\beta }\right\}+\hbar^2\left\{p_2,-\frac{3 q_1^3 \left(7 \beta ^2-10 \beta  q_2^2-q_2^4\right)}{4 \beta  \left(\beta +q_2^2\right){}^2}\right.\\
   &\qquad\left.+\frac{q_1 \left(68 \beta ^3-83 \beta ^2 q_2^2-10 \beta  q_2^4-3 q_2^6\right)}{4
   \beta  \left(\beta +q_2^2\right){}^2}+\frac{6 q_1 \left(\beta +q_2^2\right)}{q_1^2-\beta }+\frac{24 \beta  q_1 \left(\beta +q_2^2\right)}{\left(q_1^2-\beta \right){}^2}\right\}
\end{aligned}\]
and one more fourth-order integral \((H_1-H_2,X^{(1)})\).
\begin{enumerate}
\item[(6)] \(\displaystyle H=\tfrac{1}{2}(p_1^2+p_2^2)+\frac{3\hbar^2q_1(q_1^3+2\beta_1)}{(q_1^3-\beta_1)^2}+\frac{\beta_2}{q_2^2}\)
\end{enumerate}
There are three fourth-order integrals
\[\begin{aligned}\displaystyle X^{(1)}&=\{\{m_{12}, m_{12}\},\{m_{12},m_{12}\}\}-4\beta_1\{p_2^3, m_{12}\}\\
&\qquad+\left\{p_1^2,\frac{36 \beta _1 q_1 q_2^4 \hbar ^2}{\left(q_1^3-\beta _1\right){}^2}+\frac{12 q_1 q_2^4 \hbar ^2}{q_1^3-\beta _1}+4 \beta _2 q_1^2+5 q_2^2 \hbar ^2\right\}\\
&\qquad+\left\{p_1p_2,-\frac{108 \beta _1 q_2^3 q_1^2 \hbar ^2}{\left(q_1^3-\beta _1\right){}^2}-\frac{24 q_2^3 q_1^2 \hbar ^2}{q_1^3-\beta _1}-\frac{8 \beta _2
   q_1^3}{q_2}+\frac{8 \beta _1 \beta _2}{q_2}-10 q_2 q_1 \hbar ^2\right\}\\
   &\qquad+\left\{p_2^2,\frac{108 \beta _1^2 q_2^2 \hbar ^2}{\left(q_1^3-\beta _1\right){}^2}+\frac{108 \beta _1 q_2^2 \hbar ^2}{q_1^3-\beta _1}+\frac{4 \beta _2
   q_1^4}{q_2^2}-\frac{16 \beta _1 \beta _2 q_1}{q_2^2}+5 q_1^2 \hbar ^2+12 q_2^2 \hbar ^2\right\}\\
   &\qquad+\frac{1}{2} \hbar ^2 \left(4 \beta _2+5 \hbar ^2\right)+\frac{324 \beta _1^2 q_2^4 q_1^2 \hbar ^4}{\left(q_1^3-\beta _1\right){}^4}+\frac{54 \left(9
   \beta _1^2 q_1 q_2^2 \hbar ^4+4 \beta _1 q_1^2 q_2^4 \hbar ^4\right)}{\left(q_1^3-\beta _1\right){}^3}+\frac{6 \beta _2 q_1^2 \hbar
   ^2}{q_2^2}\\
   &\qquad+\frac{18 \left(3 \beta _1 \hbar ^4+4 \beta _1 \beta _2 \hbar ^2+q_1 q_2^2 \hbar ^4\right)}{q_1^3-\beta _1}+\frac{18 \left(3 \beta _1^2
   \hbar ^4+4 \beta _1^2 \beta _2 \hbar ^2+21 \beta _1 q_1 q_2^2 \hbar ^4+2 q_1^2 q_2^4 \hbar ^4\right)}{\left(q_1^3-\beta _1\right){}^2}\\
   &\qquad+\frac{4
   \beta _2^2 q_1^4}{q_2^4}-\frac{16 \beta _1 \beta _2^2 q_1}{q_2^4}\end{aligned}\]
\[\begin{aligned}\displaystyle X^{(2)}&=\{p_1^2,\{m_{12},m_{12}\}\}+\left\{p_1^2,\frac{36 \beta _1 q_1 q_2^2 \hbar ^2}{\left(q_1^3-\beta _1\right){}^2}+\frac{12 q_1 q_2^2 \hbar ^2}{q_1^3-\beta _1}+\frac{2 \beta _2
   q_1^2}{q_2^2}+\frac{\hbar ^2}{2}\right\}\\
   &\qquad+\left\{p_1p_2,-\frac{54 \beta _1 q_2 q_1^2 \hbar ^2}{\left(q_1^3-\beta _1\right){}^2}-\frac{12 q_2 q_1^2 \hbar ^2}{q_1^3-\beta _1}\right\}+\left\{p_2^2,\frac{18 \beta _1^2 \hbar ^2}{\left(q_1^3-\beta _1\right){}^2}+\frac{18 \beta _1 \hbar ^2}{q_1^3-\beta _1}+\frac{\hbar ^2}{2}\right\}\\
   &\qquad+\frac{324 \beta _1^2 q_1^2 q_2^2 \hbar ^4}{\left(q_1^3-\beta _1\right){}^4}+\frac{27 \left(3 \beta _1^2 q_1 \hbar ^4+8 \beta _1 q_1^2 q_2^2 \hbar
   ^4\right)}{\left(q_1^3-\beta _1\right){}^3}+\frac{\beta _2 \hbar ^2}{q_2^2}+\frac{3 \left(12 \beta _1 \beta _2 \hbar ^2+q_1 q_2^2 \hbar
   ^4\right)}{q_2^2 \left(q_1^3-\beta _1\right)}\\
   &\qquad+\frac{9 \left(4 \beta _1^2 \beta _2 \hbar ^2+7 \beta _1 q_1 q_2^2 \hbar ^4+4 q_1^2 q_2^4 \hbar
   ^4\right)}{q_2^2 \left(q_1^3-\beta _1\right){}^2}\end{aligned}\]
\[\begin{aligned}\displaystyle X^{(3)}&=\{p_1^2p_2, m_{12}\}+\left\{p_1^2,\frac{2 \beta _2 q_1}{q_2^2}-\frac{q_2^2 \hbar ^2}{2 \beta _1}\right\}+\left\{p_1p_2,-\frac{27 \beta _1 q_1 q_2 \hbar ^2}{\left(q_1^3-\beta _1\right){}^2}-\frac{9 q_1 q_2 \hbar ^2}{q_1^3-\beta _1}+\frac{q_1 q_2 \hbar ^2}{\beta _1}\right\}\\
&\qquad+\left\{p_2^2,\frac{18 \beta _1 q_1^2 \hbar ^2}{\left(q_1^3-\beta _1\right){}^2}+\frac{3 q_1^2 \hbar ^2}{q_1^3-\beta _1}-\frac{q_1^2 \hbar ^2}{2 \beta _1}\right\}+\frac{\hbar ^4}{4 \beta _1}+\frac{81 \beta _1^2 \hbar ^4}{2 \left(q_1^3-\beta _1\right){}^3}-\frac{\beta _2 q_1^2 \hbar ^2}{\beta _1 q_2^2}\\
&\qquad-\frac{9
   \left(-9 \beta _1 q_2^2 \hbar ^4-8 \beta _1 \beta _2 q_1^2 \hbar ^2+2 q_1 q_2^4 \hbar ^4\right)}{2 q_2^2 \left(q_1^3-\beta _1\right){}^2}-\frac{3
   \left(-3 \beta _1 q_2^2 \hbar ^4-4 \beta _1 \beta _2 q_1^2 \hbar ^2+2 q_1 q_2^4 \hbar ^4\right)}{2 \beta _1 q_2^2 \left(q_1^3-\beta _1\right)}\end{aligned}\]
\begin{enumerate}
\item[(6a)] \(\displaystyle H=\tfrac{1}{2}(p_1^2+p_2^2)+\frac{3\hbar^2q_1(q_1^3+2\beta)}{(q_1^3-\beta)^2}\)
\end{enumerate}
This a special case of (6) where \(\beta_1=\beta,\beta_2=0\). There is a first-order integral \(p_2\), a third-order integral \(p_2^3\) and three more fourth-order integrals \[(p_2,X^{(1)}),(p_2,(p_2,(p_2,X^{(1)}))),(p_2,X^{(3)})\]
\begin{enumerate}
\item[(6b)] \(\displaystyle H=\tfrac{1}{2}(p_1^2+p_2^2)+\frac{3\hbar^2q_1(q_1^3+2\beta)}{(q_1^3-\beta)^2}+\frac{\hbar^2}{q_2^2}\)
\end{enumerate}
This is a special case of (6) where \(\beta_1=\beta,\beta_2=\hbar^2\). There is a third-order integral
\[X^{(4)}=p_2^3+3\hbar^2\left\{p_2,\frac{1}{q_2^2}\right\}\]
and one more fourth-order integral
\[\begin{aligned}\displaystyle X^{(5)}&=\{p_1,\{m_{12},\{m_{12},m_{12}\}\}-\beta p_2^3\}+\hbar^2\left\{p_1^2,-\frac{12 q_1 q_2^3}{q_1^3-\beta }-\frac{36 \beta  q_1 q_2^3}{\left(q_1^3-\beta \right){}^2}-2 q_2-\frac{3 q_1^2}{q_2}\right\}\\
&\qquad+\hbar^2\left\{p_1p_2,\frac{18 q_2^2 q_1^2}{q_1^3-\beta }+\frac{81 \beta  q_2^2 q_1^2}{\left(q_1^3-\beta \right){}^2}-\frac{3 \beta }{q_2^2}+\frac{3 q_1^3}{q_2^2}+2 q_1\right\}\\
&\qquad+\hbar^2\left\{p_2^2,-\frac{54 \beta ^2 q_2}{\left(q_1^3-\beta \right){}^2}-\frac{54 \beta  q_2}{q_1^3-\beta }-6 q_2\right\}+\hbar^4\left(-\frac{324 \beta ^2 q_1^2 q_2^3}{\left(q_1^3-\beta \right){}^4}-\frac{27 \left(9 \beta ^2 q_1 q_2+8 \beta  q_1^2 q_2^3\right)}{\left(q_1^3-\beta
   \right){}^3}\right.\\
   &\qquad\left.-\frac{9 \left(6 \beta ^2+21 \beta  q_1 q_2^2+4 q_1^2 q_2^4\right)}{q_2 \left(q_1^3-\beta \right){}^2}-\frac{9 \left(6 \beta +q_1
   q_2^2\right)}{q_2 \left(q_1^3-\beta \right)}-\frac{15}{2 q_2}-\frac{9 q_1^2}{2 q_2^3}\right)\end{aligned}\]
\begin{enumerate}
\item[(7)] \(\displaystyle H=\tfrac{1}{2}(p_1^2+p_2^2)+\hbar^2\left[\frac{3q_1(q_1^3+2\beta_1)}{(q_1^3-\beta_1)^2}+\frac{2(q_2^2+\beta_2)}{(q_2^2-\beta_2)^2}-\frac{1}{8q_2^2}\right]\)
\end{enumerate}
There is one fourth-order integral
\[\begin{aligned}\displaystyle X&=\{\{m_{12},m_{12}\},\{m_{12},m_{12}\}-2\beta_2p_1^2\}-4\beta_1\{p_2^3,m_{12}\}\\
&\qquad+\hbar^2\left\{p_1^2,-\beta _2+q_1^2 \left(\frac{\beta _2^5}{2 q_2^2 \left(q_2^2-\beta _2\right){}^4}-\frac{45 \beta _2^2 q_2^4-29 \beta _2^3 q_2^2}{\left(q_2^2-\beta
   _2\right){}^4}+\frac{-5 \beta _2^4+21 \beta _2 q_2^6+15 q_2^8}{2 \left(q_2^2-\beta _2\right){}^4}\right)\right.\\
   &\qquad\left.+\frac{12 q_1 \left(q_2^4-2 \beta _2
   q_2^2\right)}{q_1^3-\beta _1}+\frac{36 q_1 \left(\beta _1 q_2^4-2 \beta _1 \beta _2 q_2^2\right)}{\left(q_1^3-\beta _1\right){}^2}+5 q_2^2\right\}\\
   &\qquad+\hbar^2\left\{p_1p_2,-\frac{q_1^3 \left(-\beta _2^2+34 \beta _2 q_2^2+15 q_2^4\right)}{q_2 \left(q_2^2-\beta _2\right){}^2}-\frac{24 q_2 q_1^2 \left(q_2^2-\beta
   _2\right)}{q_1^3-\beta _1}-\frac{108 q_1^2 \left(\beta _1 q_2^3-\beta _1 \beta _2 q_2\right)}{\left(q_1^3-\beta _1\right){}^2}\right.\\
   &\qquad\left.+\frac{\beta _1
   \left(-\beta _2^2+34 \beta _2 q_2^2+15 q_2^4\right)}{q_2 \left(q_2^2-\beta _2\right){}^2}-10 q_2 q_1\right\}+\hbar^2\left\{p_2^2,-\beta _2+q_1^4 \left(\frac{15 q_2^6}{2 \left(q_2^2-\beta _2\right){}^4}\right.\right.\\
   &\qquad\left.-\frac{-10 \beta _2^3+6 \beta _2 q_2^4+11 \beta _2^2 q_2^2}{\left(q_2^2-\beta
   _2\right){}^4}-\frac{\beta _2^4}{2 q_2^2 \left(q_2^2-\beta _2\right){}^4}\right)+q_1 \left(\frac{30 \beta _1 q_2^6}{\left(q_2^2-\beta
   _2\right){}^4}\right.\\
   &\qquad\left.-\frac{4 \left(10 \beta _1 \beta _2^3+15 \beta _1 q_2^6-6 \beta _1 \beta _2 q_2^4-11 \beta _1 \beta _2^2
   q_2^2\right)}{\left(q_2^2-\beta _2\right){}^4}+\frac{2 \beta _1 \beta _2^4}{q_2^2 \left(q_2^2-\beta _2\right){}^4}\right)+\frac{36 \left(3 \beta _1
   q_2^2-\beta _1 \beta _2\right)}{q_1^3-\beta _1}\\
   &\qquad\left.+\frac{36 \left(3 \beta _1^2 q_2^2-\beta _1^2 \beta _2\right)}{\left(q_1^3-\beta _1\right){}^2}+5
   q_1^2+12 q_2^2\right\}+\hbar^4\left(\frac{\beta _2^5}{4 q_2^2 \left(q_2^2-\beta _2\right){}^4}\right.\\
   &\qquad+\frac{54 \left(4 q_1^2 \beta _1 q_2^4+9 q_1 \beta _1^2 q_2^2-8 q_1^2 \beta _1 \beta _2
   q_2^2-3 q_1 \beta _1^2 \beta _2\right)}{\left(q_1^3-\beta _1\right){}^3}+\frac{324 \left(q_1^2 q_2^4 \beta _1^2-2 q_1^2 q_2^2 \beta _1^2 \beta
   _2\right)}{\left(q_1^3-\beta _1\right){}^4}\\
   &\qquad-\frac{3 \left(133 q_2^4 \beta _2^2-67 q_2^2 \beta _2^3\right)}{2 \left(q_2^2-\beta
   _2\right){}^4}+\tfrac{3 \left(6 q_1 q_2^8+63 \beta _1 q_2^6-14 q_1 \beta _2 q_2^6+10 q_1 \beta _2^2 q_2^4+117 \beta _1 \beta _2 q_2^4-2 q_1 \beta
   _2^3 q_2^2+9 \beta _1 \beta _2^2 q_2^2+3 \beta _1 \beta _2^3\right)}{q_2^2 \left(q_1^3-\beta _1\right) \left(q_2^2-\beta _2\right){}^2}\\
   &\qquad+\tfrac{9
   \left(4 q_1^2 q_2^{10}+42 q_1 \beta _1 q_2^8-16 q_1^2 \beta _2 q_2^8+21 \beta _1^2 q_2^6+20 q_1^2 \beta _2^2 q_2^6-98 q_1 \beta _1 \beta _2 q_2^6-8
   q_1^2 \beta _2^3 q_2^4+70 q_1 \beta _1 \beta _2^2 q_2^4\right)}{q_2^2 \left(q_1^3-\beta _1\right){}^2 \left(q_2^2-\beta _2\right){}^2}\\
   &\qquad+\frac{9\left(39 \beta _1^2 \beta _2 q_2^4-14 q_1 \beta _1 \beta _2^3 q_2^2+3 \beta _1^2 \beta _2^2
   q_2^2+\beta _1^2 \beta _2^3\right)}{q_2^2 \left(q_1^3-\beta _1\right){}^2 \left(q_2^2-\beta _2\right){}^2}+\frac{5 \left(5 q_2^8+73 \beta _2
   q_2^6+\beta _2^4\right)}{4 \left(q_2^2-\beta _2\right){}^4}\\
   &\qquad+q_1^4 \left(\frac{225 q_2^4}{16 \left(q_2^2-\beta _2\right){}^4}+\frac{135 \beta _2
   q_2^2}{4 \left(q_2^2-\beta _2\right){}^4}+\frac{147 \beta _2^2}{8 \left(q_2^2-\beta _2\right){}^4}-\frac{9 \beta _2^3}{4 \left(q_2^2-\beta
   _2\right){}^4 q_2^2}+\frac{\beta _2^4}{16 \left(q_2^2-\beta _2\right){}^4 q_2^4}\right)\\
   &\qquad+q_1 \left(-\frac{225 \beta _1 q_2^4}{4 \left(q_2^2-\beta
   _2\right){}^4}+\frac{135 \beta _1 \beta _2 q_2^2}{\left(q_2^2-\beta _2\right){}^4}+\frac{147 \beta _1 \beta _2^2}{2 \left(q_2^2-\beta
   _2\right){}^4}-\frac{3 \left(90 \beta _1 \beta _2 q_2^2+49 \beta _1 \beta _2^2\right)}{\left(q_2^2-\beta _2\right){}^4}\right.\\
   &\qquad\left.+\frac{9 \beta _1 \beta
   _2^3}{\left(q_2^2-\beta _2\right){}^4 q_2^2}-\frac{\beta _1 \beta _2^4}{4 \left(q_2^2-\beta _2\right){}^4 q_2^4}\right)+q_1^2 \left(\frac{45
   q_2^6}{4 \left(q_2^2-\beta _2\right){}^4}-\frac{177 \beta _2^2 q_2^2}{2 \left(q_2^2-\beta _2\right){}^4}\right.\\
   &\qquad\left.\left.+\frac{3 \left(33 q_2^4 \beta _2-7 \beta
   _2^3\right)}{\left(q_2^2-\beta _2\right){}^4}-\frac{3 \beta _2^4}{4 \left(q_2^2-\beta _2\right){}^4 q_2^2}\right)\right)\end{aligned}\]
\begin{enumerate}
\item[(8)] \(\displaystyle H=\tfrac{1}{2}(p_1^2+p_2^2)+\hbar^2\left[\frac{3q_1(q_1^3+2\beta_1)}{(q_1^3-\beta_1)^2}+\frac{3q_2(q_2^3+2\beta_2)}{(q_2^3-\beta_2)^2}\right]\)
\end{enumerate}
There is one fourth-order integral
\[\begin{aligned}\displaystyle X&=\{\{m_{12}, m_{12}\},\{m_{12}, m_{12}\}+5\hbar^2\}+4\{\beta_2p_1^3-\beta_1p_2^3, m_{12}\}+\hbar^2\left\{p_1^2,\frac{12 q_1 q_2^4 \left(2 \beta _1+q_1^3\right)}{\left(q_1^3-\beta _1\right){}^2}\right.\\
&\qquad\left.-\frac{48 \beta _2 q_1 q_2 \left(2 \beta
   _1+q_1^3\right)}{\left(q_1^3-\beta _1\right){}^2}+\frac{108 \beta _2^2 q_1^2}{\left(q_2^3-\beta _2\right){}^2}+\frac{108 \beta _2
   q_1^2}{q_2^3-\beta _2}+12 q_1^2\right\}+\hbar^2\left\{p_1p_2,-\frac{12 q_2^2 q_1^3 \left(7 \beta _2+2 q_2^3\right)}{\left(q_2^3-\beta _2\right){}^2}\right.\\
   &\qquad\left.-\frac{24 q_1^2 \left(q_2^3-\beta _2\right)}{q_1^3-\beta
   _1}-\frac{108 q_1^2 \left(\beta _1 q_2^3-\beta _1 \beta _2\right)}{\left(q_1^3-\beta _1\right){}^2}+\frac{12 \beta _1 q_2^2 \left(7 \beta _2+2
   q_2^3\right)}{\left(q_2^3-\beta _2\right){}^2}\right\}+\hbar^2\left\{p_2^2,\frac{12 q_2 q_1^4 \left(2 \beta _2+q_2^3\right)}{\left(q_2^3-\beta _2\right){}^2}\right.\\
   &\qquad\left.-\frac{48 \beta _1 q_2 q_1 \left(2 \beta
   _2+q_2^3\right)}{\left(q_2^3-\beta _2\right){}^2}+\frac{108 \beta _1^2 q_2^2}{\left(q_1^3-\beta _1\right){}^2}+\frac{108 \beta _1
   q_2^2}{q_1^3-\beta _1}+12 q_2^2\right\}+\hbar^4\left(\frac{36 q_2^2 q_1^4 \left(2 \beta _2+q_2^3\right){}^2}{\left(q_2^3-\beta _2\right){}^4}\right.\\
   &\qquad+\frac{324 q_1^2 \left(\beta _1^2 q_2^4-4 \beta _1^2 \beta _2
   q_2\right)}{\left(q_1^3-\beta _1\right){}^4}+\frac{18 q_2 q_1^2 \left(7 \beta _2^2+19 \beta _2 q_2^3+q_2^6\right)}{\left(q_2^3-\beta
   _2\right){}^3}-\frac{144 \beta _1 q_2^2 q_1 \left(2 \beta _2+q_2^3\right){}^2}{\left(q_2^3-\beta _2\right){}^4}\\
   &\qquad+\frac{54 \left(4 \beta _1 q_1^2
   q_2^4+9 \beta _1^2 q_1 q_2^2-16 \beta _1 \beta _2 q_1^2 q_2\right)}{\left(q_1^3-\beta _1\right){}^3}+\frac{17 \beta _2^2+506 \beta _2 q_2^3+17
   q_2^6}{2 \left(q_2^3-\beta _2\right){}^2}\\
   &\qquad+\frac{18 \left(15 \beta _1 \beta _2^2+15 \beta _1 q_2^6-2 \beta _2 q_1 q_2^5+78 \beta _1 \beta _2
   q_2^3+\beta _2^2 q_1 q_2^2+q_1 q_2^8\right)}{\left(q_1^3-\beta _1\right) \left(q_2^3-\beta _2\right){}^2}\\
   &\qquad\left.+\tfrac{18 \left(15 \beta _1^2 \beta
   _2^2+21 \beta _1 q_1 q_2^8-12 \beta _2 q_1^2 q_2^7+15 \beta _1^2 q_2^6-42 \beta _1 \beta _2 q_1 q_2^5+18 \beta _2^2 q_1^2 q_2^4+78 \beta _1^2 \beta
   _2 q_2^3+21 \beta _1 \beta _2^2 q_1 q_2^2-8 \beta _2^3 q_1^2 q_2+2 q_1^2 q_2^{10}\right)}{\left(q_1^3-\beta _1\right){}^2 \left(q_2^3-\beta
   _2\right){}^2}\right)\end{aligned}\]
\begin{enumerate}
    \item[(9)] \(H=\tfrac{1}{2}(p_1^2+p_2^2)+\beta(q_1^2+9q_2^2)\)
    \end{enumerate}
There is one third-order integral
\[X=\{p_1^2,m_{12}\}+6\beta q_2\{p_1,q_1^2\}-\tfrac{2}{3}\beta q_1^3p_2\]
and one fourth-order integral \((H_1-H_2,X)\).
\begin{enumerate}
    \item[(10)] \(H=\tfrac{1}{2}(p_1^2+p_2^2)+\beta(q_1^2+9q_2^2)+\dfrac{\hbar^2}{q_1^2}\)
    \end{enumerate}
There is one third-order integral
\[X=\{p_1^2,m_{12}\}+3q_2\left\{p_1,2\beta q_1^2-\frac{\hbar^2}{q_1^2}\right\}+\left(\frac{\hbar^2}{q_1}-\frac{2\beta q_1^3}{3}\right)p_2\]
and one fourth-order integral \((H_1-H_2,X)\).
\begin{enumerate}
    \item[(11)] \(\displaystyle H=\tfrac{1}{2}(p_1^2+p_2^2)+\hbar^2\left[\frac{q_1^2+9q_2^2}{8\beta^2}+\frac{2(q_1^2+\beta)}{(q_1^2-\beta)^2}\right]\)
    \end{enumerate}
There is one third-order integral
\[X=\{p_1^2,m_{12}\}+3\hbar^2q_2\left\{p_1,\frac{q_1^2}{4\beta^2}-\frac{2(q_1^2+\beta)}{(q_1^2-\beta)^2}\right\}+\hbar^2q_1p_2\left[\frac{2(q_1^2+3\beta)}{(q_1^2-\beta)^2}-\frac{q_1^2}{12\beta^2}\right]\]
and one fourth-order integral \((H_1-H_2,X)\).
\begin{enumerate}
\item[(12)] \(\displaystyle H=\tfrac{1}{2}(p_1^2+p_2^2)+\hbar^2\left[\frac{2(q_1^2+\beta_1)}{(q_1^2-\beta_1)^2}-\frac{1}{8q_1^2}\right]+\beta_2q_2\)
\end{enumerate}
There are two fourth-order integrals
\[\begin{aligned}
    X^{(1)}&=\{p_1^2,\{m_{12},m_{12}\}\}+\left\{p_1^2,\frac{16 \beta _1 q_2^2 \hbar ^2}{\left(q_1^2-\beta _1\right){}^2}+\frac{1}{2} \left(2 \beta _1 \beta _2 q_2+\hbar ^2\right)+\frac{8 q_2^2 \hbar ^2}{q_1^2-\beta _1}\right.\\
    &\qquad\left.+\beta _2
   \left(-q_2\right) q_1^2-\frac{q_2^2 \hbar ^2}{2 q_1^2}\right\}+\left\{p_1p_2,-\frac{24 \beta _1 q_2 q_1 \hbar ^2}{\left(q_1^2-\beta _1\right){}^2}-\frac{8 q_2 q_1 \hbar ^2}{q_1^2-\beta _1}+\beta _2 q_1^3\right.\\
   &\qquad\left.-\beta _1 \beta _2 q_1+\frac{q_2 \hbar ^2}{2 q_1}\right\}+\frac{\hbar ^2p_2^2 \left(\beta _1^2+14 \beta _1 q_1^2+q_1^4\right)}{2 \left(q_1^2-\beta _1\right){}^2}+\frac{64 \beta _1^2 q_2^2 \hbar ^4}{\left(q_1^2-\beta _1\right){}^4}\\
   &\qquad+\frac{2 \left(11 \beta _1 \hbar ^4+6 q_2^2 \hbar ^4\right)}{\left(q_1^2-\beta _1\right){}^2}+\frac{8 \left(3 \beta
   _1^2 \hbar ^4+8 \beta _1 q_2^2 \hbar ^4\right)}{\left(q_1^2-\beta _1\right){}^3}-\frac{17}{4} \beta _2 q_2 \hbar ^2\\
   &\qquad+\frac{-\beta _1 \hbar ^4-2 \beta _1^2 \beta _2 q_2 \hbar ^2-16
   q_2^2 \hbar ^4}{8 \beta _1 q_1^2}-\frac{2 \left(-\beta _1 \hbar ^4+2 \beta _1^2 \beta _2 q_2 \hbar ^2-q_2^2 \hbar ^4\right)}{\beta _1 \left(q_1^2-\beta _1\right)}\\
   &\qquad+\frac{1}{4} \beta
   _2^2 q_1^4-\frac{1}{2} \beta _1 \beta _2^2 q_1^2+\frac{q_2^2 \hbar ^4}{16 q_1^4}
\end{aligned}\]
\[\begin{aligned}
    X^{(2)}&=\{p_1^3,m_{12}\}+\left\{p_1^2,-\frac{16 \beta _1 q_2 \hbar ^2}{\left(q_1^2-\beta _1\right){}^2}-\frac{8 q_2 \hbar ^2}{q_1^2-\beta _1}+\frac{1}{2} \beta _2 q_1^2+\frac{q_2 \hbar ^2}{2 q_1^2}\right\}\\
    &\qquad+\hbar^2\left\{p_1p_2,\frac{-\beta _1^2+34 \beta _1 q_1^2+15 q_1^4}{4 q_1 \left(q_1^2-\beta _1\right){}^2}\right\}+\frac{\beta _2 \hbar ^2}{4}-\frac{64 \beta _1^2 q_2 \hbar ^4}{\left(q_1^2-\beta _1\right){}^4}\\
    &\qquad-\frac{64 \beta _1 q_2 \hbar ^4}{\left(q_1^2-\beta _1\right){}^3}+\frac{2 q_2 \hbar
   ^4}{\beta _1 q_1^2}-\frac{2 \left(q_2 \hbar ^4-2 \beta _1^2 \beta _2 \hbar ^2\right)}{\beta _1 \left(q_1^2-\beta _1\right)}-\frac{4 \left(3 q_2 \hbar ^4-\beta _1^2 \beta _2 \hbar
   ^2\right)}{\left(q_1^2-\beta _1\right){}^2}-\frac{q_2 \hbar ^4}{16 q_1^4}
\end{aligned}\]
\begin{enumerate}
\item[(13)] \(\displaystyle H=\tfrac{1}{2}(p_1^2+p_2^2)+\frac{3\hbar^2q_1(q_1^3+2\beta_1)}{(q_1^3-\beta_1)^2}+\beta_2q_2\)
\end{enumerate}
There are three fourth-order integrals
\[\begin{aligned}
    X^{(1)}&=\{p_1^2,\{m_{12},m_{12}\}\}+\left\{p_1^2,\frac{36 \beta _1 q_1 q_2^2 \hbar ^2}{\left(q_1^3-\beta _1\right){}^2}+\frac{12 q_1 q_2^2 \hbar ^2}{q_1^3-\beta _1}-\beta _2 q_1^2 q_2+\frac{\hbar ^2}{2}\right\}\\
    &\qquad+\left\{p_1p_2,-\beta _1 \beta _2-\frac{54 \beta _1 q_2 q_1^2 \hbar ^2}{\left(q_1^3-\beta _1\right){}^2}-\frac{12 q_2 q_1^2 \hbar ^2}{q_1^3-\beta _1}+\beta _2 q_1^3\right\}\\
    &\qquad+\frac{p_2^2 \hbar ^2 \left(\beta _1^2+34 \beta _1 q_1^3+q_1^6\right)}{2 \left(q_1^3-\beta _1\right){}^2}+\frac{324 \beta _1^2 q_1^2 q_2^2 \hbar ^4}{\left(q_1^3-\beta _1\right){}^4}+\frac{27 \left(3 \beta _1^2 q_1 \hbar ^4+8 \beta _1 q_1^2 q_2^2 \hbar ^4\right)}{\left(q_1^3-\beta
   _1\right){}^3}\\
   &\qquad-\frac{13}{2} \beta _2 q_2 \hbar ^2+\frac{3 \left(q_1 \hbar ^4-6 \beta _1 \beta _2 q_2 \hbar ^2\right)}{q_1^3-\beta _1}+\frac{9 \left(7 \beta _1 q_1 \hbar ^4-2 \beta
   _1^2 \beta _2 q_2 \hbar ^2+4 q_1^2 q_2^2 \hbar ^4\right)}{\left(q_1^3-\beta _1\right){}^2}\\
   &\qquad+\frac{1}{4} \beta _2^2 q_1^4-\beta _1 \beta _2^2 q_1
\end{aligned}\]
\[\begin{aligned}
    X^{(2)}&=\{p_1^3,m_{12}\}+\left\{p_1^2,-\frac{36 \beta _1 q_1 q_2 \hbar ^2}{\left(q_1^3-\beta _1\right){}^2}-\frac{12 q_1 q_2 \hbar ^2}{q_1^3-\beta _1}+\frac{1}{2} \beta _2 q_1^2\right\}+\hbar^2\left\{p_1p_2,\frac{3 q_1^2 \left(7 \beta _1+2 q_1^3\right)}{\left(q_1^3-\beta _1\right){}^2}\right\}\\
    &\qquad+\frac{\beta _2 \hbar ^2}{4}-\frac{324 \beta _1^2 q_1^2 q_2 \hbar ^4}{\left(q_1^3-\beta _1\right){}^4}-\frac{216 \beta _1 q_1^2 q_2 \hbar ^4}{\left(q_1^3-\beta _1\right){}^3}+\frac{9
   \beta _1 \beta _2 \hbar ^2}{q_1^3-\beta _1}-\frac{9 \left(4 q_1^2 q_2 \hbar ^4-\beta _1^2 \beta _2 \hbar ^2\right)}{\left(q_1^3-\beta _1\right){}^2}
\end{aligned}\]
\[\begin{aligned}
    X^{(3)}&=p_1^3p_2+\left\{p_1^2,\frac{q_2 \hbar ^2}{\beta _1}+\beta _2 q_1\right\}+\hbar^2\left\{p_1p_2,\frac{q_1 \left(17 \beta _1^2+11 \beta _1 q_1^3-q_1^6\right)}{\beta _1 \left(q_1^3-\beta _1\right){}^2}\right\}-\frac{\beta _2 q_1^2 \hbar ^2}{2 \beta _1}\\
    &\qquad+\frac{18 \left(\beta _1 \beta _2 q_1^2 \hbar ^2+q_1 q_2 \hbar ^4\right)}{\left(q_1^3-\beta _1\right){}^2}+\frac{3 \left(\beta _1 \beta _2
   q_1^2 \hbar ^2+2 q_1 q_2 \hbar ^4\right)}{\beta _1 \left(q_1^3-\beta _1\right)}
\end{aligned}\]
\begin{enumerate}
\item[(14)] \(H=\tfrac{1}{2}(p_1^2+p_2^2)+\beta_1(q_1^2+4q_2^2)+\displaystyle\frac{\beta_2}{q_1^2}+\frac{\beta_3}{q_2^2}\)
\end{enumerate}
There is one fourth-order integral
\[\begin{aligned}X&=\{p_1^2,m_{12}\}+\left\{p_1^2,\frac{4 \beta _2 q_2^2}{q_1^2}+q_1^2 \left(\frac{2 \beta _3}{q_2^2}-4 \beta _1 q_2^2\right)+\frac{\hbar ^2}{2}\right\}+\left\{p_1p_2,\frac{4 q_2 \left(\beta _1 q_1^4-\beta _2\right)}{q_1}\right\}\\
&\qquad+\frac{\hbar^2p_2^2}{2}+3 \beta _1 q_1^2 \hbar ^2-2 \beta _1 q_2^2 \left(4 \beta _2+\hbar ^2\right)+\frac{\beta _3 \hbar ^2}{q_2^2}+\frac{\beta _2 \hbar ^2}{q_1^2}+4 \beta _1^2 q_2^2 q_1^4+\frac{4 \beta _2^2
   q_2^2}{q_1^4}\end{aligned}\]
\begin{enumerate}
\item[(14a)] \(H=\tfrac{1}{2}(p_1^2+p_2^2)+\beta_1(q_1^2+4q_2^2)+\displaystyle\frac{\beta_2}{q_1^2}+\frac{\hbar^2}{q_2^2}\)
\end{enumerate}
This is a special case of (14) where \(\beta_3=\hbar^2\). There is another fourth-order integral
\[\begin{aligned}
    X^{(1)}&=\{p_1p_2^2,m_{12}\}-\left(8\beta_1q_2^3+\frac{\hbar^2}{q_2}\right)p_1^2+\left\{p_1p_2,q_2\left(8\beta_1q_2^2+\frac{3\hbar^2}{q_2^2}\right)\right\}+2\left(\beta_1q_1^2-\frac{\beta_2}{q_1^2}\right)\{p_2^2,q_2\}\\
    &\qquad+4 \beta _1 q_2 \hbar ^2+\frac{2 \beta _1 q_1^2 \left(8 \beta _1 q_2^4+\hbar ^2\right)}{q_2}-\frac{2 \beta _2 \left(8 \beta _1 q_2^4+\hbar ^2\right)}{q_1^2 q_2}-\frac{3 \hbar ^4}{2
   q_2^3}
\end{aligned}\]
\begin{enumerate}
\item[(15)] \(\displaystyle H=\tfrac{1}{2}(p_1^2+p_2^2)+\beta_1(q_1^2+16q_2^2)+\frac{\beta_2}{q_1^2}\)
\end{enumerate}
There is one fourth-order integral
\[\begin{aligned}
    X&=\{p_1^3,m_{12}\}+4q_2\left\{p_1^2,3\beta_1q_1^2-\frac{\beta_2}{q_1^2}\right\}+2p_2\left\{p_1,\frac{\beta_2}{q_1}-\beta_1q_1^3\right\}\\
    &\qquad+2 \beta _1 q_2 \left(4 \beta _2+3 \hbar ^2\right)-4 \beta _1^2 q_2 q_1^4-\frac{4 \beta _2^2 q_2}{q_1^4}
\end{aligned}\]
\begin{enumerate}
\item[(16)] \(\displaystyle H=\tfrac{1}{2}(p_1^2+p_2^2)+\beta_1(9q_1^2+4q_2^2)+\frac{\beta_2}{q_1^2}\)
\end{enumerate}
There is one fourth-order integral
\[\begin{aligned}
    X&=\{p_1p_2^2,m_{12}\}+\frac{8}{3}\beta_1q_2^3p_1^2-24\beta_1\{p_1p_2,q_1q_2^2\}+2\left(9\beta_1q_1^2-\frac{\beta_2}{q_1^2}\right)\{p_2^2,q_2\}\\
    &\qquad-12 \beta _1 q_2 \hbar ^2-48 \beta _1^2 q_1^2 q_2^3+\frac{16 \beta _1 \beta _2 q_2^3}{3 q_1^2}
\end{aligned}\]
\begin{enumerate}
\item[(17)] \(\displaystyle H=\tfrac{1}{2}(p_1^2+p_2^2)+\beta_1(9q_1^2+4q_2^2)+\frac{\beta_2}{q_1^2}+\frac{\hbar^2}{q_2^2}\)
\end{enumerate}
There is one fourth-order integral
\[\begin{aligned}
    X&=\{p_1p_2^2,m_{12}\}+\left(\frac{8}{3}\beta_1q_2^3-\frac{\hbar^2}{q_2}\right)p_1^2+\left\{p_1p_2,\frac{3\hbar^2q_1}{q_2^2}-24\beta_1q_1q_2^2\right\}\\
    &\qquad+2\left(9\beta_1q_1^2-\frac{\beta_2}{q_1^2}\right)\{p_2^2,q_2\}-\frac{3 \hbar ^2 \left(8 \beta _1 q_2^4+\hbar ^2\right)}{2 q_2^3}-\frac{6 \beta _1 q_1^2 \left(8 \beta _1 q_2^4-3 \hbar ^2\right)}{q_2}\\
    &\qquad+\frac{2 \beta _2 \left(8 \beta _1 q_2^4-3
   \hbar ^2\right)}{3 q_1^2 q_2}
\end{aligned}\]
\begin{enumerate}
\item[(18)] \(\displaystyle H=\tfrac{1}{2}(p_1^2+p_2^2)+\beta_1\left(q_1^2+4q_2^2+\frac{\beta_2^2}{q_1^2}\right)+\frac{\beta_3}{q_2^2}+\hbar^2\left[\frac{2(q_1^2+\beta_2)}{(q_1^2-\beta_2)^2}-\frac{1}{8q_1^2}\right]\) (1)
\end{enumerate}
There is one fourth-order integral
\[\begin{aligned}
    X&=\{p_1^2,\{m_{12},m_{12}\}\}+\left\{p_1^2,\frac{16 \beta _2 q_2^2 \hbar ^2}{\left(q_1^2-\beta _2\right){}^2}+\frac{8 q_2^2 \hbar ^2}{q_1^2-\beta _2}+\frac{8 \beta _1 \beta _2^2 q_2^2-q_2^2 \hbar ^2}{2 q_1^2}\right.\\
    &\qquad\left.-\frac{2 q_1^2
   \left(2 \beta _1 q_2^4-\beta _3\right)}{q_2^2}+\frac{\hbar ^2}{2}\right\}+\left\{p_1p_2,-\frac{24 \beta _2 q_2 q_1 \hbar ^2}{\left(q_1^2-\beta _2\right){}^2}-\frac{8 q_2 q_1 \hbar ^2}{q_1^2-\beta _2}\right.\\
   &\qquad\left.+\frac{q_2 \hbar ^2-8 \beta _1 \beta _2^2 q_2}{2 q_1}+4 \beta _1 q_2
   q_1^3\right\}+\frac{p_2^2 \hbar ^2 \left(\beta _2^2+14 \beta _2 q_1^2+q_1^4\right)}{2 \left(q_1^2-\beta _2\right){}^2}+\frac{64 \beta _2^2 q_2^2 \hbar ^4}{\left(q_1^2-\beta _2\right){}^4}\\
   &\qquad+\frac{8 \left(3 \beta _2^2 \hbar ^4+8 \beta _2 q_2^2 \hbar ^4\right)}{\left(q_1^2-\beta _2\right){}^3}+3 \beta _1
   q_1^2 \hbar ^2-\frac{-\beta _3 \hbar ^2+17 \beta _1 q_2^4 \hbar ^2+8 \beta _1^2 \beta _2^2 q_2^4}{q_2^2}\\
   &\qquad+\frac{-\beta _2 \hbar ^4+8 \beta _1 \beta _2^3 \hbar ^2+128 \beta _1 \beta
   _2^2 q_2^2 \hbar ^2-16 q_2^2 \hbar ^4}{8 \beta _2 q_1^2}+\frac{-16 \beta _1 \beta _2^2 q_2^2 \hbar ^2+64 \beta _1^2 \beta _2^4 q_2^2+q_2^2 \hbar ^4}{16 q_1^4}\\
   &\qquad+\frac{2 \left(8 \beta
   _2^2 \beta _3 \hbar ^2+11 \beta _2 q_2^2 \hbar ^4+6 q_2^4 \hbar ^4\right)}{q_2^2 \left(q_1^2-\beta _2\right){}^2}+\frac{2 \left(8 \beta _2^2 \beta _3 \hbar ^2+\beta _2 q_2^2 \hbar
   ^4-16 \beta _1 \beta _2^2 q_2^4 \hbar ^2+q_2^4 \hbar ^4\right)}{\beta _2 q_2^2 \left(q_1^2-\beta _2\right)}\\
   &\qquad+4 \beta _1^2 q_1^4 q_2^2
\end{aligned}\]
\begin{enumerate}
\item[(18a)] \(\displaystyle H=\tfrac{1}{2}(p_1^2+p_2^2)+\beta_1\left(q_1^2+4q_2^2+\frac{\beta_2^2}{q_1^2}\right)+\hbar^2\left[\frac{2(q_1^2+\beta_2)}{(q_1^2-\beta_2)^2}-\frac{1}{8q_1^2}\right]\)
\end{enumerate}
This is a special case of (18) where \(\beta_3=0\). There is one more fourth-order integral
\[\begin{aligned}
    X&=\{p_1^3,m_{12}\}+q_2\left\{p_1^2,4 \beta _1 \beta _2-\frac{16 \beta _2 \hbar ^2}{\left(q_1^2-\beta _2\right){}^2}-\frac{8 \hbar ^2}{q_1^2-\beta _2}+\frac{\hbar ^2-8 \beta _1 \beta _2^2}{2 q_1^2}\right\}\\
    &\qquad+\left\{p_1p_2,\frac{12 \beta _2 q_1 \hbar ^2}{\left(q_1^2-\beta _2\right){}^2}+\frac{4 q_1 \hbar ^2}{q_1^2-\beta _2}+\frac{8 \beta _1 \beta _2^2-\hbar ^2}{4 q_1}+2 \beta _1 q_1^3-4 \beta _1 \beta
   _2 q_1\right\}\\
   &\qquad-\frac{64 \beta _2^2 q_2 \hbar ^4}{\left(q_1^2-\beta _2\right){}^4}-\frac{64 \beta _2 q_2 \hbar ^4}{\left(q_1^2-\beta _2\right){}^3}-\frac{12 q_2 \hbar ^4}{\left(q_1^2-\beta
   _2\right){}^2}-\frac{2 \left(q_2 \hbar ^4-8 \beta _1 \beta _2^2 q_2 \hbar ^2\right)}{\beta _2 \left(q_1^2-\beta _2\right)}\\
   &\qquad+\frac{16 \beta _1 \beta _2^2 q_2 \hbar ^2-64 \beta _1^2
   \beta _2^4 q_2-q_2 \hbar ^4}{16 q_1^4}+\frac{-17 \beta _1 \beta _2^2 q_2 \hbar ^2+8 \beta _1^2 \beta _2^4 q_2+2 q_2 \hbar ^4}{\beta _2 q_1^2}\\
   &\qquad+4 \beta _1^2 q_1^4 q_2-8 \beta _1^2
   \beta _2 q_1^2 q_2
\end{aligned}\]
\begin{enumerate}
\item[(19)] \(\displaystyle H=\tfrac{1}{2}(p_1^2+p_2^2)+\beta_1\left(q_1^2+16q_2^2+\frac{\beta_2^2}{q_1^2}\right)+\hbar^2\left[\frac{2(q_1^2+\beta_2)}{(q_1^2-\beta_2)^2}-\frac{1}{8q_1^2}\right]\)
\end{enumerate}
There is one fourth-order integral
\[\begin{aligned}
    X&=\{p_1^3,m_{12}\}+q_2\left\{p_1^2,-\frac{16 \beta _2 \hbar ^2}{\left(q_1^2-\beta _2\right){}^2}-\frac{8 \hbar ^2}{q_1^2-\beta _2}+\frac{\hbar ^2-8 \beta _1 \beta _2^2}{2 q_1^2}+12 \beta _1 q_1^2\right\}\\
    &\qquad+\left\{p_1p_2,\frac{12 \beta _2 q_1 \hbar ^2}{\left(q_1^2-\beta _2\right){}^2}+\frac{4 q_1 \hbar ^2}{q_1^2-\beta _2}+\frac{8 \beta _1 \beta _2^2-\hbar ^2}{4 q_1}-2 \beta _1 q_1^3\right\}-\frac{64 \beta _2^2 q_2 \hbar ^4}{\left(q_1^2-\beta _2\right){}^4}\\
    &\qquad-\frac{64 \beta _2 q_2 \hbar ^4}{\left(q_1^2-\beta _2\right){}^3}+\beta _1 q_2 \left(8 \beta _1 \beta _2^2+21 \hbar
   ^2\right)-\frac{2 \left(q_2 \hbar ^4-48 \beta _1 \beta _2^2 q_2 \hbar ^2\right)}{\beta _2 \left(q_1^2-\beta _2\right)}\\
   &\qquad-\frac{4 \left(3 q_2 \hbar ^4-16 \beta _1 \beta _2^2 q_2 \hbar
   ^2\right)}{\left(q_1^2-\beta _2\right){}^2}-\frac{2 \left(8 \beta _1 \beta _2^2 q_2 \hbar ^2-q_2 \hbar ^4\right)}{\beta _2 q_1^2}\\
   &\qquad+\frac{16 \beta _1 \beta _2^2 q_2 \hbar ^2-64 \beta
   _1^2 \beta _2^4 q_2-q_2 \hbar ^4}{16 q_1^4}-4 \beta _1^2 q_1^4 q_2
\end{aligned}\]
\begin{enumerate}
\item[(20)] \(\displaystyle H=\tfrac{1}{2}(p_1^2+p_2^2)+\hbar^2\left[\frac{4q_1^2+q_2^2}{32\beta_1^2}+\frac{2(q_1^2+\beta_1)}{(q_1^2-\beta_1)^2}\right]+\frac{\beta_2}{q_2^2}\)
\end{enumerate}
There is one fourth-order integral
\[\begin{aligned}
    X&=\{p_1^2p_2,m_{12}\}+\left(\frac{2 \beta _2}{q_2^2}-\frac{q_2^2 \hbar ^2}{16 \beta _1^2}\right)\left\{p_1^2,q_1 \right\}-\hbar^2\left\{p_1p_2,\frac{q_2 \left(20 \beta _1^3-6 \beta _1 q_1^4+33 \beta _1^2 q_1^2+q_1^6\right)}{4 \beta _1^2 \left(q_1^2-\beta _1\right){}^2}\right\}\\
    &\qquad+\frac{p_2^2 q_1 \hbar ^2 \left(20 \beta _1^3-6 \beta _1 q_1^4+17 \beta _1^2 q_1^2+q_1^6\right)}{4 \beta _1^2 \left(q_1^2-\beta _1\right){}^2}+\frac{12 \beta _1 q_1 \hbar ^4}{\left(q_1^2-\beta _1\right){}^3}-\frac{q_1^3 \hbar ^2 \left(q_2^4 \hbar ^2-32 \beta _1^2 \beta _2\right)}{64 \beta _1^4 q_2^2}\\
    &\qquad+\frac{q_1 \hbar ^2
   \left(-32 \beta _1^2 \beta _2-2 \beta _1 q_2^2 \hbar ^2+q_2^4 \hbar ^2\right)}{16 \beta _1^3 q_2^2}+\frac{32 \beta _1^2 \beta _2 q_1 \hbar ^2-q_1 q_2^4 \hbar ^4}{8 \beta _1^2 q_2^2
   \left(q_1^2-\beta _1\right)}\\
   &\qquad+\frac{6 \beta _1 q_1 q_2^2 \hbar ^4+32 \beta _1^2 \beta _2 q_1 \hbar ^2-q_1 q_2^4 \hbar ^4}{2 \beta _1 q_2^2 \left(q_1^2-\beta _1\right){}^2}
\end{aligned}\]
\begin{enumerate}
\item[(21)] \(\displaystyle H=\tfrac{1}{2}(p_1^2+p_2^2)+\hbar^2\left[\frac{9(4q_1^2+q_2^2)}{32\beta_1^2}+\frac{2(q_1^2+\beta_1)}{(q_1^2-\beta_1)^2}+\frac{1}{q_1^2}\right]+\frac{\beta_2}{q_2^2}\)
\end{enumerate}
There is one fourth-order integral
\[\begin{aligned}
    X&=\{p_1^2p_2,m_{12}\}+\left(\frac{2\beta_2}{q_2^2}-\frac{9\hbar^2q_2^2}{16\beta_1^2}\right)\{p_1^2,q_1\}\\
    &\qquad-\hbar^2\left\{p_1p_2,\frac{3 q_2 \left(4 \beta _1^4-14 \beta _1 q_1^6+31 \beta _1^2 q_1^4-8 \beta _1^3 q_1^2+3 q_1^8\right)}{4 \beta _1^2 q_1^2 \left(q_1^2-\beta _1\right){}^2}\right\}\\
    &\qquad+\frac{p_2^2 \hbar ^2 \left(4 \beta _1^4-42 \beta _1 q_1^6+69 \beta _1^2 q_1^4-8 \beta _1^3 q_1^2+9 q_1^8\right)}{4 \beta _1^2 q_1 \left(q_1^2-\beta _1\right){}^2}+\frac{3 \hbar ^4}{2 q_1^3}\\
    &\qquad+\frac{12 \beta _1 q_1 \hbar ^4}{\left(q_1^2-\beta _1\right){}^3}-\frac{9 q_1^3 \hbar ^2 \left(9 q_2^4 \hbar ^2-32 \beta _1^2 \beta _2\right)}{64 \beta _1^4 q_2^2}+\frac{3 q_1 \hbar ^2
   \left(-32 \beta _1^2 \beta _2-3 \beta _1 q_2^2 \hbar ^2+9 q_2^4 \hbar ^2\right)}{8 \beta _1^3 q_2^2}\\
   &\qquad+\frac{32 \beta _1^2 \beta _2 \hbar ^2-9 q_2^4 \hbar ^4}{16 \beta _1^2 q_1
   q_2^2}+\frac{32 \beta _1^2 \beta _2 q_1 \hbar ^2-9 q_1 q_2^4 \hbar ^4}{8 \beta _1^2 q_2^2 \left(q_1^2-\beta _1\right)}+\frac{6 \beta _1 q_1 q_2^2 \hbar ^4+32 \beta _1^2 \beta _2
   q_1 \hbar ^2-9 q_1 q_2^4 \hbar ^4}{2 \beta _1 q_2^2 \left(q_1^2-\beta _1\right){}^2}
\end{aligned}\]
\begin{enumerate}
\item[(22)] \(\displaystyle H=\tfrac{1}{2}(p_1^2+p_2^2)+\hbar^2\left[\frac{4q_1^2+9q_2^2}{32\beta_1^2}+\frac{2(q_1^2+\beta_1)}{(q_1^2-\beta_1)^2}\right]+\frac{\beta_2}{q_2^2}\)
\end{enumerate}
There is one fourth-order integral
\[\begin{aligned}
    X&=\{p_1^2p_2,m_{12}\}+\left(\frac{2\beta_2}{q_2^2}-\frac{9\hbar^2q_2^2}{16\beta_1^2}\right)\{p_1^2,q_1\}+\hbar^2\left\{p_1p_2,\frac{3 q_2 \left(-8 \beta _1^3-2 \beta _1 q_1^4-7 \beta _1^2 q_1^2+q_1^6\right)}{4 \beta _1^2 \left(q_1^2-\beta _1\right){}^2}\right\}\\
    &\qquad+\frac{p_2^2 q_1 \hbar ^2 \left(72 \beta _1^3+2 \beta _1 q_1^4+23 \beta _1^2 q_1^2-q_1^6\right)}{12 \beta _1^2 \left(q_1^2-\beta _1\right){}^2}+\frac{12 \beta _1 q_1 \hbar ^4}{\left(q_1^2-\beta _1\right){}^3}+\frac{3 q_1 \hbar ^4}{8 \beta _1^2}+\frac{q_1^3 \hbar ^2 \left(9 q_2^4 \hbar ^2-32 \beta _1^2 \beta _2\right)}{192
   \beta _1^4 q_2^2}\\
   &\qquad-\frac{q_1 \left(-32 \beta _1^2 \beta _2 \hbar ^2-6 \beta _1 q_2^2 \hbar ^4+9 q_2^4 \hbar ^4\right)}{2 \beta _1 q_2^2 \left(q_1^2-\beta _1\right){}^2}+\frac{32
   \beta _1^2 \beta _2 q_1 \hbar ^2-9 q_1 q_2^4 \hbar ^4}{8 \beta _1^2 q_2^2 \left(q_1^2-\beta _1\right)}
\end{aligned}\]
\end{document}